\documentclass[a4paper,11pt]{article}

\usepackage{fullpage,latexsym,amsthm,amsmath,color,amssymb,url,hyperref}
\usepackage{tikz}
\usepackage[ruled,linesnumbered]{algorithm2e}
\usetikzlibrary{arrows,decorations.pathreplacing,shapes}

\usepackage{color}

\input{colordvi}

\newcommand{\mynewtheorem}[2]{
  \newaliascnt{#1}{dummy}
  \newtheorem{#1}[#1]{#2}
  \aliascntresetthe{#1}
  \expandafter\def\csname #1autorefname\endcsname{#2}
}

\definecolor{Red}{rgb}{1, 0 ,0}
\definecolor{Blue}{rgb}{0, 0 ,1}

\newcommand{\blue}[1]{{\color{Blue}{#1}}}

\newtheorem{lemma}{Lemma}
\newtheorem{claim}{Claim}

\newtheorem{observation}{Observation}
\newtheorem{corollary}{Corollary}
\newtheorem{definition}{Definition}
\newtheorem{theorem}{Theorem}

\newcommand{\T}{\mathbf{T}}
\renewcommand{\L}{\mathbf{L}}
\newcommand{\lca}{\mathsf{lca}}

\newcommand{\infL}{\prec_{\L}}

\renewcommand{\H}{\mathsf{H}}
\newcommand{\PQ}{\textsf{PQ}}
\newcommand{\FPQ}{\textsf{FPQ}}
\newcommand{\perm}{\mathsf{permute}}
\newcommand{\rev}{\mathsf{reverse}}

\title{\textbf{Tree-layout based graph classes: \\ proper chordal graphs}\footnote{Research supported by the project DEMOGRAPH (ANR-16-CE40-0028) and the French-German Collaboration ANR/DFG Project UTMA (ANR-20-CE92-0027). An extended abstract of this paper appeared in:  \emph{International Symposium on Theoretical Aspects of Computer Science} (STACS). Leibniz International Proceedings in Informatics (LIPIcs), Volume 289, pp. 55:1-55:18, 2024.
}}

\author{Christophe Paul\thanks{LIRMM, Univ Montpellier, CNRS, Montpellier, France.} \and Evangelos Protopapas \thanks{LIRMM, Univ Montpellier, CNRS, Montpellier, France.}}

\date{\today}

\begin{document}

\maketitle

\begin{abstract}
Many standard graph classes are known to be characterized by means of layouts (a vertex ordering) excluding some patterns. Important such graph classes are  among others: proper interval graphs, interval graphs, chordal graphs, permutation graphs, (co-)comparability graphs. For example, a graph $G=(V,E)$ is a proper interval graph if and only if $G$ has a layout $\L_G$ such that for every triple of vertices such that $x\prec_{\L_G} y\prec_{\L_G} z$, if $xz\in E$, then $xy\in E$ and $yz\in E$. Such a triple $x$, $y$, $z$ is called an \emph{indifference triple} and layouts excluding indifference triples are known as \emph{indifference layouts}. 

In this paper, we investigate the concept of \emph{tree-layouts}. A tree-layout $\T_G=(T,r,\rho_G)$ of a graph $G=(V,E)$ is a tree $T$ rooted at some node $r$ and equipped with a one-to-one mapping $\rho_G$ between $V$ and the nodes of $T$ such that for every edge $xy\in E$, either $x$ is an ancestor of $y$, denoted $x\prec_{\T_G} y$, or $y$ is an ancestor of $x$. Clearly, layouts are tree-layouts. 

Excluding a pattern in a tree-layout is defined similarly as excluding a pattern in a layout, but now using the ancestor relation. Unexplored graph classes can be defined by means of tree-layouts excluding some patterns. As a proof of concept, we show that excluding non-indifference triples in tree-layouts yields a natural notion of \emph{proper chordal graphs}. We characterize proper chordal graphs and position them in the hierarchy of known subclasses of chordal graphs. We also provide a canonical representation of proper chordal graphs that encodes all the indifference tree-layouts rooted at some vertex. Based on this result, we first design a polynomial time recognition algorithm for proper chordal graphs. We then show that the problem of testing isomorphism between two proper chordal graphs is in $\mathsf{P}$, whereas this problem is known to be $\mathsf{GI}$-complete on chordal graphs.

\end{abstract}

\newpage



\section{Introduction}
\label{sec_intro}

\paragraph{Context.} A graph class $\mathcal{C}$ is hereditary if for every graph $G\in \mathcal{C}$ and every induced subgraph $H$ of $G$, which we denote $H\subseteq_{\sf i} G$, we have that $H\in \mathcal{C}$. A \emph{minimal forbidden subgraph} for $\mathcal{C}$ is a graph $F\notin \mathcal{C}$ such that for every induced subgraph $H\subseteq_{\sf i} F$, $H\in\mathcal{C}$. Clearly, a hereditary graph class $\mathcal{C}$ is characterized by its set of minimal  forbidden subgraphs.
Let $\mathcal{F}$ be a set of graphs that are pairwise not induced subgraphs of one another. We say that a graph $G$ is an \emph{$\mathcal{F}$-free graph}, if it does not contain any graph of $\mathcal{F}$ as an induced subgraph. 
Many graph classes are characterized by a finite set $\mathcal{F}$ of minimal forbidden subgraphs. If $\mathcal{F}=\{H\}$, then we simply say that $G$ is $H$-free if $H\not\subseteq_{\sf i} G$.  A popular example of such a class is the set of cographs~\cite{Lerchs71,Sumner73}. A graph $G$ is a \emph{cograph} if either $G$ is the single vertex graph, or it is the disjoint union of two cographs, or its complement is a cograph. It is well known that $G$ is a cograph if and only if it is a $P_4$-free graph~\cite{Lerchs71,CorneilLS81}. Unfortunately, the set of minimal obstructions of a hereditary graph family may not be finite. This is the case for \emph{chordal graphs}~\cite{HajnalS58,Berge61} which are defined as the graphs that do not contain a chordless cycle of length at least $4$ as an induced subgraph.

An interesting approach, initiated by Skrien~\cite{Skrien82} and Damaschke~\cite{Damaschke90}, to circumvent this issue, is to embed graphs in some additional structure such as vertex orderings, also called \emph{layouts}. An \emph{ordered graph} is then defined as a pair $(G,\prec_G)$ such that $\prec_G$ is a total ordering of the vertex set $V$ of the graph $G=(V,E)$. We say that an ordered graph $(H,\prec_H)$ is a \emph{pattern} of the ordered graph $(G,\prec_G)$, which we denote by $(H,\prec_H)\subseteq_{\sf p} (G,\prec_G)$, if $H\subseteq_{\sf i} G$ and for every pair of vertices $x$ and $y$ of $H$, $x{\prec_G} y$ if and only if $x{\prec_H} y$. A graph $G$ excludes the pattern $(H,\prec_H)$, if there exists a layout $\prec_G$ of $G$ such that $(H,\prec_H)\not\subseteq_{\sf p} (G,\prec_G)$. More generally, a graph class $\mathcal{C}$ excludes a set $\mathcal{P}$ of patterns if for every graph $G\in \mathcal{C}$, there exists a layout $\prec_G$ such that for every pattern $(H,\prec_H)\in\mathcal{P}$, $(H,\prec_H)\not\subseteq_{\sf p} (G,\prec_G)$. We let $\frak{L}(\mathcal{P})$ denote the class of graphs excluding a pattern from  $\mathcal{P}$.  Hereafter, a small size pattern $(H,\prec_H)$ will be encoded by listing its set of (ordered) edges and non-edges. There are two patterns on two vertices and eight patterns on three vertices, see \autoref{fig_size3_pattern}.

\begin{figure}[htbh]
\begin{center}
\begin{tikzpicture}[thick,scale=1.3]
\tikzstyle{sommet}=[circle, draw, fill=black, inner sep=0pt, minimum width=4pt]

\begin{scope}[xshift=-1.5cm,yshift=2cm]
\draw (-1,0) to[bend left=40] (0,0);
\draw (-1,0) node[sommet]{};
\draw (0,0) node[sommet]{};
\node[below] (x) at (-1,0) {$1$};
\node[below] (a) at (-0.5,0) {$\prec$};
\node[below] (y) at (0,0) {$2$};
\node[below] (p) at (-0.5,-0.5) {$\langle{12}\rangle$};
\end{scope}

\begin{scope}[xshift=2.5cm,yshift=2cm]
\draw[red,thick,dashed] (-1,0) to[bend left=30] (0,0);
\draw (-1,0) node[sommet]{};
\draw (0,0) node[sommet]{};
\node[below] (x) at (-1,0) {$1$};
\node[below] (a) at (-0.5,0) {$\prec$};
\node[below] (y) at (0,0) {$2$};
\node[below] (p) at (-0.5,-0.5) {$\langle\overline{12}\rangle$};
\end{scope}

\begin{scope}[shift=(0:-4.5)]
\draw (-1,0) to[bend left=40] (1,0);
\draw (0,0) to[bend left=30] (1,0);
\draw(-1,0) to[bend left=30] (0,0);
\draw (-1,0) node[sommet]{};
\draw (0,0) node[sommet]{};
\draw (1,0) node[sommet]{};
\node[below] (x) at (-1,0) {$1$};
\node[below] (a) at (-0.5,0) {$\prec$};
\node[below] (y) at (0,0) {$2$};
\node[below] (b) at (0.5,0) {$\prec$};
\node[below] (z) at (1,0) {$3$};
\node[below] (p) at (0,-0.5) {$\langle{12},{13},{23}\rangle$};
\end{scope}

\begin{scope}[shift=(0:-1.5)]
\draw (-1,0) to[bend left=40] (1,0);
\draw[red,thick,dashed] (0,0) to[bend left=30] (1,0);
\draw(-1,0) to[bend left=30] (0,0);
\draw (-1,0) node[sommet]{};
\draw (0,0) node[sommet]{};
\draw (1,0) node[sommet]{};
\node[below] (x) at (-1,0) {$1$};
\node[below] (a) at (-0.5,0) {$\prec$};
\node[below] (y) at (0,0) {$2$};
\node[below] (b) at (0.5,0) {$\prec$};
\node[below] (z) at (1,0) {$3$};
\node[below] (p) at (0,-0.5) {$\langle 12,13,\overline{23}\rangle$};
\end{scope}
                
\begin{scope}[shift=(0:1.5)]
\draw (-1,0) to[bend left=40] (1,0);
\draw (0,0) to[bend left=30] (1,0);
\draw[red,thick,dashed] (-1,0) to[bend left=30] (0,0);
\draw (-1,0) node[sommet]{};
\draw (0,0) node[sommet]{};
\draw (1,0) node[sommet]{};
\node[below] (x) at (-1,0) {$1$};
\node[below] (a) at (-0.5,0) {$\prec$};
\node[below] (y) at (0,0) {$2$};
\node[below] (b) at (0.5,0) {$\prec$};
\node[below] (z) at (1,0) {$3$};
\node[below] (p) at (0,-0.5) {$\langle\overline{12},{13},{23}\rangle$};
\end{scope}

\begin{scope}[shift=(0:4.5)]
\draw (-1,0) to[bend left=40] (1,0);
\draw[red,thick,dashed] (0,0) to[bend left=30] (1,0);
\draw[red,thick,dashed] (-1,0) to[bend left=30] (0,0);
\draw (-1,0) node[sommet]{};
\draw (0,0) node[sommet]{};
\draw (1,0) node[sommet]{};
\node[below] (x) at (-1,0) {$1$};
\node[below] (a) at (-0.5,0) {$\prec$};
\node[below] (y) at (0,0) {$2$};
\node[below] (b) at (0.5,0) {$\prec$};
\node[below] (z) at (1,0) {$3$};
\node[below] (p) at (0,-0.5) {$\langle\overline{12},{13},\overline{23}\rangle$};
\end{scope}

\begin{scope}[xshift=-4.5cm,yshift=-2cm]
\draw[red,thick,dashed] (-1,0) to[bend left=40] (1,0);
\draw (0,0) to[bend left=30] (1,0);
\draw  (-1,0) to[bend left=30] (0,0);
\draw (-1,0) node[sommet]{};
\draw (0,0) node[sommet]{};
\draw (1,0) node[sommet]{};
\node[below] (x) at (-1,0) {$1$};
\node[below] (a) at (-0.5,0) {$\prec$};
\node[below] (y) at (0,0) {$2$};
\node[below] (b) at (0.5,0) {$\prec$};
\node[below] (z) at (1,0) {$3$};
\node[below] (p) at (0,-0.5) {$\langle{12},\overline{13},{23}\rangle$};
\end{scope}
                
\begin{scope}[xshift=-1.5cm,yshift=-2cm]
\draw[red,thick,dashed] (-1,0) to[bend left=40] (1,0);
\draw[red,thick,dashed] (0,0) to[bend left=30] (1,0);
\draw (-1,0) to[bend left=30] (0,0);
\draw (-1,0) node[sommet]{};
\draw (0,0) node[sommet]{};
\draw (1,0) node[sommet]{};
\node[below] (x) at (-1,0) {$1$};
\node[below] (a) at (-0.5,0) {$\prec$};
\node[below] (y) at (0,0) {$2$};
\node[below] (b) at (0.5,0) {$\prec$};
\node[below] (z) at (1,0) {$3$};
\node[below] (p) at (0,-0.5) {$\langle{12},\overline{13},\overline{23}\rangle$};

\end{scope}
                
\begin{scope}[xshift=1.5cm,yshift=-2cm]
\draw[red,thick,dashed] (-1,0) to[bend left=40] (1,0);
\draw (0,0) to[bend left=30] (1,0);
\draw[red,thick,dashed] (-1,0) to[bend left=30] (0,0);
\draw (-1,0) node[sommet]{};
\draw (0,0) node[sommet]{};
\draw (1,0) node[sommet]{};
\node[below] (x) at (-1,0) {$1$};
\node[below] (a) at (-0.5,0) {$\prec$};
\node[below] (y) at (0,0) {$2$};
\node[below] (b) at (0.5,0) {$\prec$};
\node[below] (z) at (1,0) {$3$};
\node[below] (p) at (0,-0.5) {$\langle\overline{12},\overline{13},{23}\rangle$};
\end{scope}

\begin{scope}[xshift=4.5cm,yshift=-2cm]
\draw[red,thick,dashed] (-1,0) to[bend left=40] (1,0);
\draw[red,thick,dashed]  (0,0) to[bend left=30] (1,0);
\draw[red,thick,dashed]  (-1,0) to[bend left=30] (0,0);
\draw (-1,0) node[sommet]{};
\draw (0,0) node[sommet]{};
\draw (1,0) node[sommet]{};
\node[below] (x) at (-1,0) {$1$};
\node[below] (a) at (-0.5,0) {$\prec$};
\node[below] (y) at (0,0) {$2$};
\node[below] (b) at (0.5,0) {$\prec$};
\node[below] (z) at (1,0) {$3$};
\node[below] (p) at (0,-0.5) {$\langle\overline{12},\overline{13},\overline{23}\rangle$};
\end{scope}

\end{tikzpicture}
\end{center}
\caption{The  patterns on at most $3$ vertices. $\mathcal{L}(\langle\overline{12},{13},{23}\rangle)$ is the class of chordal graphs. \label{fig_size3_pattern}}
\end{figure}
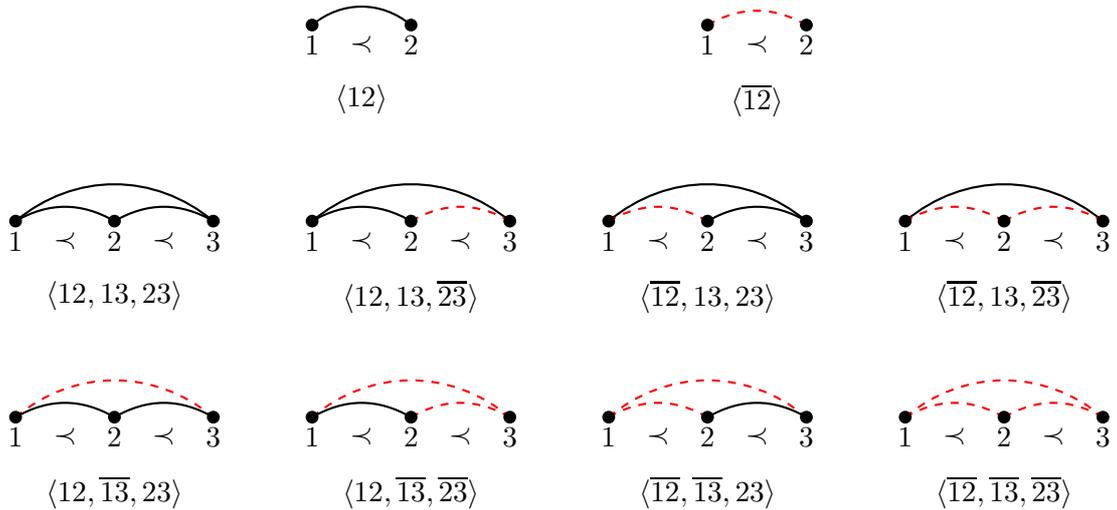

Interestingly, it is known that chordal graphs are characterized by excluding $\mathcal{P}_{\sf chordal}=\big\{\langle\overline{12},{13},{23}\rangle\big\}$, see \autoref{fig_size3_pattern}~\cite{Damaschke90,Duchet84}. This characterization relies on the fact that a graph is chordal if and only if it admits a \emph{simplicial elimination ordering}~\cite{Dirac61,Rose70}. A vertex is \emph{simplicial} if its neighbourhood induces a clique. A simplicial elimination ordering can be defined by a layout $\prec_G$ of $G=(V,E)$ such that every vertex $x$ is simplicial in the subgraph $G[\{y\in V\mid y\prec_{G} x\}]$. Observe that a vertex is simplicial if and only if it is not the mid vertex of a $P_3$, the induced path on three vertices implies the excluded pattern characterization.

In fact, Ginn~\cite{Ginn99} prove that for every pattern $(H,\prec_H)$ such that $H$ is neither the complete graph nor the edge-less graph, characterizing the graph family $\mathcal{L}((H,\prec_H))$ requires an infinite family of forbidden induced subgraphs. Observe however that excluding a unique pattern is important for that result. Indeed, cographs are characterized as $P_4$-free graphs (see discussion above) and need a set $\mathcal{P}_{\sf cograph}$ of several excluded patterns (see \autoref{fig_cograph_patterns}) to be characterized~\cite{Damaschke90}:
\[
\begin{array}{rcl}
\mathcal{P}_{\sf cograph} & =  & \big\{\langle {12},\overline{13},{23} \rangle, \langle \overline{12},{13},\overline{23}\rangle, \langle\overline{12},{13},\overline{14},{23},{24},\overline{34}\rangle, 
\langle{12},\overline{13},{14},\overline{23},\overline{24},{34}\rangle\big\}
\end{array}
\]

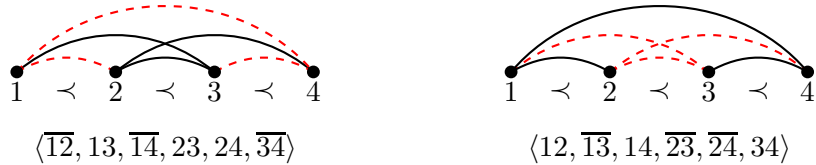
\begin{figure}[htbh]
\begin{center}
\begin{tikzpicture}[thick,scale=1.3]
\tikzstyle{sommet}=[circle, draw, fill=black, inner sep=0pt, minimum width=4pt]
                
\begin{scope}[shift=(0:-4)]
\draw (-1,0) to[bend left=40] (1,0);
\draw (0,0) to[bend left=30] (1,0);
\draw (0,0) to[bend left=40] (2,0);
\draw[red,thick,dashed](-1,0) to[bend left=30] (0,0);
\draw[red,thick,dashed](1,0) to[bend left=30] (2,0);
\draw[red,thick,dashed](-1,0) to[bend left=50] (2,0);
\draw (-1,0) node[sommet]{};
\draw (0,0) node[sommet]{};
\draw (1,0) node[sommet]{};
\draw (2,0) node[sommet]{};
\node[below] (x) at (-1,0) {$1$};
\node[below] (a) at (-0.5,0) {$\prec$};
\node[below] (y) at (0,0) {$2$};
\node[below] (b) at (0.5,0) {$\prec$};
\node[below] (v) at (1,0) {$3$};
\node[below] (c) at (1.5,0) {$\prec$};
\node[below] (z) at (2,0) {$4$};
\node[below] (p) at (0.5,-0.5) {$\langle\overline{12},{13},\overline{14},{23},{24},\overline{34}\rangle$};
\end{scope}


\begin{scope}[shift=(0:1)]
\draw[red,thick,dashed] (-1,0) to[bend left=40] (1,0);
\draw[red,thick,dashed] (0,0) to[bend left=30] (1,0);
\draw[red,thick,dashed] (0,0) to[bend left=40] (2,0);
\draw (-1,0) to[bend left=30] (0,0);
\draw (1,0) to[bend left=30] (2,0);
\draw (-1,0) to[bend left=50] (2,0);
\draw (-1,0) node[sommet]{};
\draw (0,0) node[sommet]{};
\draw (1,0) node[sommet]{};
\draw (2,0) node[sommet]{};
\node[below] (x) at (-1,0) {$1$};
\node[below] (a) at (-0.5,0) {$\prec$};
\node[below] (y) at (0,0) {$2$};
\node[below] (b) at (0.5,0) {$\prec$};
\node[below] (v) at (1,0) {$3$};
\node[below] (c) at (1.5,0) {$\prec$};
\node[below] (z) at (2,0) {$4$};
\node[below] (p) at (0.5,-0.5) {$\langle{12},\overline{13},{14},\overline{23},\overline{24},{34}\rangle$};
\end{scope}

    \end{tikzpicture}
\end{center}
\vspace{-0.4cm}
\caption{The two size $4$ forbidden patterns of cographs. 
\label{fig_cograph_patterns}}
\end{figure}

In~\cite{DuffusGR95}, Duffus et al. investigate the computational complexity of the recognition problem of $\mathcal{L}((H,\prec_H))$ for a fixed ordered graph $(H,\prec_H)$. They conjectured that if $H$ is neither the complete graph nor the edge-less graph, then recognizing $\mathcal{L}((H,\prec_H))$ is \textsf{NP}-complete if $H$ or its complement is $2$-connected. Hell et al.~\cite{HellMR14} have recently shown that if $\mathcal{P}$ only  contains patterns of size at most $3$, then  $\frak{L}(\mathcal{P})$ can be recognized in polynomial time using a $2$-\textsf{SAT} approach. Besides chordal graphs (see discussion above), these graph classes comprise very important graph classes, among others:
\begin{itemize}
\item Bipartite graphs exclude $\mathcal{P}_{\sf bip}=\big\{\langle {12},{13},{23}\rangle,\langle {12},\overline{13},{23}\rangle \big\}$: If $B=(X,Y,E)$, then every layout such that for every $x\in X$, $y\in Y$, $x\prec_B y$ is $\mathcal{P}_{\sf bip}$-free.
\item Forests exclude $\mathcal{P}_{\sf forest}=\big\{\langle {12},{13},{23}\rangle,\langle \overline{12},{13},{23}\rangle \big\}$: A $\mathcal{P}_{\sf forest}$-free layout is obtained by ordering the vertices of $T$ according to a depth first search ordering of $T$.
\item \emph{Co-comparability graphs}~\cite{Ghouila-Houri62,Gallai67}  exclude $\mathcal{P}_{\sf cocomp}=\big\{\langle \overline{12},{13},\overline{23}\rangle \big\}$: A co-comparability graph is a graph whose complement can be transitively oriented. A $\mathcal{P}_{\sf cocomp}$-free layout is obtained as a linear extension of a transitive orientation of $\overline{G}$;
\item \emph{Interval graphs}~\cite{Hajos57,Benzer59,Golumbic80} exclude $\mathcal{P}_{\sf int}=\mathcal{P}_{\sf cocomp}\cup\mathcal{P}_{\sf chordal}$~\cite{Olariu91}: A graph is an interval graph if it is the intersection graph of a set of intervals on the real line. The existence of a $\mathcal{P}_{\sf int}$-layout for interval graphs follows from the fact that a graph is an interval graph if and only if it is chordal and co-comparability.
\item \emph{Proper interval graphs}~\cite{Roberts68,Roberts69} 
exclude $\mathcal{P}_{\sf proper}=\big\{ \langle {12},{13},\overline{23}\rangle, \langle \overline{12},{13},{23}\rangle \big\}$~\cite{Damaschke90}. A graph is a proper interval graph if it is the intersection graph of a set of proper intervals on the real line (no interval is a subset of another one). We observe that proper interval graphs where originally characterized by the existence of a so-called \emph{indifference orderings} ~\cite{Roberts68,Roberts69} that are exactly the layouts excluding $\mathcal{P}_{\sf int}\cup \big\{ \langle {12},{13},\overline{23}\rangle \big\}$~\cite{Roberts68,Roberts69}. As we will see later, a layout excluding $\mathcal{P}_{\sf int}\cup \big\{ \langle {12},{13},\overline{23}\rangle \big\}$ is a $\mathcal{P}_{\sf proper}$-free layout.

\item \emph{Trivially perfect graphs}~\cite{Golumbic78} exclude $\mathcal{P}_{\sf trivPer}=\mathcal{P}_{\sf chordal}\cup \mathcal{P}_{\sf comp}$. A graph $G$ is a trivially perfect graph if and only if it is $\{P_4,C_4\}$-free, or equivalently $G$ is the comparability graph of a rooted tree $T$ (two vertices are adjacent if one is the ancestor of the other). A $\mathcal{P}_{\sf trivPer}$-free layout is obtained from a depth first search ordering of $T$. Moreover every layout of the $P_4$ and the $C_4$ contains one of the patterns of $\mathcal{P}_{\sf trivPer}$ (see~\cite{FeuilloleyH21}).
\end{itemize}
For more examples, the reader should refer to~\cite{Damaschke90,FeuilloleyH21}. Feuilloley and Habib~\cite{FeuilloleyH21} list all the graph classes that can be obtained by excluding a set of patterns each of size at most $3$. Moreover, for most of them (but two), they argue about the existence of a linear time recognition algorithm.

\paragraph{From layouts to tree-layouts.}
A layout $\prec_G$ of a graph $G=(V,E)$ on $n$ vertices can be viewed as an embedding of $G$ into a path $P$ on $n$ vertices rooted at one of its extremities. Under this view point, it becomes natural to consider graph embeddings in graphs that are more general than rooted paths. Recently, Guzman-Pro et al.~\cite{Guzman-ProHH23} have studied embeddings in a cyclic ordering. In this paper, we consider embedding the vertices of a graph in a rooted tree, yielding the notion of \emph{tree-layout}.

\begin{definition} \label{def_treelayout}
Let $G=(V,E)$ be a graph on $n$ vertices. A \emph{tree-layout} of $G$ is a triple $\T_G=(T,r,\rho_G)$ where $T$ is a tree on a set $V_T$ of $n$ nodes rooted at $r$ and $\rho_G:V\rightarrow V_T$ is a bijection such that for every edge $xy\in E$, either $x$ is an ancestor of $y$, denoted by $x\prec_{\T} y$, or $y$ is an ancestor of $x$.
\end{definition}

\begin{figure}[ht]
\begin{center}
\begin{tikzpicture}[thick,scale=0.7]
\tikzstyle{sommet}=[circle, draw, fill=black, inner sep=0pt, minimum width=4pt]

\draw  (0,0)  node[sommet]{}
    -- (0,2) node[sommet]{}
    -- (1,3.4) node[sommet]{}
    -- (2,2) node[sommet]{}
    -- (2,0) node[sommet]{}
    -- cycle;
\draw (0,2) -- (2,2);
\draw (1,-1.34) node[sommet]{};
\draw (-1.34,1) node[sommet]{};
\draw (-1.34,-1.34) node[sommet]{};
\draw (-1.34,1) -- (-1.34,-1.34);
\draw (-1.34,1) -- (0,2);
\draw (-1.34,-1.34) -- (0,0);
\draw (-1.34,-1.34) -- (1,-1.34);
\draw (0,0) -- (1,-1.34);
\draw (2,0) -- (1,-1.34);

\node[right] (a) at (1.1,3.4) {$a$};
\node[left] (b) at (-0.1,2) {$b$};
\node[right] (c) at (2.1,2) {$c$};
\node[left] (d) at (-1.44,1) {$d$};
\node[left] (e) at (-0.1,0) {$e$};
\node[right] (f) at (2.1,0) {$f$};
\node[left] (g) at (-1.44,-1.34) {$g$};
\node[right] (h) at (1.1,-1.34) {$h$};

\node[above] (G) at (-1.5,2.5) {$G=(V,E)$};

\draw[red,thick] (8,4) .. controls (8.6,2.5) .. (9,1) ;
\draw[red,thick] (8,4) .. controls (7.4,2.5) .. (7,1) ;
\draw[red,thick] (8,4) .. controls (10,1) .. (10,-2) ;
\draw[red,thick] (8,4) .. controls (6,1.75) .. (6,-0.5) ;
\draw[red,thick] (8,2.5) .. controls (7.85,1.75) .. (7,1) ;
\draw[red,thick] (8,2.5) .. controls (8.15,1.75) .. (9,1) ;
\draw[red,thick] (8,2.5) .. controls (8.3,0.25) .. (8,-2) ;
\draw[red,thick] (7,1) .. controls (6.85,0.25) .. (6,-0.5) ;
\draw[red,thick] (9,1) .. controls (9.25,0.25) .. (9,-0.5) ;
\draw[red,thick] (9,1) .. controls (8.4,-0.5) .. (8,-2) ;
\draw[red,thick] (9,-0.5) .. controls (8.8,-1.25) .. (8,-2) ;
\draw[red,thick] (9,-0.5) .. controls (9.2,-1.25) .. (10,-2) ;

\draw[very thick] (8,4) node[sommet]{}
-- (8,2.5) node[sommet]{}
-- (7,1) node[sommet]{}
--  (6,-0.5) node[sommet]{};

\draw[very thick] (8,2.5)
-- (9,1) node[sommet]{}
-- (9,-0.5) node[sommet]{}
-- (8,-2) node[sommet]{};

\draw[very thick] (9,-0.5) -- (10,-2) node[sommet]{};

\node[right] (bb) at (8.1,4) {$r=\rho(b)$};
\node[left] (ff) at (8,2.5) {$\rho(f)$};
\node[left] (cc) at (6.9,1) {$\rho(c)$};
\node[left] (aa) at (5.9,-0.5) {$\rho(a)$};
\node[right] (ee) at (9.1,1) {$\rho(e)$};
\node[right] (gg) at (9.1,-0.5) {$u=\rho(g)$};
\node[left] (hh) at (7.9,-2) {$\rho(h)$};
\node[right] (dd) at (10.1,-2) {$\rho(d)$};

\node[above] (T) at (5.5,2.5) {$(T,r,\rho)$};

    \end{tikzpicture}
\end{center}
\caption{\label{fig_tree_layout} A tree-layout $(T,r,\rho)$ of a graph $G=(V,E)$.}
\end{figure}
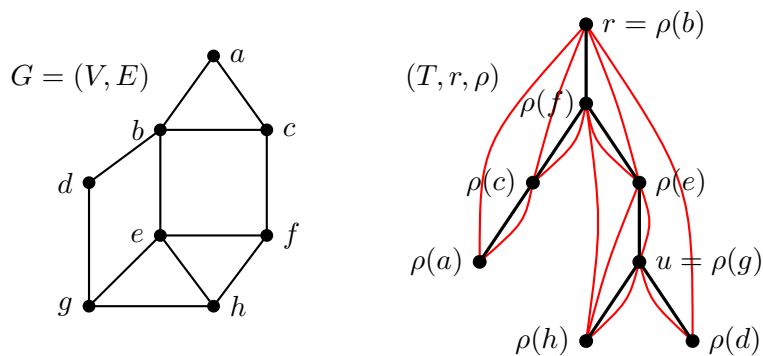

We observe that, from \autoref{def_treelayout}, in a \emph{tree-layout} $\T_G=(T,r,\rho_G)$ of a graph  $G$, $T$ is not a \emph{Trémaux tree} since it is not necessarily a spanning tree of $G$ (see \cite{NesetrilO06Treedepth} and \cite{Bienstock90OnEmbedding} for similar concepts). However, it is easy to see that a tree-layout $\T_G=(T,r,\rho_G)$ such that $T$ is a path is a layout of $G$. So from now on, we shall define a \emph{layout} as a triple  $\L_H=(P,r,\rho_H)$,where $P$ is a path that fulfils the conditions of \autoref{def_treelayout}. An ordered graph then becomes a pair $(H,\L_H)$ where $\L_H=(P,r,\rho_H)$ is a layout of $G$. Excluding a pattern $(H,\L_H)$  in a tree-layout of a graph $G$ is defined similarly as excluding a pattern in a layout, but now using the ancestor relation. For a set $\mathcal{P}$ of patterns, we can also define the class $\frak{T}(\mathcal{P})$ of graphs admitting a tree-layout that excludes every pattern $P\in\mathcal{P}$. If $\mathcal{P}=\{(H,\L_H)\}$, we simply write  $\mathcal{T}((H,\L_H))$. Observe that, as a layout is a tree-layout, for a fixed set $\mathcal{P}$ of patterns, we always have $\frak{L}(\mathcal{P})\subseteq \frak{T}(\mathcal{P})$. As an introductory example, let us consider the pattern $\langle \overline{12}\rangle$. The following observation directly follows from the definitions of a tree-layout and trivially perfect graphs. Indeed, recall that a trivially perfect graph $G$ is the comparability graph of a rooted tree $T$, that is two vertices are adjacent if one is the ancestor of the other. It follows that if $r$ is the root of  $T$, then $(T,r,\mathbb{I})$, where $\mathbb{I}$ is the identity, is a tree-layout of $G$.

\begin{observation} \label{obs_trivially_perfect}
The class $\mathcal{L}(\langle \overline{12}\rangle)$ is the set of complete graphs while the class $\mathcal{T}(\langle \overline{12}\rangle)$ is the set of trivially perfect graphs.
\end{observation}

So the class of trivially perfect graphs can be viewed as the tree-like version of the class of complete graphs. This view point motivates the systematic study of graph classes defined by excluding a fixed pattern (or a fixed set of patterns) in a tree-layout. The first questions are probably to understand the graph classes we obtain by excluding patterns of size at most $3$ in tree-layouts. Do we retrieve some known graph classes or do we define novel graph classes? What is the computational complexity of the recognition problem of these graph classes? What are the relationship of these novel graph classes with the known ones? Do these novel graph classes allows to solve in polynomial time problems that are \textsc{NP}-hard on arbitrary graphs?

\paragraph{Our contributions.}
As a first study of tree-layout based graph classes, we consider the patterns characterizing interval graphs and proper interval graphs. We first show that if we consider the interval graphs patterns $\mathcal{P}_{\sf int}$, the same phenomena as for $\big\{\langle \overline{12}\rangle\big\}$ holds, leading to a novel (up to our knowledge) characterization of chordal graphs as being exactly $\frak{T}(\mathcal{P}_{\sf int})$ (see \autoref{th_chordal_tree}). 


As already discussed, proper interval graphs are obtained by restricting interval graphs to the intersection of a set of \emph{proper} intervals (no interval is a subinterval of another). This restriction leads to a distinct graph class as the $K_{1,3}$ is an interval graph but not a proper one. Following this line, in his seminal paper~\cite{Gavril74} characterizing chordal graphs as the intersection graphs of a subset of subtrees of a tree, Gavril considered the class of intersection graphs of a set of \emph{proper} subtrees of a tree  (no subtree is contained in an another). Using an easy reduction, Gavril proved that this again yields a characterization of chordal graphs. So this left open the question of proposing a natural definition for proper chordal graphs, a class of graphs that should be sandwiched between proper interval graphs and chordal graphs but incomparable to interval graphs. Our main contribution is to propose such a natural definition of \emph{proper chordal graphs} by means of forbidden patterns on tree-layouts: a graph is proper chordal if it belongs to  $\frak{T}(\mathcal{P}_{\sf proper})$, the class of graphs admitting a $\mathcal{P}_{\sf proper}$-free tree-layout. Recall that $\mathcal{P}_{\sf proper}$ are the patterns characterizing proper interval graphs on layouts. \autoref{tab_graph_classes} resumes the discussion above about graph classes respectively obtained from layouts and tree-layout excluding a fixed set of patterns and positions proper chordal graphs with respect to trivially perfect graphs, chordal graphs and proper interval graphs.  In a recent paper,  Chaplick~\cite{Chaplick19Intersection} investigated this question and considered the class of intersection graphs of non-crossing paths in a tree.


\begin{table}[h]
\[
\begin{array}{|c|c|c|}
\hline
\mbox{\bf Forbidden patterns} & \mbox{\bf Layouts} & \mbox{\bf Tree-layouts}\\
\hline
\langle \overline{12}\rangle & \mbox{Cliques} & \mbox{Trivially perfect graphs}\\
\hline
\langle {12},{13},\overline{23}\rangle, \langle\overline{12},{13},{23}\rangle & \mbox{Proper interval graphs} & \mbox{\textbf{Proper chordal graphs}}\\
\hline
\langle \overline{12},{13},\overline{23}\rangle, \langle\overline{12},{13},{23}\rangle & \mbox{Interval graphs} & \mbox{Chordal graphs}\\
\hline
\end{array}
\]
\caption{Graph classes obtained by excluding $\langle\overline{12}\rangle$, $\mathcal{P}_{\sf proper}$ and $\mathcal{P}_{\sf int}$.
\label{tab_graph_classes}}
\end{table}

We then provide a thorough study of the class of proper chordal graphs, both on their combinatorial as well as their algorithmic aspects. In reference to the \emph{indifference layout} characterizing proper interval graphs, we call a 
$\mathcal{P}_{\sf proper}$-free tree-layout an \emph{indifference tree-layout}. Our first result (see \autoref{lem_indifference_treelayout}) is a characterization of indifference tree-layouts (and henceforth of proper chordal graphs). As discussed in \autoref{sec_FPQ}, (proper) interval graphs have, in general, multiple (proper) interval representations, which are all captured by a canonical tree structure called \PQ-tree~\cite{BoothL76}. We show that the set of indifference tree-layouts can also be represented in a canonical and compact way, see~\autoref{th_canonical}.
These structural results have very interesting algorithmic implications. First, we can design a polynomial time recognition algorithm for proper chordal graphs, see \autoref{th_recognition}. Second, we show that the isomorphism problem restricted to proper chordal graphs is polynomial time solvable, see~\autoref{th_isomorphism}. Interestingly, this problem is \textsf{GI}-complete on (strongly) chordal graphs~\cite{Muller96}. So considering proper chordal graphs allows us to push the tractability further towards its limit.

\blue{
}


\section{Preliminaries}

\subsection{Notations and definitions}

\paragraph{Graphs.}
In this paper, every graph is finite, loopless, and without multiple edges.
A graph is a pair $G=(V,E)$ where $V$ is its vertex set and $E\subseteq V^2$ is the set of edges. For two vertices $x,y\in V$, we let $xy$ denoted the edge $e=\{x,y\}$. We say that the vertices $x$ and $y$ are incident with the edge $xy$. The neighbourhood of a vertex $x$ is the set of vertices $N(x)=\{y\in V\mid xy\in E\}$. Let $S$ be a subset of vertices of $V$. The closed neighbourhood of a vertex $x$ is $N[x]=N(x)\cup\{x\}$. The graph resulting from the removal of a vertex subset $S$ of $V(G)$ is denoted by $G-S$. If $S=\{x\}$ is a singleton we write $G-x$ instead of $G-\{x\}$. The subgraph of $G$ induced by $S$ is $G[S]=G-(V\setminus S)$. We say that $S$ is a \emph{separator} of $G$ if $G-S$ contains more connected components than $G$. We say that $S$ separates $X\subseteq V(G)$ from $Y\subseteq V(G)$ if $X$ and $Y$ are subsets of distinct connected components of $G-S$. If $x$ is a vertex such that $x\notin S$ and $S\subseteq N(x)$, then we say that $x$ is \emph{$S$-universal}.

\paragraph{Rooted trees.}
A rooted tree is a pair $(T,r)$ where $r$ is a distinguished node\footnote{To avoid confusion, we reserve the term vertex for graphs and nodes for trees.} of the tree $T$. We say that a node $u$ is an \emph{ancestor} of the node $v$ (and that $u$ is a \emph{descendant} of $v$) if $u$ belongs to the unique path of $T$ from $r$ to $v$. If $u$ is an ancestor of $v$, then we write $u\prec_{(T,r)} v$. For a node $u$ of $T$,  the set $A_{(T,r)}(u)$ contains every \emph{ancestor} of $u$, that is every node $v$ of $T$ such that $v\prec_{(T,r)} u$. Likewise, the set $D_{(T,r)}(u)$ contains every \emph{descendant} of $u$, that is every node $v$ of $T$ such that $u\prec_{(T,r)} v$. We may also write $A_{(T,r)}[u]=A_{(T,r)}(u)\cup\{u\}$ and $D_{(T,r)}[u]=D_{(T,r)}(u)\cup\{u\}$. The least common ancestor of two nodes $u$ and $v$ is denoted $\lca_{(T,r)}(u,v)$. For a node $u$, we define $T_u$ as the subtree of $(T,r)$ rooted at $u$ and containing the descendants of $u$. We let $\mathcal{L}(T,r)$ 
denote the set of leaves of $(T,r)$ and for a node $u$, $\mathcal{L}_{(T,r)}(u)$ is the set of leaves of $(T,r)$ that are descendants of the node $u$. 

\paragraph{Ordered trees.}
An ordered tree is a rooted tree $(T,r)$ such that the children of every internal node are totally ordered. For an internal node $v$, we let denote $\sigma_{(T,r)}^v$ the permutation of its children. An ordered tree is non-trivial if it contains at least one internal node. An ordered tree $(T,r)$ defines a permutation $\sigma_{(T,r)}$ of its leaf set $\mathcal{L}(T,r)$ as follows.
For every pair of leaves $x,y\in \mathcal{L}(T,r)$, let $u_x$ and $u_y$ be the children of $v=\mathbf{lca}_{(T,r)}(x,y)$ respectively being an ancestor of $x$ and of $y$. Then $x\prec_{\sigma_{(T,r)}} y$ if and only if $u_x\prec_{\sigma_{T,r}^v}u_y$.

\paragraph{Tree-layouts of graphs.}
Let $\T=(T,r,\rho)$ be a tree-layout of a graph $G=(V,E)$ (see \autoref{def_treelayout}). Let us recall that from~\autoref{def_treelayout}, $G$ is a subgraph of the transitive closure of $(T,r)$. Let $x$ and $y$ be two vertices of $G$. We note $x\prec_{\T} y$ if $\rho(x)$ is an ancestor of $\rho(y)$ in $(T,r)$ and use $A_{\T}(x)$, $D_{\T}(x)$ to respectively denote the ancestors and descendants of $x$. The notations $\mathcal{L}(\T)$, $\T_x$, $\mathcal{L}_{\T}(x)$ are derived from the notations defined above. 

\subsection{Chordal and interval graphs}
\label{sec_chordal}

Let us discuss in further details the classes of chordal, interval and proper interval graphs. A graph is \emph{chordal} if it does not contain any cycle $C_k$ with $k\geq 4$ as induced subgraph (or equivalently every cycle of length at least $4$ contains a chord)~\cite{HajnalS58,Berge61}. Gavril~\cite{Gavril74} characterized chordal graphs as the intersection graphs of a family of subtrees of a tree. Given a chordal graph $G=(V,E)$, a  \emph{tree-intersection model} of $G=(V,E)$ is defined as a triple $\mathbf{M}^{\sf T}_G=(T,\mathcal{T},\tau_G)$ where $T$ is a tree, $\mathcal{T}$ is a family of subtrees of $T$ and $\tau_G:V\rightarrow \mathcal{T}$ is a bijection such that $xy\in E$ if and only if $\tau_G(x)$ intersects $\tau_G(y)$. Hereafter, we denote by $T^x\in \mathcal{T}$ the subtree of $T$ such that $T^x=\tau_G(x)$. Since chordal graphs are characterized by the existence of a simplicial elimination ordering~\cite{Rose70}, we obtain the following.

\begin{theorem} \cite{Damaschke90}
The class of chordal graphs is $\frak{L}(\mathcal{P}_{\sf chordal})$ where \\
\centerline{ $\mathcal{P}_{\sf chordal}=\big\{\langle\overline{12},{13},{23}\rangle\big\}$.}
\end{theorem}

A graph $G=(V,E)$ is an \emph{interval graph} if it is the intersection graph of a family of intervals on a line~\cite{Hajos57,Benzer59,Golumbic80}. So the intersection model of an interval graph $G$ is $\mathbf{M}^{\sf I}_G=(P,\mathcal{P},\tau_G)$ where $P$ is a path, $\mathcal{P}$ is a family of subpaths of $P$. 
The interval, or subpath, $\tau_G(x)\in\mathcal{P}$ will be denoted $P^x$. Observe that if $G$ contains $n$ vertices then  there exists $\mathbf{M}^{\sf I}_G=(P,\mathcal{P},\tau_G)$ s.t. $P$ has at most $2n$ nodes. It follows that for every $x\in V$, $P^x$ can be represented by an interval $[l(x),r(x)]$.
As compared to chordal graphs, we restrict the intersection model to paths, interval graphs form a subset of chordal graphs. Moreover, it is easy to observe that the intersection model of an interval naturally defines a transitive orientation of the non-edges of $G$. These observations lead to the following.

\begin{theorem} \cite{Damaschke90} \label{th_interval}
The class of interval graphs is $\frak{L}(\mathcal{P}_{\sf int})$ where  \\
\centerline{$\mathcal{P}_{\sf int}=\big\{\langle\overline{12},{13},{23}\rangle,\langle\overline{12},{13},\overline{23}\rangle\big\}$.}
\end{theorem}

To sketch the proof of \autoref{th_interval}, we describe how a $\mathcal{P}_{\sf int}$-free layout of $G$ can be obtained from an interval-intersection model of $G$ and vice-versa. Suppose that $\mathbf{M}^{\sf I}_G$ is an interval-intersection model of the interval graph $G=(V,E)$. Then a $\mathcal{P}_{\sf int}$-free layout $\L_G=(P,r,\rho_G)$ of $G$ is obtained as follows: for every $x,y\in V$, $x\prec_{\L_G} y$ if and only if $r(x)<r(y)$ or $l(x)< l(y)<r(y)< r(x)$ (see Figure~\ref{fig_interval_layout}).

\begin{figure}[ht]
\begin{center}
\begin{tikzpicture}[thick,scale=0.7]
\tikzstyle{sommet}=[circle, draw, fill=black, inner sep=0pt, minimum width=4pt]
                

\draw (1,0.75) -- (3,0.75) ;
\draw (2,1.5) -- (6,1.5) ;
\draw (4,0) -- (12,0) ;
\draw (5,0.75) -- (8,0.75) ;
\draw (7,1.5) -- (10,1.5) ;
\draw (9,0.75) -- (11,0.75) ;

\node[below] (a) at (2,0.65) {$a$};
\node[below] (b) at (4,1.45) {$b$};
\node[below] (c) at (6.5,0.7) {$c$};
\node[below] (d) at (8.5,1.45) {$d$};
\node[below] (e) at (10,0.7) {$e$};
\node[below] (x) at (8,-0.05) {$x$};

\draw  (-7,1.2)  node[sommet]{}
    -- (-5.5,1.2) node[sommet]{}
    -- (-4,1.2) node[sommet]{}
    -- (-2.5,1.2) node[sommet]{}
    -- (-1,1.2) node[sommet]{};
\draw (-1,1.2) -- (-3.25,0) node[sommet]{} ;
\draw (-5.5,1.2) -- (-3.25,0);
\draw (-4,1.2) -- (-3.25,0);
\draw (-2.5,1.2) -- (-3.25,0);

\node[above] (a) at (-7,1.3) {$a$};
\node[above] (b) at (-5.5,1.3) {$b$};
\node[above] (c) at (-4,1.3) {$c$};
\node[above] (d) at (-2.5,1.3) {$d$};
\node[above] (e) at (-1,1.3) {$e$};
\node[below] (x) at (-3.25,-0.05) {$x$};
    \end{tikzpicture}
\end{center}
\caption{\label{fig_interval_layout} An interval graph $G=(V,E)$ and an interval representation of $G$. The layout $\L_{G}=(P,r,\rho)$ associated to that interval representation is $a \prec_{\L_G} b \prec_{\L_G} x \prec_{\L_G} c \prec_{\L_G} d \prec_{\L_G} e$.}
\end{figure}
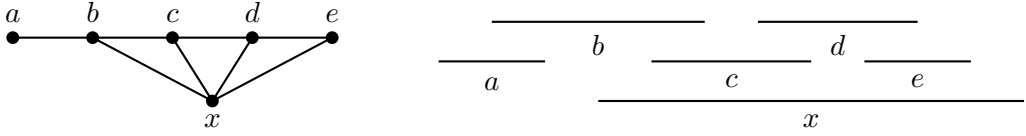

Suppose that $\L_G$ is a $\mathcal{P}_{\sf int}$-free layout of $G$. For every vertex $x\in V$, we define $\sigma(x)=|A_{\T}(x)|+1$. Then an interval-intersection model $\mathbf{M}^{\sf I}_G$ of $G$ is obtained as follows: with every vertex $x\in V$, we associate the interval $[\sigma(x),\sigma(y)]$ where $y$ is the vertex in $N[x]$ such that $x\prec_{\L_G} y$ and for every vertex $z$ such that $y \prec_{\L_G} z$, $z\notin N(x)$.

From the discussion above, it is legitimate to characterize chordal graphs as the \emph{"tree-like"} interval graphs. Indeed, this is how the tree intersection model was introduced in the seminal paper of Gavril~\cite{Gavril74}. It is then natural to wonder whether the phenomena observed with the classes $\mathcal{L}(\langle\overline{12}\rangle)$ and  $\mathcal{T}(\langle\overline{12}\rangle)$ (see the discussion in \autoref{sec_intro}) is confirmed with the set of patterns $\mathcal{P}_{\sf int}$. Surprisingly this question yields (up to our knowledge), the following novel characterization of chordal graphs.

\begin{theorem} \label{th_chordal_tree}
The class of chordal graphs is $\frak{T}(\mathcal{P}_{\sf int})$.
\end{theorem}
\begin{proof} The proof is based on a construction similar to the one we sketched for \autoref{th_interval}.

\noindent
$\Rightarrow$ Let $\mathbf{M}^{\sf T}_G=(T',\mathcal{T}',\tau_G)$ be a tree-intersection model of $G$. We construct a tree-layout $\T_G=(T,r,\rho_G)$ of $G$ as follows. First we root $T'$ at an arbitrary node $r$. This defines for every vertex $x$, a node $\rho'(x)$ which is the root of the subtree $\tau(x)$ of $T'$.  Observe that we may assume that $r$ is the root $\rho'(x)$ for some vertex $x\in V$. We set $R(T')=\{u\in V_{T'}\mid \exists x\in V, \rho'(x)=u\}$. The tree $T$ is obtained from $T'$ by first contracting every tree-edge $uv$ such that $v$ the child of $u$ does not belong to $R(T')$. 
These contractions preserves the mapping $\rho'$. Then every node $u$ of the resulting tree is replaced by a path $P_u$ of length $\ell=|\{x\in V\mid u=\rho(x)\}|$. Finally, $\rho'$ is modified to turn it into a bijection $\rho_G$ from $V$ to $V_{T}$ as follows: the vertices mapped to a given node $u$ of $T'$ are mapped to distinct nodes of $P_u$ (see \autoref{fig_intersection_treelayout}). Let us prove that the construction is correct.

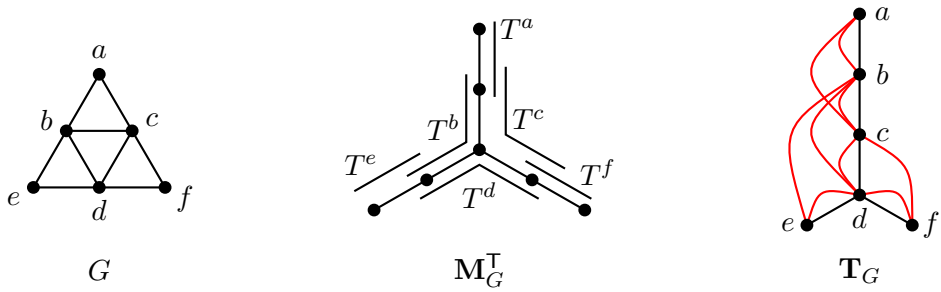
\begin{figure}[ht]
\begin{center}
\begin{tikzpicture}[thick,scale=1]
\tikzstyle{sommet}=[circle, draw, fill=black, inner sep=0pt, minimum width=4pt]
          
\begin{scope}
\draw (90:1) node[sommet]{};
\draw (210:1) node[sommet]{};
\draw (330:1) node[sommet]{};
\draw (270:0.5) node[sommet]{};
\draw (150:0.5) node[sommet]{};
\draw (30:0.5) node[sommet]{};

\node (a) at (90:1.3){$a$};
\node (b) at (150:0.8){$b$};
\node (c) at (30:0.8){$c$};
\node (d) at (270:0.8){$d$};
\node (e) at (210:1.3){$e$};
\node (f) at (330:1.3){$f$};

\draw (90:1) -- (210:1);
\draw (90:1) -- (330:1);
\draw (210:1) -- (330:1);
\draw (150:0.5) -- (30:0.5);
\draw (150:0.5) -- (270:0.5);
\draw (270:0.5) -- (30:0.5);

\node[] (G) at (90:-1.6){$G$};
\end{scope}

\begin{scope}[shift=(0:5)]

\node (1) at (90:1.6){};
\draw (1) node[sommet]{};
\node (2) at (90:0.8){};
\draw (2) node[sommet]{};
\node (3) at (0:0){};
\draw (3) node[sommet]{};
\node (4) at (210:0.8){};
\draw (4) node[sommet]{};
\node (5) at (210:1.6){};
\draw (5) node[sommet]{};
\node (6) at (330:0.8){};
\draw (6) node[sommet]{};
\node (7) at (330:1.6){};
\draw (7) node[sommet]{};
\draw (1.center) -- (3.center);
\draw (5.center) -- (3.center);
\draw (7.center) -- (3.center);

\begin{scope}[xshift=0.2cm]
\draw (90:1.7) -- (90:0.7);
\node[xshift=0.3cm] (Ta) at (90:1.6) {$T^a$};
\end{scope}
\begin{scope}[shift=(150:0.2)]
\draw (90:0.9) -- (90:0) -- (210:0.9) ;
\node[xshift=-0.3cm] (Tb) at (90:0.2) {$T^b$};
\end{scope}
\begin{scope}[shift=(30:0.4)]
\draw (90:0.9) -- (90:0) -- (330:0.9) ;
\node[xshift=0.3cm] (Tc) at (90:0.2) {$T^c$};
\end{scope}
\begin{scope}[yshift=-0.2cm]
\draw (210:0.9) -- (0:0) -- (330:0.9);
\node (Td) at (90:-0.4) {$T^d$};
\end{scope}
\begin{scope}[shift=(150:0.4)]
\draw (210:1.5) -- (210:0.5);
\node[shift=(90:0.3)] (Te) at (210:1.4) {$T^e$};
\end{scope}
\begin{scope}[shift=(30:0.2)]
\draw (330:1.5) -- (330:0.5);
\node[shift=(90:0.4)] (Tf) at (330:1.6) {$T^f$};
\end{scope}

\node[] (T) at (90:-1.6){$\mathbf{M}^{\sf T}_G$};

\end{scope}

\begin{scope}[xshift=10cm,yshift=-0.6cm]
\draw[red,thick] (90:2.4) .. controls (100:2.05) .. (90:1.6) ;
\draw[red,thick,shift=(90:-0.8)] (90:2.4) .. controls (100:2.05) .. (90:1.6) ;
\draw[red,thick,shift=(90:-1.6)] (90:2.4) .. controls (100:2.05) .. (90:1.6) ;
\draw[red,thick,rotate=240,shift=(90:-1.6)] (90:1.6) .. controls (100:2.05) .. (90:2.4) ;
\draw[red,thick,rotate=120,shift=(90:-1.6)] (90:1.6) .. controls (80:2.05) .. (90:2.4) ;

\draw[red,thick] (210:0.8) .. controls (140:1.4) .. (90:1.6) ;
\draw[red,thick,rotate=240] (210:0.8) .. controls (150:0.8) .. (90:0.8) ;

\draw[red,thick] (90:0) .. controls (135:1.2) .. (90:1.6) ;
\draw[red,thick,shift=(90:0.8)] (90:0) .. controls (135:1.2) .. (90:1.6) ;

\node (1) at (90:2.4){};
\draw (1) node[sommet]{};
\node (2) at (90:1.6){};
\draw (2) node[sommet]{};
\node (3) at (90:0.8){};
\draw (3) node[sommet]{};
\node (4) at (0:0){};
\draw (4) node[sommet]{};
\node (5) at (210:0.8){};
\draw (5) node[sommet]{};
\node (6) at (330:0.8){};
\draw (6) node[sommet]{};
\draw (1.center) -- (4.center);
\draw (5.center) -- (4.center);
\draw (6.center) -- (4.center);

\node[shift=(0:0.3)] (a) at (1){$a$};
\node[shift=(0:0.3)] (b) at (2){$b$};
\node[shift=(0:0.3)] (c) at (3){$c$};
\node[shift=(-90:0.3)] (d) at (4){$d$};
\node[shift=(180:0.25)] (e) at (5){$e$};
\node[shift=(0:0.25)] (f) at (6){$f$};

\node[] (T) at (90:-1.){$\T_G$};

\end{scope}

\end{tikzpicture}
\end{center}

\caption{The $3$-sun graph $G$ on the left, a tree intersection model $\mathbf{M}^{\sf T}_G$ of $G$ in the center and a $\mathcal{P}_{\sf int}$-free tree-layout of $G$ on the right. \label{fig_intersection_treelayout}}
\end{figure}

We first argue that $\T_G=(T,r,\rho_G)$ is a tree-layout of $G$, that is for every edge $xy\in E$, either $x\prec_{\T_G} y$ or $y\prec_{\T_G} x$. To see this, it suffices to observe that as the subtree $T^x$ intersects the subtree $T^y$, then either $\rho_G(x)$ is an ancestor of $\rho_G(y)$ in $T'$ or vice-versa.

Consider three vertices $x$, $y$ and $z$ such that $x\prec_{\T_G} y \prec_{\T_G} z$ and suppose that $xz\in E$. The fact that $x\prec_{\T_G} y \prec_{\T_G} z$ implies that either $\rho'(y)$ is a descendant in $T'$ (when rooted at $r$) of $\rho'(x)$ or $\rho'(x)=\rho'(y)$ and that either $\rho'(y)$ is an ancestor of $\rho'(z)$ or $\rho'(y)=\rho'(z)$. Since $xz\in E$, then the subtrees $T^x$ and $T^z$ intersect. This implies $T^x$ contains the path of $T'$ from $\rho'(x)$ to $\rho'(z)$, and thereby contains the node $\rho'(y)$. It follows that $xy\in E$, implying that the pattern on $x$, $y$ and $z$ is neither isomorphic to $\langle\overline{12},{13},{23}\rangle$ nor to $\langle\overline{12},{13},\overline{23}\rangle$.

\medskip
\noindent
$\Leftarrow$ Suppose that $\T_G=(T,r,\rho_G)$ is a $\mathcal{P}_{\sf int}$-free tree-layout of $G$. For every vertex $x$, we define the subset of vertices $\delta(x)=\{y\in D_{\T}(x)\mid xy\in E \wedge \forall z, y\prec_{\T_G} z, xz\notin E\}$. Then we build a tree-intersection model $\mathbf{M}^{\sf T}_G=(T,\mathcal{T},\tau_G)$ of $G$ as follows: for every $x\in V$, $\tau_G(x)$ is the smallest subtree of $T$ containing the vertices of $\delta(x)\cup\{x\}$. It is straightforward to see that if $xy\in E$ then $\tau_G(x)$ and $\tau_G(y)$ intersect. Suppose the $\tau_G(x)$ and $\tau_G(y)$ intersect. Without loss of generality we can assume that $\rho_G(x)\prec_{\T_G} \rho_G(y)$ and thereby  $\rho_G(y)$ belongs to $\tau_G(x)$. For the sake of contradiction, suppose that $xy\notin E$. By definion of $\tau_G(x)$, there exists a vertex $z$ such that $\rho_G(y)\prec_{\T_G} \rho_G(z)$ and $xz\in E$. In turn, by the condition of Theorem~\ref{th_chordal_tree} we have that $xy\in E$: contradiction.
\end{proof}

We conclude the discussion on chordal graphs by noticing the observation that chordal graphs are also characterized as $\frak{T}(\mathcal{P}_{\sf chordal})$. This follows from \autoref{lem_chordal} below and the observation that $\mathcal{P}_{\sf chordal}\subset \mathcal{P}_{\sf int}$. So we obtain that $\frak{T}(\mathcal{P}_{\sf chordal})=\frak{T}(\mathcal{P}_{\sf int})$. Hence for a fixed set of forbidden patterns, moving from layouts to tree-layouts does not always produce a larger class of graphs.

\begin{lemma} \label{lem_chordal}
Let $G$ be a chordal graph. Every $\mathcal{P}_{\sf chordal}$-free tree-layout $\T_G=(T,r,\rho_G)$ of  $G$ can be transformed into a $\mathcal{P}_{\sf chordal}$-free layout $\L_G=(P,r,\rho'_G)$ of $G$.
\end{lemma}
\begin{proof}
The mapping $\rho'_G$ is obtained by a depth-first search ordering of the nodes of $T$. That is $x\prec_{\L_G} y$ if and only if $\rho_G(x)$ is visited before $\rho_G(y)$ in such an ordering. Now consider three vertices $x$, $y$ and $z$ such that $x\prec_{\L_G} y\prec_{\L_G} z$ and such that $xz\in E$. To prove that $\L_G$ is $\mathcal{P}_{\sf chordal}$-free, we show that if $yz\in E$, then $xy\in E$. As $xz\in E$, we have that $x\prec_{\T_G} z$. Similarly if $yz\in E$, then we have $y\prec_{\T_G} z$. If we also have $x\prec_{\T_G} y$, then since $\T_G$ is $\mathcal{P}_{\sf chordal}$-free, $yx\in E$. So suppose that $x\not\prec_{\T_G} y$. But since $x\prec_{\L_G} y$, $\rho_G(y)$ is visited after $\rho_G(x)$ and all its descendents, including $\rho_G(y)$: a contradiction.
\end{proof}

\subsection{Proper interval graphs and proper chordal graphs}

\paragraph{Proper interval graphs.}

A graph $G=(V,E)$ is a \emph{proper interval graph} if it is an interval graph admitting an interval model $\mathbf{M}^{\sf I}_G$ such that for every pair of intervals none is a subinterval of another~\cite{Roberts68,Roberts69}. Proper interval graphs are also characterized as $K_{1,3}$-free interval graphs or as \emph{unit interval graphs}, that are interval graphs admitting an interval model $\mathbf{M}^{\sf I}_G$ in which every interval has the same length. 
Let us consider the following set of patterns:
\[
\mathcal{P}_{\sf indifference}=\big\{\langle\overline{12},{13},{23}\rangle,\langle\overline{12},{13},\overline{23}\rangle,\langle{12},{13},\overline{23}\rangle\big\}.
\]
A  $\mathcal{P}_{\sf indifference}$-free layout is called an \emph{indifference layout}. Indifference layouts have several characterizations:

\begin{theorem} \label{lem_indifference_layout}   \cite{LoogesO93optim,Roberts68}
Let $\L_G$ be a layout of a graph $G$. The following properties are equivalent. 
\begin{enumerate}
\item $\L_G$ is an indifference layout;
\item for every vertex $v$, $N[v]$ is consecutive in $\L_G$;
\item every maximal clique is consecutive in $\L_G$;
\item for every pair of vertices $x$ and $y$ with $x\prec_{\L_G} y$, $N(y)\cap A_{\L_G}(x)\subseteq N(x)\cap A_{\L_G}(x)$ and $N(x)\cap D_{\L_G}(y)\subseteq N(y)\cap D_{\L_G}(y)$.
\end{enumerate}
Moreover $G$ has an indifference layout if and only if it is a proper interval graph.
\end{theorem}

The following lemma shows that $\mathcal{P}_{\sf indifference}$ is not a minimal set of patterns to characterize proper interval graphs, yielding a proof of \autoref{th_proper_int} below.

\begin{lemma} \label{lem_minimal_proper}
Let $\L_G$ be a $\big\{\langle\overline{12},{13},{23}\rangle,\langle{12},{13},\overline{23}\rangle\big\}$-free layout of a connected graph $G$. Then $\L_G$ an indifference layout.
\end{lemma}
\begin{proof}
For the sake of contradiction, suppose that $\L_G$ contains the pattern $\langle\overline{12},{13},\overline{23}\rangle$. 
So there exists a triple  $x,y,z$ of vertices such that $x \prec_{\L_G} y \prec_{\L_G} z$, $xz\in E$, $xy\notin E$ and $yz\notin E$. 
Let us consider such a triple that minimizes the set $\{v\in V\mid x \prec_{\L_G} v \prec_{\L_G} z\}$ and call it a \emph{bad} triple.

Since $G$ is connected, it contains a shortest $x,y$-path $P$. We first observe that every internal vertex $v$ of $P$ satisfies $x \prec_{\L_G} v \prec_{\L_G} y$. Let $v$ be the leftmost vertex of $P$ in $\L_G$.  Assume that $v\prec_{\L_G} x$. Then $v$ has two neighbors $a$ and $b$ in $P$ and at least one of them is distinct from $x$ and $y$. Since $\L_G$ excludes $\langle{12},{13},\overline{23}\rangle$, $v\prec_{\L_G} a$ and $v\prec_{\L_G} b$, $a$ and $b$ are adjacent vertices, contradicting $P$ to be a shortest $x,y$-path. Thereby the leftmost vertex $v$ of $P$ in $\L_G$ is $x$. For a symmetric argument, since $\L_G$ excludes $\langle \overline{12},{13},{23}\rangle$, the rightmost vertex of $P$ is $y$.

Since $xy\notin E$, $P$ contains at least one vertex. Consider $v$ the unique neighbour of $x$ on $P$, then $x\prec_{\L_G} v\prec_{\L_G} y$. Since $\L_G$ excludes $\langle{12},{13},\overline{23}\rangle$, we have $vz\in E$, which in turns implies that $vy\notin E$. But then since $yz\notin E$, the triple $v,y,z$ forms a $\langle\overline{12},{13},\overline{23}\rangle$ pattern. Since  $x\prec_{\L_G} v\prec_{\L_G} y$, this contradicts the choice of $x,y,z$ as a bad triple.
\end{proof}

\begin{theorem} \label{th_proper_int} \cite{Damaschke90}
The class of proper interval graphs is $\frak{L}(\mathcal{P}_{\sf proper})$ where  
\[
\mathcal{P}_{\sf proper}=\big\{\langle\overline{12},{13},{23}\rangle,\langle{12},{13},\overline{23}\rangle\big\}.
\]
\end{theorem}

\paragraph{Proper chordal graphs.}
From the discussion of \autoref{sec_chordal}, we would like to understand what are the \emph{''tree-like"} proper interval graphs. In~\cite{Gavril74}, Gavril already considered this question, but didn't provide a conclusive answer. He observed that any family $\mathcal{T}$ of subtrees of a tree $T$ can be transformed into a proper family $\mathcal{T}'$ (no subtree in $\mathcal{T}'$ is a subtree of another one) by adding a private pendant leaf to every subtree of $\mathcal{T}$. Clearly the intersection graph of $\mathcal{T}$ is the same as the one of $\mathcal{T}'$. It follows that considering a proper family of subtrees also characterizes chordal graphs. To cope with \emph{"tree-like"} proper interval graphs, we rather propose the following definition.

\begin{definition}
A graph $G=(V,E)$ is a proper chordal graph if $G\in\frak{T}(\mathcal{P}_{\sf proper})$.
\end{definition}

As a consequence of the fact that $\mathcal{P}_{\sf proper}\subset \mathcal{P}_{\sf indifference}$ and of \autoref{lem_minimal_proper}, one can observe that a $\frak{T}(\mathcal{P}_{\sf proper})=\frak{T}(\mathcal{P}_{\sf indifference})$. Thereby hereafter,
a $\mathcal{P}_{\sf proper}$-free tree-layout $\T_G$ of a graph $G$ will be called an \emph{indifference tree-layout}. 
We first prove that \autoref{lem_indifference_layout} generalizes to indifference tree-layouts.

\begin{theorem} \label{lem_indifference_treelayout}
Let $\T_G=(T,r,\rho_G)$ be a tree-layout of a graph $G$. The following properties are equivalent. 
\begin{enumerate}
\item $\T_G$ is an indifference tree-layout;
\item for every vertex $x$, the vertices of $N[x]$ induces a connected subtree of $T$;
\item for every maximal clique $K$, the vertices of $K$ appear consecutively on a path from $r$ in $T$;
\item for every pair of vertices $x$ and $y$ such that $x\prec_{\T_G} y$, $N(y)\cap A_{\T_G}(x)\subseteq N(x)\cap A_{\T_G}(x)$ and $N(x)\cap D_{\T_G} (y)\subseteq N(y)\cap D_{\T_G}(y)$.
\end{enumerate}
\end{theorem}
\begin{proof}
$(1\Leftrightarrow 2)$ Suppose that $N[x]$ is not connected. Then, there exists $z\in N[x]$ such that the path from $x$ to $z$ in $T$ contains a vertex $y\notin N[x]$. Observe that as $xz\in E$, we have either $x\prec_{\T_G} y\prec_{\T_G} z$ or $z\prec_{\T_G} y\prec_{\T_G} x$, with $xz\in E$ but $xy\notin E$. But then, $\T$ is not $\mathcal{P}_{\sf proper}$-free and hence not an indifference tree-layout.

If $\T_G$ is not an indifference tree-layout, then it contains a pattern of $\mathcal{P}_{\sf proper}$, namely three vertices $x$, $y$ and $z$ such that $x\prec_{\T} y\prec_{\T} z$, $xz\in E$ but either $xy\notin E$ or $yz\not\in E$. In the former case, $N[x]$ is not connected in $T$, in the latter case, $N[z]$ is not connected in $T$.

\smallskip
\noindent
$(1\Leftrightarrow 3)$ 
Suppose that $\T_G$ is an indifference tree-layout. Let $K$ be a maximal clique of $G$. First observe that if $x$ and $y$ are vertices of $K$, then either $x\prec_{\T} y$ or $y\prec_{\T} x$. So let $x$ be the vertex of $K$ that is the closest to the root $r$ and $y$ the vertex of $K$ that is the furthest from $r$. Let $z$ be a vertex of $G$ such that $x\prec_{\T_G} z\prec_{\T_G} y$. Since $\T_G$ is a $\big\{\langle \overline{12},{13},\overline{23}\rangle, \langle\overline{12},{13},{23}\rangle,\langle{12},{13},\overline{23}\rangle\big\}$-free tree-layout (2), $xz\in E$ and $yz\in E$. Now suppose that there exists $z'\neq z$ such that $x\prec_{\T} z'\prec_{\T} y$, by the same argument we have $xz'\in E$ and $yz'\in E$. Since $\T_G$ is $\mathcal{P}_{\sf proper}$-free, we observe again that $zz'\in E$. It follows, by the choice of $x$ and $y$ that $K=\{z\in V(G)\mid x\prec_{\T_G} z\prec_{\T_G} y\}$, proving the statement.

For the converse, let us consider three vertices $x$, $y$ and $z$ such that $x\prec_{\T_G} y\prec_{\T_G} z$ and $xz\in E$.
Then there exists a maximal clique $K$ containing the edge $xz$. By assumption, the vertices of $K$ are consecutive on a path from $r$. As $K$ contains $x$ and $z$, it also contains $y$. Thereby, we have that $xy\in E$ and $yz\in E$. It follows that $\T_G$ is $\mathcal{P}_{\sf indifference}$-free and so an indifference tree-layout.

\smallskip
\noindent
$(1\Leftrightarrow 4)$ 
Suppose that $\T_G$ is an indifference tree-layout. Let $x$ and $y$ be two vertices such that $x\prec_{\T_G} y$. Suppose that there exists $z$ such that $z\prec_{\T} x$ and $zy\in E$. As $\T_G$ is $\big\{\langle \overline{12},{13},\overline{23}\rangle, \langle\overline{12},{13},{23}\rangle,\langle{12},{13},\overline{23}\rangle\big\}$-free (2),
  $xz\in E$, which implies that $N(y)\cap A_{\T}(x)\subseteq N(x)\cap A_{\T}(x)$. Similarly, suppose that there exists $z$ such that $y\prec_{\T} z$ and $xz\in E$. As $\T$ is 
$\big\{\langle \overline{12},{13},\overline{23}\rangle, \langle\overline{12},{13},{23}\rangle,\langle{12},{13},\overline{23}\rangle\big\}$-free,
 $yz\in E$, which implies that $N(x)\cap D_{\T} (y)\subseteq N(y)\cap D_{\T}(y)$.

For the converse, let us consider three vertices $x$, $y$ and $z$ such that $x\prec_{\T_G} y\prec_{\T_G} z$ and $xz\in E$.
As $N(z)\cap A_{\T_G}(y)\subseteq N(y)\cap A_{\T_G}(y)$, $xz\in E$ implies that $yx\in E$. Similarly, as $N(x)\cap D_{\T}(y)\subseteq N(y)\cap D_{\T}(y)$, $xz\in E$ implies that $yz\in E$. It follows that $\T_G$ is $\mathcal{P}_{\sf indifference}$-free and so an indifference tree-layout.
  %
%
\end{proof}

Clearly, as an indifference layout is an indifference tree-layout, every proper interval graph is a proper chordal graph. Also, as $\mathcal{P}_{\sf chordal}\subset \mathcal{P}_{\sf proper}$, proper chordal graphs are chordal graphs. The following observation proves that not every chordal graph is a proper chordal graph. For $k\geq 3$, the $k$-sun is a graph on $2k$ vertices $X\cup Y$ with $X=\{x_1,\dots x_k\}$ and $Y=\{y_1,\dots y_k\}$ such that $X$ is a clique and $Y$ is an independent set and for every $i\in [k]$, $y_i$ is
 adjacent to $x_i$ and $x_{(i+1)\mod k}$ (see  \autoref{fig_intersection_treelayout}).
 
\begin{observation} \label{obs_ksun}
For every, $k\geq 3$, the $k$-sun is not a proper chordal graph.
\end{observation}
\begin{proof}
We prove the claim for $k=3$. Its generalization to arbitrary value of $k\geq 3$ is straightforward.

Let $G$ be the $3$-sun with the vertex set as in \autoref{fig_intersection_treelayout}. Suppose that $\T_G=(T,r,\rho_G)$ is an indifference tree-layout of $G$. Let $b$, $c$ and $d$ be the three non-simplicial vertices of the $3$-sun $G$. As they form a maximal clique, these three vertices are mapped to consecutive nodes of the same path from the root to a leaf of $\T$ (see \autoref{lem_indifference_treelayout}). Suppose without loss of generality that $b\prec_{\T_G} c\prec_{\T_G} d$. Let $e$ be the simplicial vertex of $G$ adjacent to $b$ and $d$. As $b$, $c$ and $d$ are consecutive in $\T_G$, either $e\prec_{\T_G} b$ or $d\prec_{\T_G} e$. We observe that as $e$ is non-adjacent to $c$, in both cases $\T_G$ is not  $\mathcal{P}_{\sf proper}$-free: contradiction.
\end{proof}

\subsection{Relationship between proper chordal graphs and subclasses of chordal graphs.}

Let us now discuss the relationship between proper chordal graphs and important subclasses of chordal graphs, namely interval graphs, directed path graphs and strongly chordal graphs. \autoref{fig_classes} summarizes these relations.

\begin{figure}[h]
\begin{center}
\begin{tikzpicture}
\tikzstyle{sommet}=[circle, draw, fill=black, inner sep=0pt, minimum width=2pt]


\begin{scope}[rotate=275,yshift=1.7cm,xshift=-0cm]
  \def\eggheight{2.8cm}
  \path[draw, circle, fill=gray!20]
  plot[domain=-pi:pi,samples=100]
  ({.5*\eggheight *cos(\x/3.5 r)*sin(\x r)},{-\eggheight*(cos(\x r))})
  -- cycle;
\node[rotate=15, very thick] at (0.8,1) {\tiny  proper chordal};
\end{scope}

\begin{scope}[rotate=270,yshift=0.1cm,xshift=0.1cm]
  \def\eggheight{1cm}
  \path[draw, circle]
  plot[domain=-pi:pi,samples=100]
  ({.75*\eggheight *cos(\x/4 r)*sin(\x r)},{-\eggheight*(cos(\x r))})
  -- cycle;
  
\node[] at (0,0) {\tiny proper interval};

\end{scope}

\begin{scope}[rotate=270,yshift=0.5cm,xshift=-0.4cm]
  \def\eggheight{1.9cm}
  \path[draw, circle]
  plot[domain=-pi:pi,samples=100]
  ({.8*\eggheight *cos(\x/4 r)*sin(\x r)},{-\eggheight*(cos(\x r))})
  -- cycle;

\node[] at (-1.2,-0.2) {\tiny interval};

\begin{scope}[rotate=120,scale=0.3,xshift=-2cm,yshift=3.5cm]
\foreach \i in {0,36,72,108,144,180}{
	\draw (\i:1.2) node[sommet]{}; 
}
\foreach \i in {0,36,72,108,144}{
	\draw (\i:1.2) -- (30+\i:1.2);
}
\draw (0:0) node[sommet]{}; 
\foreach \i in {0,36,72,108,144,180}{
	\draw (\i:1.2) -- (0:0);
}
\end{scope}

\begin{scope}[rotate=90,scale=0.3,xshift=-0.5cm,yshift=1cm]
\foreach \i in {0,45,90,135,180}{
	\draw (\i:1.2) node[sommet]{}; 
}
\foreach \i in {0,45,90,135}{
	\draw (\i:1.2) -- (45+\i:1.2);
}
\draw (0:0) node[sommet]{}; 
\foreach \i in {0,45,90,135,180}{
	\draw (\i:1.2) -- (0:0);
}
\draw (45:1.2) -- (135:1.2); 
\end{scope}

\begin{scope}[xshift=-0.2cm,yshift=1.1cm,scale=0.3]
\draw (0,0) node[sommet]{};
\draw (90:1) node[sommet]{};
\draw (210:1) node[sommet]{};
\draw (330:1) node[sommet]{};
	 
\draw (0,0) -- (90:1);
\draw (0,0) -- (210:1);
\draw (0,0) -- (330:1);
\end{scope}
\end{scope}

\begin{scope}[rotate=270,yshift=.95cm,xshift=-0.60cm]
  \def\eggheight{2.75cm}
  \path[draw, circle]
  plot[domain=-pi:pi,samples=100]
  ({.78*\eggheight *cos(\x/4 r)*sin(\x r)},{-\eggheight*(cos(\x r))})
  -- cycle;

\node[] at (-1.6,-0.4) {\tiny rooted directed path};

\begin{scope}[xshift=-0.1cm,yshift=1.9cm,scale=0.25,rotate=14]
\draw (0,0) node[sommet]{};
\draw (90:1) node[sommet]{};
\draw (210:1) node[sommet]{};
\draw (330:1) node[sommet]{};
\draw (90:2) node[sommet]{};
\draw (210:2) node[sommet]{};
\draw (330:2) node[sommet]{};
	 
\draw (0,0) -- (90:1) -- (90:2);
\draw (0,0) -- (210:1) -- (210:2);
\draw (0,0) -- (330:1) -- (330:2);
\end{scope}

%

\begin{scope}[xshift=-1.25cm,yshift=1.15cm,scale=0.25,rotate=80]
\foreach \i in {0,36,72,108,144,180}{
	\draw (\i:1.2) node[sommet]{}; 
}
\foreach \i in {0,36,72,108,144}{
	\draw (\i:1.2) -- (30+\i:1.2);
}
\draw (0:0) node[sommet]{}; 
\foreach \i in {0,36,72,108,144,180}{
	\draw (\i:1.2) -- (0:0);
}
\draw (-90:1.2) node[sommet]{}; 
\draw (0:2.4) node[sommet]{}; 
\draw (180:2.4) node[sommet]{}; 

\draw (-90:1.2) -- (0:0);
\draw (0:2.4) -- (0:1.2) ; 
\draw (180:2.4) -- (180:1.2) ; 

\end{scope}

\end{scope}

\begin{scope}[rotate=270,yshift=1.4cm,xshift=-0.9cm]
  \def\eggheight{3.5cm}
  \path[draw, circle]
  plot[domain=-pi:pi,samples=100]
  ({.78*\eggheight *cos(\x/4 r)*sin(\x r)},{-\eggheight*(cos(\x r))})
  -- cycle;

\node[] at (-2.1,-0.3) {\tiny strongly chordal};

\begin{scope}[xshift=0.7cm,yshift=2.6cm,scale=0.24,rotate=-30]

\draw (0:0) node[sommet]{}; 
\foreach \i in {60,120,180,270,330}{
	\draw (\i:1.2) node[sommet]{}; 
	\draw (0,0) -- (\i:1.2);
}

\draw (60:1.2) -- (120:1.2)
	 -- (180:1.2)
	 -- (270:1.2)
	 -- (330:1.2);

\begin{scope}[xshift=-1.2cm]
\foreach \i in {270,210}{
	\draw (\i:1.2) node[sommet]{}; 
	\draw (0,0) -- (\i:1.2);
}
\draw (270:1.2) -- (210:1.2);
\draw (270:1.2) -- (0:1.2);

\end{scope}
\end{scope}

\begin{scope}[xshift=-1.3cm,yshift=1.6cm,scale=0.25,rotate=-65]
\draw (0,0) node[sommet]{};
\draw (90:1) node[sommet]{};
\draw (210:1) node[sommet]{};
\draw (330:1) node[sommet]{};
\draw (90:2) node[sommet]{};
\draw (210:2) node[sommet]{};
\draw (330:2) node[sommet]{};
	 
\draw (90:1) -- (210:1);
\draw (210:1) -- (330:1);
\draw (330:1) -- (90:1);
\draw[] (0,0) .. controls (50:1) .. (90:2) ;
\draw[] (0,0) .. controls (170:1) .. (210:2) ;
\draw[] (0,0) .. controls (290:1) .. (330:2) ;

\draw (0,0) -- (90:1) -- (90:2);
\draw (0,0) -- (210:1) -- (210:2);
\draw (0,0) -- (330:1) -- (330:2);
\end{scope}

\end{scope}

\begin{scope}[rotate=270,yshift=1.9cm,xshift=-1.2cm]
  \def\eggheight{4.3cm}
  \path[draw, circle]
  plot[domain=-pi:pi,samples=100]
  ({.78*\eggheight *cos(\x/4 r)*sin(\x r)},{-\eggheight*(cos(\x r))})
  -- cycle;
\node[] at (-2.7,-0.3) {\tiny chordal};

\begin{scope}[xshift=1cm,yshift=3.2cm,scale=0.3,rotate=-30]
\draw (90:1) node[sommet]{};
\draw (210:1) node[sommet]{};
\draw (330:1) node[sommet]{};
\draw (270:0.5) node[sommet]{};
\draw (150:0.5) node[sommet]{};
\draw (30:0.5) node[sommet]{};

\draw (90:1) -- (210:1);
\draw (90:1) -- (330:1);
\draw (210:1) -- (330:1);
\draw (150:0.5) -- (30:0.5);
\draw (150:0.5) -- (270:0.5);
\draw (270:0.5) -- (30:0.5);
\end{scope}

\end{scope}

\end{tikzpicture}
\end{center}
\caption{\label{fig_classes} Relationship between proper chordal graphs and subclasses of chordal graphs.}
\end{figure}
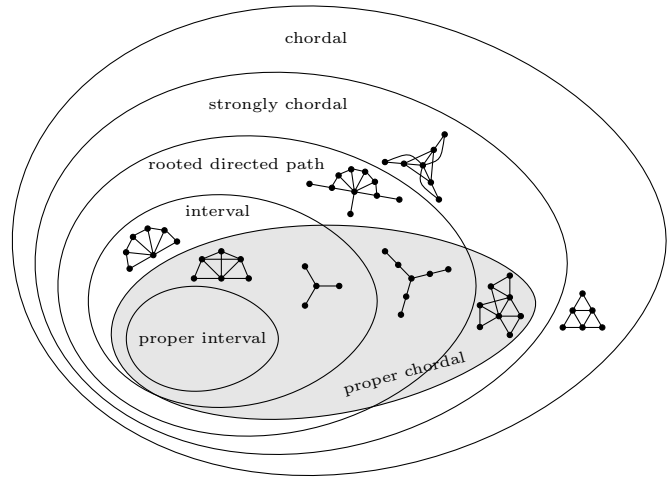

\paragraph{Interval graphs.} Trees are proper chordal graphs but not interval graphs. For $k\geq 3$, the \emph{$k$-fan} is the graph obtained by adding a universal vertex to the path on $k+1$ vertices (see \autoref{fig_6fan}). It is easy to see that for every $k\geq 3$, the $k$-fan is an interval graph. By \autoref{obs_interval}, we show that the $5$-fan is not a proper chordal graphs. It follows that proper chordal graphs and interval graphs form incomparable graph classes.

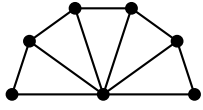
\begin{figure}[ht]
\begin{center}
\begin{tikzpicture}[thick,scale=1]
\tikzstyle{sommet}=[circle, draw, fill=black, inner sep=0pt, minimum width=4pt]

\foreach \i in {0,36,72,108,144,180}{
	\draw (\i:1.2) node[sommet]{}; 
}
\foreach \i in {0,36,72,108,144}{
	\draw (\i:1.2) -- (36+\i:1.2);
}
\draw (0:0) node[sommet]{}; 
\foreach \i in {0,36,72,108,144,180}{
	\draw (\i:1.2) -- (0:0);
}
\end{tikzpicture}
\end{center}
\caption{\label{fig_6fan} The $5$-fan.}
\end{figure}

\begin{observation} \label{obs_interval} 
The $k$-fan, for $k\geq 5$,  is an interval graph but not a proper chordal graph.
\end{observation}
\begin{proof}
First, for every $k$, the $k$-fan is an interval graph. As proper chordal graphs is an hereditary class, it suffices to prove that the $5$-fan is not proper chordal. We first prove that the $4$-fan $G$ has a unique indifference tree-layout $\T_G$.
Let $v$ be the universal vertex of $G$ and $v_1,\dots v_5$ be the vertices of the path appearing in the natural ordering. For $i\in[1,4]$, we let $C_i$ denote the maximal clique $\{v,v_i,v_{i+1}\}$. Suppose that $\T_G=(T,r,\rho_G)$ is an indifference tree-layout of $G$. By \autoref{lem_indifference_treelayout}, every maximal clique $C_i$  has to occur consecutively on a path of $T$ from $r$ to some leaf. Observe that $C_1\cap C_2$ and $C_3\cap C_3$ implies that $v_2\prec_{\T_G} v\prec_{\T_G} v_3$ or that $v_3\prec_{\T_G} v\prec_{\T_G} v_2$. Similarly, $C_2\cap C_3$ and $C_3\cap C_4$ implies that $v_3\prec_{\T_G} v\prec_{\T_G} v_4$ or that $v_4\prec_{\T_G} v\prec_{\T_G} v_3$. It follows that $v_3\prec_{\T_G} v$, $v\prec_{\T_G} v_2$ and $v\prec_{\T_G} v_4$ but $v_2$ and $v_4$ are not on a common path from the root of $T$. In turn, this forces $v_2\prec_{\T_G} v_1$ and $v_4\prec_{\T_G} v_5$ and this is the unique indifference tree-layout of the $4$-fan $G$. 

Now a $5$-fan $H$ can be obtained from  the $4$-fan $G$ by adding a vertex $v_6$ adjacent to $v_5$ and $v$, forming a new maximal clique $C_5=\{v,v_5,v_6\}$. Clearly, $H$ is a proper chordal graph if and only if $\T_G$ could be extended to an indifference tree-layout of $H$ where $C_5$ is consecutive, which is not possible since in $\T_G$, the vertices $v$ and $v_5$ are not consecutive.
\end{proof}

\paragraph{Rooted directed path graphs.} Rooted directed path graphs form an interesting graph class that is sandwiched between interval graphs and chordal graphs. A graph $G=(V,E)$ is a \emph{directed path graph} if and only if it is the intersection graph of a set of directed subpaths (a subpath from the root to a leaf) of a rooted tree~\cite{MonmaW86}.

\begin{figure}[ht]
\begin{center}
\begin{tikzpicture}[thick,scale=1]
\tikzstyle{sommet}=[circle, draw, fill=black, inner sep=0pt, minimum width=4pt]

\draw (0:0) node[sommet]{}; 
\foreach \i in {60,120,180,270,330}{
	\draw (\i:1.2) node[sommet]{}; 
	\draw (0,0) -- (\i:1.2);
}

\node[right] (a) at (0:0.1) {$a$};
\node[right] (b) at (60:1.3) {$b$};
\node[left] (c) at (120:1.3) {$c$};
\node[left] (d) at (170:1.3) {$d$};
\node[below] (e) at (270:1.3) {$e$};
\node[right] (f) at (330:1.3) {$f$};

\draw (60:1.2) -- (120:1.2)
	 -- (180:1.2)
	 -- (270:1.2)
	 -- (330:1.2);

\begin{scope}[xshift=-1.2cm]
\foreach \i in {270,210}{
	\draw (\i:1.2) node[sommet]{}; 
	\draw (0,0) -- (\i:1.2);
}
\draw (270:1.2) -- (210:1.2);
\draw (270:1.2) -- (0:1.2);

\node[below] (g) at (270:1.3) {$g$};
\node[left] (h) at (210:1.3) {$h$};
\end{scope}

\begin{scope}[xshift=5cm,yshift=-1cm]

\draw (90:3) -- (90:2) -- (90:1) -- (0:0) -- (210:1) -- (239:1.7) ;
\draw (0:0) -- (330:1) -- (301:1.7) ;

\node[right] (h) at (90:3) {$h$};
\node[right] (g) at (90:2) {$g$};
\node[right] (d) at (90:1) {$d$};
\node[right] (a) at (30:0.1) {$a$};
\node[left] (c) at (210:1) {$c$};
\node[left] (b) at (240:1.7) {$b$};
\node[right] (e) at (330:1) {$e$};
\node[right] (f) at (300:1.7) {$f$};

\draw[red,thick] (90:3) .. controls (0.25,2.5) .. (90:2) ;
\draw[red,thick] (90:3) .. controls (0.8,2) .. (90:1) ;
\draw[red,thick] (90:2) .. controls (0.25,1.5) .. (90:1) ;
\draw[red,thick] (0,0) .. controls (0.25,0.5) .. (90:1) ;
\draw[red,thick] (0,0) .. controls (-0.8,1) .. (90:2) ;
\draw[red,thick] (0,0) .. controls (-0.8,1) .. (90:2) ;
\draw[red,thick] (0,0) .. controls (-0.8,1) .. (90:2) ;

\draw[red,thick] (0,1) .. controls (-0.6,0.2) .. (210:1) ;
\draw[red,thick] (0,1) .. controls (0.6,0.2) .. (330:1) ;
\draw[red,thick] (0,0) .. controls (-0.2,-0.4) .. (210:1) ;
\draw[red,thick] (0,0) .. controls (0.2,-0.4) .. (330:1) ;

\draw[red,thick] (0,0) .. controls (-0.2,-1) .. (239:1.7) ;
\draw[red,thick] (0,0) .. controls (0.2,-1) .. (301:1.7) ;
\draw[red,thick] (210:1) .. controls (-0.6,-0.9) .. (239:1.7) ;
\draw[red,thick] (330:1) .. controls (0.6,-0.9) .. (301:1.7) ;

\draw (90:3) node[sommet]{};
\draw (90:2) node[sommet]{};
\draw (90:1) node[sommet]{};
\draw (0:0) node[sommet]{};
\draw (210:1) node[sommet]{};
\draw (239:1.7) node[sommet]{};
\draw (330:1) node[sommet]{};
\draw (301:1.7) node[sommet]{};

\end{scope}

\end{tikzpicture}
\end{center}
\caption{\label{fig_cevenol} The C\'evenol graph and an indifference tree-layout of the Cévenol graph.}
\end{figure}
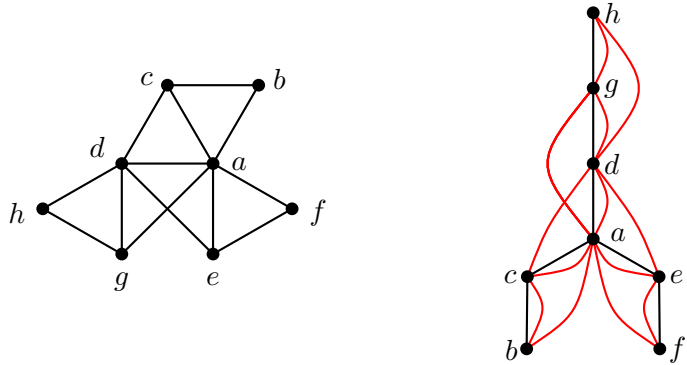

\begin{observation}
The C\'evenol graph is a proper chordal graph but not a rooted directed path graph
\end{observation}
\begin{proof}
The tree-layout of \autoref{fig_cevenol} certifies that the C\'evenol graph is a proper chordal graph. To see that the C\'evenol graph is not a rooted directed path graph, it suffices to observe that each of its tree-intersection models has a subtree with a branching node (either the subtree associated to $a$ or the one associated to $d$).
\end{proof}

\paragraph{Strongly chordal graphs.}
A vertex $v$ of a graph $G=(V,E)$ is \emph{simple} if for all $x,y\in N[v]$, $N[x]\subseteq N[y]$ or $N[y]\subseteq N[x]$. Observe that a simple vertex is a simplicial vertex.  A simplicial elimination ordering $\L$ is a \emph{strong perfect elimination ordering} if for every $w$, $x$, $y$, $z$ such that $w\infL x\infL y\infL z$, if $wy\in E$, $wz\in E$ and $xy\in E$ then $xz\in E$. A graph is strongly chordal if and only if it has a strong perfect elimination ordering~\cite{Farber83chara}.
Strongly chordal graphs are also characterized as being the $k$-sun free chordal graphs. As a consequence of \autoref{obs_ksun}, we observe that every proper chordal graph is a strongly chordal graph. Moreover the inclusion between proper chordal graphs and strongly chordal is strict as the $5$-fan is strongly chordal, but not proper chordal (see \autoref{obs_interval}).


\section{\FPQ-trees and \FPQ-hierarchies}
\label{sec_FPQ}

Let $\mathcal{P}$ be a set of patterns. In general, a graph $G\in\frak{L}(\mathcal{P})$ admits several $\mathcal{P}$-free layouts. A basic example is the complete graph $K_{\ell}$ on $\ell$ vertices, which is a proper interval graph. It is easy to observe that every layout of $K_{\ell}$ is an indifference layout (i.e. a $\mathcal{P}_{\sf proper}$-free layout), but also a  $\mathcal{P}_{\sf int}$-free layout and a  $\mathcal{P}_{\sf chordal}$-free layout. Let us discuss in more details the case of proper interval graphs and interval graphs. 

\paragraph{Proper interval graphs.} Two vertices $x$ and $y$ of a graph $G$ are \emph{true-twins} if $N[x]=N[y]$. It is easy to see that the true-twin relation is an equivalence relation. If $G$ contains some true-twins, then, by \autoref{lem_indifference_layout}, the vertices of any equivalence class occurs consecutively (and in arbitrary order) in an indifference layout. It follows that for proper interval graphs, the set of indifference layouts depends on the true-twin equivalence classes. Indeed, a proper interval graph $G$ without any pair of true-twins has a unique (up to reversal) indifference layout. 

\paragraph{Interval graphs.} In the case of interval graphs, the set of intersection models, and hence of $\mathcal{P}_{\sf int}$-free layouts, is structured by means of \emph{modules}~\cite{Gallai67}, in a similar way to the true-twin equivalence classes for $\mathcal{P}_{\sf proper}$-free layouts. A subset $M$ of vertices of a graph $G$ is a \emph{module} if for every $x\notin M$, either $M\subseteq N(x)$ or $M\subseteq \overline{N}(x)$. Observe that a true-twin equivalence class is a module. A graph may have exponentially many modules. For example, every subset of vertices of the complete graph is a module. Hsu~\cite{Hsu95} proved that interval graphs having a unique intersection model are those without any trivial module.

The set of modules of a graph forms a so-called \emph{partitive family}~\cite{CheinHM81} and can thereby be represented through a linear size tree, called the \emph{modular decomposition tree}~(see \cite{HabibP10} for a survey on modular decomposition). To recognize interval graphs in linear time, Booth and Lueker~\cite{BoothL76}  introduced the concept of \PQ{}-trees which is closely related to the modular decomposition tree or more generally to the theory of (weakly-)partitive families~\cite{CheinHM81,CunninghamE80}. Basically, a \PQ{}-tree on a set $X$ is a labelled ordered tree having $X$ as its leaf set. Since every ordered tree defines a permutation of its leaf set, by defining an equivalence relation  based on the labels of the node, every \PQ{}-tree can be associated to a set of permutations of $X$.  In the context of interval graphs, $X$ is the set of maximal cliques and a \PQ{}-tree represents the set of so-called \emph{consecutive orderings of the maximal cliques} characterizing interval graphs.

\paragraph{Proper chordal graphs.} As shown by \autoref{fig_several_indifference_treelayouts}, a proper chordal graph can also have several indifference tree-layouts. In order to represent the set of $\mathcal{P}_{\sf proper}$-tree-layouts of a given graph, we will define a structure called  \emph{\FPQ{}-hierarchies}, based on \FPQ{}-trees~\cite{LiottaRT21}, a variant of \PQ{}-trees.

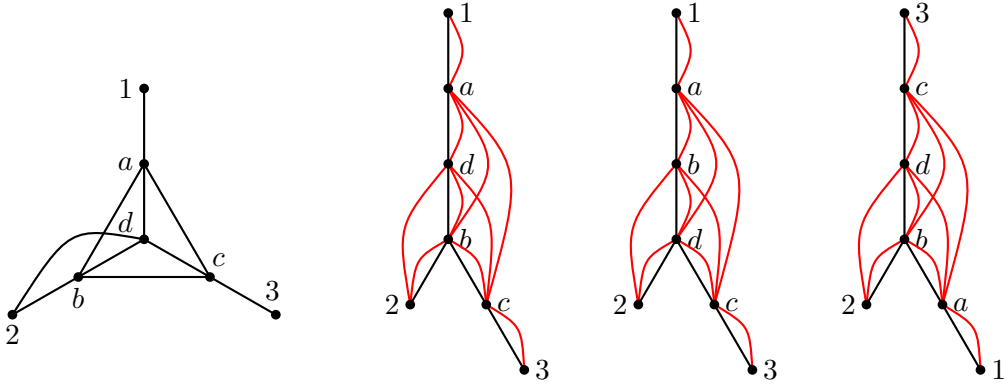
\begin{figure}[htbh]
\begin{center}
\begin{tikzpicture}[thick,scale=1]
\tikzstyle{sommet}=[circle, draw, fill=black, inner sep=1pt, minimum width=2pt]

\begin{scope}[xshift=0cm,yshift=0cm,scale=1,rotate=0]
\draw (0,0) node[sommet]{};
\draw (90:1) node[sommet]{};
\draw (210:1) node[sommet]{};
\draw (330:1) node[sommet]{};
\draw (90:2) node[sommet]{};
\draw (210:2) node[sommet]{};
\draw (330:2) node[sommet]{};

\node[left] (a) at (90:1) {$a$};
\node[left] (1) at (90:2) {$1$};
\node[below] (b) at (210:1) {$b$};
\node[below] (2) at (210:2) {$2$};
\node[right] (c) at (340:0.8) {$c$};
\node[right] (4) at (335:1.6) {$3$};
\node[above] (d) at (190:0.25) {$d$};

\draw (90:1) -- (210:1);
\draw (210:1) -- (330:1);
\draw (330:1) -- (90:1);
\draw[] (0,0) .. controls (170:1) .. (210:2) ;

\draw (0,0) -- (90:1) -- (90:2);
\draw (0,0) -- (210:1) -- (210:2) ;
\draw (0,0) -- (330:1) -- (330:2);

\end{scope}

\begin{scope}[xshift=4cm,yshift=0cm]
\draw (90:3) -- (90:2) -- (90:1) -- (90:0) -- (-60:1) -- (-60:2) ;
\draw (90:0) -- (-120:1) ;
\node[right] (1) at (90:3) {$1$};
\node[right] (a) at (90:2) {$a$};
\node[right] (b) at (90:1) {$d$};
\node[right] (b) at (90:0) {$b$};
\node[right] (c) at (-60:1) {$c$};
\node[right] (4) at (-60:2) {$3$};
\node[left] (2) at (-120:1) {$2$};

\draw[red,thick] (90:3) .. controls (0.25,2.5) .. (90:2) ;
\draw[red,thick] (90:2) .. controls (0.25,1.5) .. (90:1) ;
\draw[red,thick] (90:2) .. controls (0.7,1) .. (90:0) ;
\draw[red,thick] (90:2) .. controls (1,0.95) .. (300:1) ;
\draw[red,thick] (0,0) .. controls (0.25,0.5) .. (90:1) ;

\draw[red,thick] (90:1) .. controls (25:0.7) .. (-60:1) ;

\draw[red,thick] (0,0) .. controls (-150:0.5) .. (-120:1) ;
\draw[red,thick] (90:1) .. controls (165:0.7) .. (-120:1) ;

\draw[red,thick] (0,0) .. controls (-30:0.5) .. (-60:1) ;
\draw[red,thick] (-60:1) .. controls (-50:1.5) .. (-60:2) ;

\draw (90:3) node[sommet]{};
\draw (90:2) node[sommet]{};
\draw (90:1) node[sommet]{};
\draw (90:0) node[sommet]{};
\draw (-60:1) node[sommet]{};
\draw (-60:2) node[sommet]{};
\draw (-120:1) node[sommet]{};

\end{scope}

\begin{scope}[xshift=7cm,yshift=0cm]
\draw (90:3) -- (90:2) -- (90:1) -- (90:0) -- (-60:1) -- (-60:2) ;
\draw (90:0) -- (-120:1)  ;
\node[right] (1) at (90:3) {$1$};
\node[right] (a) at (90:2) {$a$};
\node[right] (b) at (90:1) {$b$};
\node[right] (b) at (90:0) {$d$};
\node[right] (c) at (-60:1) {$c$};
\node[right] (3) at (-60:2) {$3$};
\node[left] (2) at (-120:1) {$2$};

\draw[red,thick] (90:3) .. controls (0.25,2.5) .. (90:2) ;
\draw[red,thick] (90:2) .. controls (0.25,1.5) .. (90:1) ;
\draw[red,thick] (90:2) .. controls (0.7,1) .. (90:0) ;
\draw[red,thick] (90:2) .. controls (1,0.95) .. (300:1) ;
\draw[red,thick] (0,0) .. controls (0.25,0.5) .. (90:1) ;

\draw[red,thick] (90:1) .. controls (25:0.7) .. (-60:1) ;

\draw[red,thick] (0,0) .. controls (-150:0.5) .. (-120:1) ;
\draw[red,thick] (90:1) .. controls (165:0.7) .. (-120:1) ;

\draw[red,thick] (0,0) .. controls (-30:0.5) .. (-60:1) ;
\draw[red,thick] (-60:1) .. controls (-50:1.5) .. (-60:2) ;

\draw (90:3) node[sommet]{};
\draw (90:2) node[sommet]{};
\draw (90:1) node[sommet]{};
\draw (90:0) node[sommet]{};
\draw (-60:1) node[sommet]{};
\draw (-60:2) node[sommet]{};
\draw (-120:1) node[sommet]{};

\end{scope}

\begin{scope}[xshift=10cm,yshift=0cm]
\draw (90:3) -- (90:2) -- (90:1) -- (90:0) -- (-60:1) -- (-60:2) ;
\draw (90:0) -- (-120:1)    ;
\node[right] (1) at (90:3) {$3$};
\node[right] (a) at (90:2) {$c$};
\node[right] (b) at (90:1) {$d$};
\node[right] (b) at (90:0) {$b$};
\node[right] (c) at (-60:1) {$a$};
\node[right] (3) at (-60:2) {$1$};
\node[left] (2) at (-120:1) {$2$};

\draw[red,thick] (90:3) .. controls (0.25,2.5) .. (90:2) ;
\draw[red,thick] (90:2) .. controls (0.25,1.5) .. (90:1) ;
\draw[red,thick] (90:2) .. controls (0.7,1) .. (90:0) ;
\draw[red,thick] (90:2) .. controls (1,0.95) .. (300:1) ;
\draw[red,thick] (0,0) .. controls (0.25,0.5) .. (90:1) ;

\draw[red,thick] (90:1) .. controls (25:0.7) .. (-60:1) ;

\draw[red,thick] (0,0) .. controls (-150:0.5) .. (-120:1) ;
\draw[red,thick] (90:1) .. controls (165:0.7) .. (-120:1) ;

\draw[red,thick] (0,0) .. controls (-30:0.5) .. (-60:1) ;
\draw[red,thick] (-60:1) .. controls (-50:1.5) .. (-60:2) ;

\draw (90:3) node[sommet]{};
\draw (90:2) node[sommet]{};
\draw (90:1) node[sommet]{};
\draw (90:0) node[sommet]{};
\draw (-60:1) node[sommet]{};
\draw (-60:2) node[sommet]{};
\draw (-120:1) node[sommet]{};

\end{scope}

\end{tikzpicture}
\end{center}
\caption{\label{fig_several_indifference_treelayouts} A graph $G$ with three possible indifference tree-layouts, two of them rooted at vertex $1$, the third one at vertex $3$.}
\end{figure}

\subsection{\FPQ{}-trees.}


%

An \emph{\FPQ{}-tree} on the ground set $X$ is a labelled, ordered tree $\mathsf{T}$ such that its leaf set $\mathcal{L}(\mathsf{T})$ is mapped to $X$. The internal nodes of $\mathsf{T}$ are of three types, \textsf{F}-nodes, \textsf{P}-nodes, and \textsf{Q}-nodes. If $|X|=1$, then $\mathsf{T}$ is the tree defined by a leaf and a \textsf{Q}-node as the root. Otherwise,  \textsf{F}-nodes and \textsf{Q}-nodes have at least two children while \textsf{P}-nodes have at least three children.

Let $\mathsf{T}$ and $\mathsf{T'}$ be two \FPQ{}-trees. We say that $\mathsf{T}$ and $\mathsf{T}'$ are isomorphic if they are isomorphic as labelled trees. We say that $\mathsf{T}$ and $\mathsf{T'}$  are \emph{equivalent}, denoted $\mathsf{T}\equiv_{\FPQ{}}\mathsf{T'}$, if one can be turned into a labelled tree isomorphic to the other by a series of the following two operations: $\perm(u)$ which permutes in any possible way the children of a \textsf{P}-node $u$; and $\rev(u)$ which reverses the ordering of the children of a \textsf{Q}-node $u$. It follows that the equivalence class of an \FPQ{}-tree $\mathsf{T}$ on $X$ defines a set $\mathfrak{S}_{\FPQ{}}(\mathsf{T})$ of permutations of $X$.

Let $\mathfrak{S}$ be a subset of permutations of $X$. A subset $I\subseteq X$ is a \emph{factor} of $\mathfrak{S}$ if in every permutation of $\mathfrak{S}$, the elements of $I$ occur consecutively. It is well known that the set of factors of a set of permutations form a so called \emph{weakly-partitive family}~\cite{CheinHM81,CunninghamE80}. As a consequence, we have the following property, which was also proved in~\cite{LuekerB79}.

\begin{lemma} \label{lem_partitive} \cite{CheinHM81,CunninghamE80,LuekerB79}
Let $\mathfrak{S}_\mathsf{T}$ be the subset of permutations of a non-empty set $X$ associated to a \PQ{}-tree $\mathsf{T}$. Then a subset $I\subseteq X$ is a factor of $\mathfrak{S}_\mathsf{T}$ if and only if there exists an internal node $u$ of $\mathsf{T}$ such that 
\begin{itemize}
\item either $I=\mathcal{L}_{\mathsf{T}}(u)$; 
\item or $u$ is a \textsf{Q}-node and there exists a set of children $v_1,\dots, v_s$ of $u$ that are consecutive in $<_{T,u}$ and such that $I=\bigcup_{1\leq i\leq s} \mathcal{L}_{\mathsf{T}}(v_i)$.
\end{itemize}
We observe that $u=\lca_{\mathsf{T}}(I)$.
\end{lemma}

Given a set $\mathcal{S} \subseteq 2^{X}$ of subsets of the ground set $X$, we let ${\sf Convex}(\mathcal{S})$ denote the set of permutations of $X$ such that for every $S \in \mathcal{S}$, $S$ is a factor of ${\sf Convex}(\mathcal{S})$.

\begin{lemma}\label{lem_pq_algo} 
Let $X$ be a non-empty set and let ${\cal S} \subseteq 2^{X}$. In  linear time in $|\mathcal{S}|$, we can compute a \PQ{}-tree $\mathsf{T}$ on $X$ such that $\mathfrak{S}_{\mathsf{T}}= {\sf Convex}({\cal S})$ or decide that ${\sf Convex}({\cal S})=\emptyset$.
\end{lemma}
\begin{proof}
This follows from known recognition algorithms of convex bipartite graphs (see~\cite{HabibMPV00}).
\end{proof}

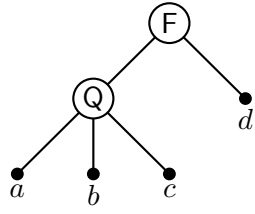
\begin{figure}[ht]
\begin{center}
\begin{tikzpicture}[thick,scale=1]
\tikzstyle{sommet}=[circle, draw, fill=black, inner sep=0pt, minimum width=4pt]          
\tikzstyle{Onode}=[circle, draw, fill=white, inner sep=0pt, minimum width=15pt]          
\tikzstyle{Pnode}=[circle, draw, fill=white, inner sep=0pt, minimum width=15pt]          
\tikzstyle{Qnode}=[circle, draw, fill=white, inner sep=0pt, minimum width=15pt]          

\begin{scope}[xshift=0cm,yshift=0cm]

\draw (0,0) -- (1,1);
\draw (1,1) -- (2,2);
\draw (1,1) -- (1,0);
\draw (1,1) -- (2,0);
\draw (2,2) -- (3,1);

\draw (0,0) node[sommet]{};
\draw (1,0) node[sommet]{};
\draw (2,0) node[sommet]{};
\draw (1,1) node[Qnode]{};
\draw (3,1) node[sommet]{};
\draw (2,2) node[Onode]{};

\node[below] (a) at (0,0) {$a$};
\node[below] (b) at (1,0) {$b$};
\node[below] (c) at (2,0) {$c$};
\node[below] (d) at (3,1) {$d$};
\node[] (o) at (2,2) {\textsf{F}};
\node[] (p) at (1,1) {\textsf{Q}};
\end{scope}

\end{tikzpicture}
\end{center}
\caption{\label{fig_OPQ} An \FPQ-tree $\mathsf{T}$ with $\mathfrak{S}_{\FPQ}(\mathsf{T})=\{abcd,cbad\}$. The set of non-trivial common factors of $\mathfrak{S}_\mathsf{T}$ is $\mathcal{I}=\{\{a,b,c\},\{a,b\},\{b,c\}\}$. But observe that ${\sf Convex}(\mathcal{I})$ also contains the permutations $dabc$ and $dcba$ as well.}
\end{figure}

A set $\mathcal{N} \subseteq 2^{X}$ of subsets of $X$ is \emph{nested} if for every $Y,Z\in\mathcal{N}$, either $Y\subseteq Z$ or $Z\subseteq Y$.
Let $\mathcal{C}$ be a collection of nested sets $\mathcal{N}_i\subseteq 2^X$ ($1\leq i\leq k$), denoted $\mathcal{C}=\langle \mathcal{N}_1,\dots,\mathcal{N}_k\rangle$. Observe that a subset $Y\subseteq X$ may occur in several nested sets of $\mathcal{C}$. We let  $\mathcal{S}\subseteq 2^{X}$ denote the union of the nested sets of $\mathcal{C}$, that is $\mathcal{S}=\bigcup_{1\leq i\leq k}\mathcal{N}_i$. We say that a permutation $\sigma\in\mathsf{Convex}(\mathcal{S})$ is $\mathcal{C}$-\emph{nested} if for every $1\leq i\leq k$ and every pair of sets $Y,Z\in\mathcal{N}_i$ such that $Z\subset Y$,  then $Y\setminus Z\prec_{\sigma} Z$. We let $\textsf{Nested-Convex}(\mathcal{C},\mathcal{S})$ denote the subset of permutations of $\textsf{Convex}(\mathcal{S})$ that are $\mathcal{C}$-nested.

\begin{lemma}\label{lem_OPQ_nested_convex}
Let ${\cal C}=\langle \mathcal{N}_1,\dots,\mathcal{N}_k\rangle$ be a collection of nested sets such that for every $1\leq i\leq k$, $\mathcal{N}_i\subset 2^X$. If $\textsf{Nested-Convex}(\mathcal{C},\mathcal{S})\neq\emptyset$, with $\mathcal{S}=\bigcup_{1\leq i\leq k}\mathcal{N}_i$, then there exists an \FPQ{}-tree $\mathsf{T}$ on $X$ such that $\mathfrak{S}_{\mathsf{T}}= \textsf{Nested-Convex}(\mathcal{C},{\cal S})$. Moreover, such an \FPQ-tree, when it exists, can be computed in polynomial time.
\end{lemma}
\begin{proof}

We first prove the following claim. Let $X$ be a non-empty set and consider ${\cal S} \subseteq 2^{X}$ containing two subsets $S_{1}$, $S_{2}$ of $S$ such that $S_{1} \subset S_{2}$. We set $S' = S_{2} \setminus S_{1}$ and $\mathcal{S}'={\cal S}\cup\{S'\}$. 

\begin{claim}
${\sf Convex}({\cal S}') = \{ \sigma\in {\sf Convex}({\cal S}) \mid \emph{$S_{2} \setminus S_{1} \prec_{\sigma} S_{1}$ or $S_{1}\prec_{\sigma} S_{2} \setminus S_{1}$}\}.$
\end{claim}

First observe that if $\sigma\in {\sf Convex}( {\cal S})$, satisfies $S_{2} \setminus S_{1} \prec_{\sigma} S_{1}$ or $S_{1} \prec_{\sigma} S_{2} \setminus S_{1}$, then $\sigma\in\textsf{Convex}({\cal S}')$. Now consider $\sigma' \in {\sf Convex}({\cal S}')$. Since $S_{1}$ and $S_{2} \setminus S_{1} $ are disjoint sets that both appear as intervals in $\sigma'$, we have that either $S_{1} \prec_{\sigma'} S_{2} \setminus S_{1} $ or $S_{2} \setminus S_{1}  \prec_{\sigma'} S_{1}$, proving the claim. 

\smallskip
Let us now consider a collection  ${\cal C}=\langle \mathcal{N}_1,\dots,\mathcal{N}_k\rangle$ of nested sets on $X$. Let $\mathcal{N}_i\in\mathcal{C}$ be a non trivial nested set, i.e. containing at least two subsets of $X$. Let $S_i^{\min}$ and $S_i^{\max}$ be respectively the largest and the smallest subset of $\mathcal{N}_i$. We set $S'_i=S_i^{\max}\setminus S_i^{\min}$. For every non-trivial nested set $\mathcal{N}_i$, we add to $\mathcal{S}$ the subset $S_i'$ resulting in $\mathcal{S}'\in 2^{X}$. By the observation above, we have $\textsf{Nested-Convex}(\mathcal{C},\mathcal{S})\subseteq\textsf{Convex}({\cal S}') $. By \autoref{lem_pq_algo}, we can compute (in linear time) a \PQ-tree $\mathsf{T}_1$ such that $\mathfrak{S}_{\T_1}=\textsf{Convex}(\mathcal{S}')$.

To compute an \FPQ{}-tree $\mathsf{T}$ on $X$ such that $\mathfrak{S}_{\mathsf{T}}= \textsf{Nested-Convex}(\mathcal{C},{\cal S})$,  the rest of the algorithm consists in freezing some \textsf{Q}-nodes of $\mathsf{T}_1$ into \textsf{F}-nodes. To that aim, we first prove the following claim.

\begin{claim} \label{cl_Pnode}
If $\mathcal{N}_i\in\mathcal{C}$ is non-trivial, then in $\mathsf{T}_1$, the least common ancestor $u$ of the elements of $S_i^{\max}$ is a \textsf{Q}-node. 
\end{claim}

Suppose towards a contradiction that $u$ is a \textsf{P}-node. By \autoref{lem_partitive}, this implies that $\mathcal{L}_{\mathsf{T}_1}(u)=S_i^{\max}$. By definition of $S_i'$, $S_i'$ and $S_i^{\min}$ partition the set $S_i^{\max}$, and thereby $\mathcal{L}_{\mathsf{T}_1}(u)$. Let $u_1$ and $u_2$ be the least common ancestors of the elements of $S_i^{\min}$ and of the elements of $S_i'$, respectively. Observe that for \autoref{lem_partitive} to also holds for $S'_i$ and $S_i^{\min}$, $u_1$ and $u_2$ have to be two distinct children of $u$. But as a \textsf{P}-node, $u$ has at least three children, implying that $\mathcal{L}_{\mathsf{T}_1}(u_1)\cup \mathcal{L}_{\mathsf{T}_1}(u_2)\neq \mathcal{L}_{\mathsf{T}_1}(u)$: a contradiction.

\smallskip
Claim \ref{cl_Pnode} allows to process $\mathsf{T}_1$ and $\mathcal{C}$ to compute the desired \FPQ-tree $\mathsf{T}$ as follows. Initially, we set $\mathsf{T}=\mathsf{T}_1$. For every non-trivial nested set $\mathcal{N}_i\in \mathcal{C}$, let $u_i$ be the node that is the least common ancestor of the element of $\mathcal{S}_i^{\max}$. By Claim \ref{cl_Pnode}, $u_i$ is not a \textsf{P}-node. If $u_i$ is a \textsf{Q}-node, then we freeze it into an \textsf{F}-node so that for every permutation $\sigma\in\mathfrak{S}_{\mathsf{T}}$ we have that $S_i'\prec_{\sigma} S_i^{\min}$. Otherwise, $u_i$ is an \textsf{F}-node and we check if in every permutation $\sigma\in\mathfrak{S}_{\mathsf{T}}$ we have that $S_i'\prec_{\sigma} S_i^{\min}$. This can be done easily in polynomial time.
\end{proof}

\subsection{\FPQ-hierarchies.} 

A \emph{hierarchy of ordered trees} $\H$ is defined on a set $\mathcal{T}=\{T_0,T_1,\dots, T_p\}$ of non-trivial ordered trees arranged in an edge-labelled tree, called the \emph{skeleton tree} $S_{\H}$. 
More formally, for $0<i\leq p$, the root $r_i$ of $T_i$ is attached, through a \emph{skeleton edge} $e_i$, to an internal node $f_i$ of some tree $T_j$ with $j<i$. Suppose that $e_i=r_if_i$ is the sketelon edge linking the root $r_i$ of $T_i$ to a node $f_i$ of $T_j$ having $c$ children. Then the label of $e_i$ is a pair of integers $I(e_i)=(a_i,b_i)\in [c]\times[c]$ with $a_i\leq b_i$. The contraction of the trees of $\mathcal{T}$ in a single node each, results in the skeleton tree $S_{\H}$.



From a hierarchy of ordered trees $\H$, we define a rooted tree $T_{\H}$ whose node set is $\bigcup_{0\leq i\leq p}\mathcal{L}(T_i)$ and that is built as follows. The root of $T_{\H}$ is $\ell_0$ the first leaf of $\mathcal{L}(T_0)$ in $\sigma_{T_0}$. For every $0\leq i\leq p$, the permutation $\sigma_{T_i}$ of $\mathcal{L}(T_i)$, defined by $T_i$, is a path of $T_{\H}$. Finally, for $1\leq i\leq p$, let $\ell_i$ be the first leaf of $\mathcal{L}(T_i)$ in $\sigma_{T_i}$. Suppose $T_i$ is connected in $\H$ to $T_j$ through the skeleton edge $e_i=r_if_i$ with label $I(e_i)=(a_i,b_i)$. Let $u_j$ the $b_i$-th child of $f_i$ in $T_j$ and $\ell$ be the leaf of $T_j$ that is a descendant of $u_j$ and largest in $\sigma_{T_j}$. Then set $\ell$ as the parent of $\ell_i$ in $T_{\H}$. See \autoref{fig_OPQ_hierarchy} for an example. 

An \emph{\FPQ-hierarchy} is a hierarchy of \FPQ-trees with an additional constraint on the labels of the skeleton edges. Let $e_i=r_if_i$ be the skeleton from the root $r_i$ of $T_i$ to the node $f_i$ of $T_j$ with $j\leq i$. If $f_j$ is a \textsf{P}-node with $c$ children, then $I(e_i)=(1,c)$. As in the case of \FPQ-trees, we say two \FPQ-hierarchies $\mathsf{H}$ and $\mathsf{H'}$ are isomorphic if they are isomorphic as labeled ordered trees. That is the types of the nodes, the skeleton edges and their labels are preserved. We say that $\mathsf{H}$ and $\mathsf{H'}$ are equivalent, denoted $\mathsf{H}\approx_{\FPQ}\mathsf{H'}$, if one can be turned into an \FPQ-hierarchy isomorphic to the other by a series of $\perm(u)$ and $\rev(u)$ operations (with $u$ being respectively a \textsf{P}-node and a \textsf{Q}-node) to modify relative ordering of the tree-children of $u$. Suppose that $u$ is a \textsf{Q}-node with $c$ children incident to a skeleton edge $e$. Then applying $\rev(u)$ transforms $I(e)=(a,b)$ into the new label $I^c(e)=(c+1-b,c+1-a)$. It follows that the equivalence class of an \FPQ{}-hierarchy $\mathsf{H}$ on the set $\mathcal{T}=\{\mathsf{T}_0,\mathsf{T}_1,\dots, \mathsf{T}_p\}$ of \FPQ-trees defines a set $\mathfrak{T}_{\FPQ{}}(\mathsf{H})$ of rooted trees on $\bigcup_{0\leq i\leq p}\mathcal{L}(\mathsf{T}_i)$. Observe that since reversing a \textsf{Q}-node modifies the labels of the incident skeleton edges, two rooted trees of $\mathfrak{T}_{\FPQ{}}(\mathsf{H})$ may not be isomorphic (see \autoref{fig_OPQ_hierarchy}).

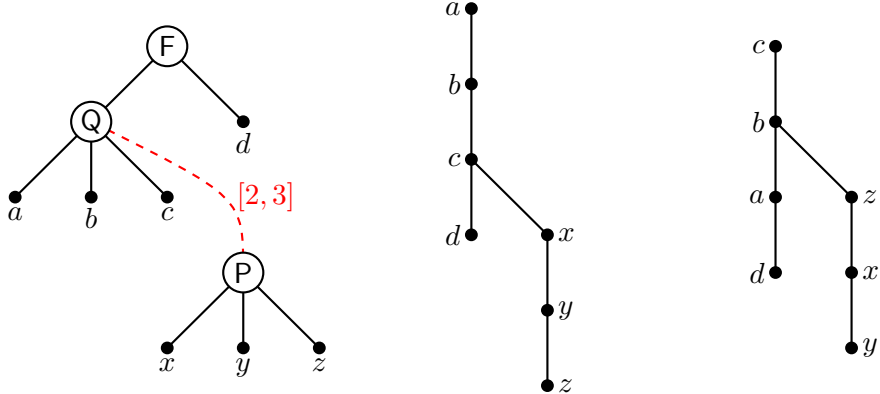
\begin{figure}[ht]
\begin{center}
\begin{tikzpicture}[thick,scale=1]
\tikzstyle{sommet}=[circle, draw, fill=black, inner sep=0pt, minimum width=4pt]          
\tikzstyle{Onode}=[circle, draw, fill=white, inner sep=0pt, minimum width=15pt]          
\tikzstyle{Pnode}=[circle, draw, fill=white, inner sep=0pt, minimum width=15pt]          
\tikzstyle{Qnode}=[circle, draw, fill=white, inner sep=0pt, minimum width=15pt]

\draw[red,thick,dashed] (3,-1) .. controls (3,0) .. (1,1) ;
\node[red] (a) at (3.3,0) {$[2,3]$};

\begin{scope}[xshift=0cm,yshift=0cm]

\draw (0,0) -- (1,1);
\draw (1,1) -- (2,2);
\draw (1,1) -- (1,0);
\draw (1,1) -- (2,0);
\draw (2,2) -- (3,1);


\draw (0,0) node[sommet]{};
\draw (1,0) node[sommet]{};
\draw (2,0) node[sommet]{};
\draw (3,1) node[sommet]{};
\draw (1,1) node[Pnode]{};
\node[] (p) at (1,1) {\textsf{Q}};
\draw (2,2) node[Onode]{};
\node[] (o) at (2,2) {\textsf{F}};


\node[below] (a) at (0,0) {$a$};
\node[below] (b) at (1,0) {$b$};
\node[below] (c) at (2,0) {$c$};
\node[below] (d) at (3,1) {$d$};

\end{scope}

\begin{scope}[xshift=1cm,yshift=0cm]

\draw (2,-1) -- (1,-2);
\draw (2,-1) -- (2,-2);
\draw (2,-1) -- (3,-2);

\draw (1,-2) node[sommet]{};
\draw (2,-2) node[sommet]{};
\draw (3,-2) node[sommet]{};
\draw (2,-1) node[Qnode]{};
\node[] (q) at (2,-1) {\textsf{P}};

\node[below] (e) at (1,-2) {$x$};
\node[below] (f) at (2,-2) {$y$};
\node[below] (g) at (3,-2) {$z$};
\end{scope}

\begin{scope}[xshift=6cm,yshift=-0.5cm]
\draw (0,3) node[sommet]{};
\draw (0,2) node[sommet]{};
\draw (0,1) node[sommet]{};
\draw (0,0) node[sommet]{};

\draw (1,0) node[sommet]{};
\draw (1,-1) node[sommet]{};
\draw (1,-2) node[sommet]{};

\draw (0,3) -- (0,0);
\draw (0,1) -- (1,0);
\draw (1,0) -- (1,-2);

\node[left] (a) at (0,3) {$a$};
\node[left] (a) at (0,2) {$b$};
\node[left] (a) at (0,1) {$c$};
\node[left] (a) at (0,0) {$d$};
\node[right] (a) at (1,0) {$x$};
\node[right] (a) at (1,-1) {$y$};
\node[right] (a) at (1,-2) {$z$};
\end{scope}

\begin{scope}[xshift=10cm,yshift=-1cm]
\draw (0,3) node[sommet]{};
\draw (0,2) node[sommet]{};
\draw (0,1) node[sommet]{};
\draw (0,0) node[sommet]{};

\draw (1,1) node[sommet]{};
\draw (1,0) node[sommet]{};
\draw (1,-1) node[sommet]{};

\draw (0,3) -- (0,0);
\draw (0,2) -- (1,1);
\draw (1,1) -- (1,-1);

\node[left] (a) at (0,3) {$c$};
\node[left] (a) at (0,2) {$b$};
\node[left] (a) at (0,1) {$a$};
\node[left] (a) at (0,0) {$d$};
\node[right] (a) at (1,1) {$z$};
\node[right] (a) at (1,0) {$x$};
\node[right] (a) at (1,-1) {$y$};
\end{scope}

\end{tikzpicture}
\end{center}
\caption{\label{fig_OPQ_hierarchy} An \FPQ-hierarchy $\mathsf{H}$. The  set $\mathfrak{T}_{\FPQ}(\mathsf{H})$ contains 12 rooted trees, two of which are depicted.
Observe that from the left to the right tree, the ordering on the leaves of the \textsf{Q}-node is reversed and  that the ordering on the leaves of the \textsf{P}-nodes are different. In both trees however, the path containing $\{x,y,z\}$ is attached below  the leaves $\{b,c\}$ since these leaves form the interval $[2,3]$ of the \textsf{Q}-node and this interval is  the label  of the unique skeleton edge.
}
\end{figure}

%


\section{Compact representation of the set of indifference tree-layouts}

In this section, we show how, for a given proper chordal graph $G$, an \FPQ{}-hierarchy $\H$ can be constructed to represent for a given vertex $x \in V$ the set of indifference tree-layout rooted at a vertex $x$ (if such an indifference tree-layout exists). To that aim, we first provide a characterization of indifference tree-layout alternative to~\autoref{lem_indifference_treelayout}. This characterization naturally leads us to define the notion of \emph{block} that for a fixed vertex $x$ of a proper chordal graph, drives the structure and the combinatorics of the set of indifference tree-layouts rooted at $x$.

\subsection{Blocks and indifference tree-layouts}

Let $S$ be a non-empty vertex subset of a connected graph $G=(V,E)$ and let $C$ be a connected component of $G-S$. We say that $x\in C$ is \emph{$S$-maximal} if for every vertex $y\in C$, $N(y)\cap S\subseteq N(x)\cap S$. Observe that if $C$ contains two distinct $S$-maximal vertices $x$ and $y$, then $N(x)\cap S=N(y)\cap S$. 

\begin{definition}\label{def_block}
Let $S$ be a subset of vertices of a graph $G=(V,E)$ and let $C$ be a connected component of $G-S$. A maximal subset of vertices $X\subseteq C$ is an \emph{$S$-block}, if every vertex of $X$ is $S$-maximal and $(N(S)\cap C)$-universal. 
\end{definition}

We let the reader observe that a connected component $C$ of $G-S$ may not contain an $S$-block. 
However, if $C$ contains an $S$-block, then it is uniquely defined. 

\begin{lemma} \label{lem_tree_layout}
Let $\T_G=(T,r,\rho_G)$ be a tree-layout of a graph $G=(V,E)$. Then $\T_G$ is an indifference tree-layout if and only if for every vertex $x$ distinct from $\rho_G^{-1}(r)$, $x$ is $A_{\T_G}(x)$-maximal and $(N(A_{\T_G}(x))\cap D_{\T_G}[x])$-universal.
\end{lemma}
\begin{proof}
Suppose that every vertex $x$ distinct from $\rho_G^{-1}(r)$, is $A_{\T_G}(x)$-maximal and $A_{\T_G}(x)$-universal.
Let $y$, $z$ be two distinct vertices of $G$ such that $y\prec_{\T_G} x\prec_{\T_G} z$ and $yz\in E$. Observe that $y\in A_{\T_G}(x)$ and 
$z\in N(A_{\T_G}(x))$. As $x$ is $A_{\T_G}(x)$-maximal and $(N(A_{\T_G}(x))\cap D_{\T_G}[x])$-universal, $yz\in E$ implies that $xy\in E$ and $xz\in E$. So $\T_G$ is an indifference tree-layout.

To prove the reverse, consider $x\in V$  distinct from $\rho_G^{-1}(r)$ and $C$ the connected component of $G-A_{\T_G}(x)$ containing $x$. First observe that as $\T_G$ is a tree-layout of $G$, every vertex $z\in C$ is a descendant of $x$ is $\T_G$. Suppose that either $x$ is not $A_{\T_G}(x)$-maximal or not $(N(A_{\T_G}(x))\cap D_{\T_G}[x])$-universal.
 In the former case, there exist two vertices $y\in A_{\T_G}(x)$ and $z\in C$ such that $yz\in E$ but $xy\notin E$.  In the latter case, there exists $y\in A_{\T}(x)$ and $z\in C$ such that $yz\in E$ but $xy\notin E$.
 By the above observation, we have that $z\prec_{\T_G} x\prec_{\T_G} y$ and so in both cases $\T_G$ is not an indifference tree-layout.
\end{proof}

So \autoref{lem_tree_layout} and the  discussion above shows that if $\T_G=(T,r,\rho_G)$ is  an indifference tree-layout of a proper chordal graph $G=(V,E)$, then every vertex $x$ distinct from $\rho_G^{-1}(r)$ can be associated with a non-empty $A_{\T_G}(x)$-block which contains $x$. Hereafter, we let $B_{\T_G}(x)$
denote the $A_{\T_G}(x)$-block containing $x$. If $x=\rho_G^{-1}(r)$, then we set $B_{\T_G}(x)=\{x\}$.

\begin{lemma} \label{lem_block_consecutive}
Let $\T_G=(T,r,\rho_G)$ be an indifference tree-layout of a connected proper chordal graph $G=(V,E)$. For every vertex $x\in V$, the vertices of the block $B_{\T_G}(x)$ appear consecutively on a path rooted at $x$ and induces a clique in $G$.
\end{lemma}
\begin{proof}
The property trivially holds for $x=\rho_G^{-1}(r)$. Let us consider $x\neq \rho_G^{-1}(r)$ and let $C$ be the connected component of $G-A_{\T_G}(x)$ containing $x$. By the definition of blocks, every vertex of $B_{\T_G}(x)$ is $(N(A_{\T_G}(x))\cap D_{\T_G}[x])$-universal. 
This implies that $B_{\T_G}(x)$ induces a clique.  Observe that in every tree-layout, the vertices of a clique appear on a path from the root to a leaf. 
By definition of $B_{\T_G}(x)$, every vertex of $B_{\T_G}(x)$ distinct from $x$ is a descendant of $x$ in $\T_G$. Let $z$ be the lowest vertex of $B_{\T_G}(x)$ in $\T_G$ and let $y\in V$ be a vertex such that $x\prec_{\T_G} y\prec_{\T_G} z$. Then observe that by the indifference property of $\T_G$, $y$ is $A_{\T_G}(x)$-maximal and $y$ is adjacent to every vertex of $N(A_{T_G}(x))$ that is a descendant of $x$. It follows that $y\in B_{\T_G}(x)$.
\end{proof}

\begin{lemma} \label{lem_maximal_block}
Let $\T_G=(T,r,\rho_G)$ be an indifference tree-layout of a connected proper chordal graph $G=(V,E)$. If $y\in B_{\T_G}(x)$, then $B_{\T_G}(y)\subset B_{\T_G}(x)$.
\end{lemma}
\begin{proof}
Suppose that there exists $z\in B_{\T_G}(y)\setminus B_{\T_G}(x)$. Let us observe that $x$, $y$ and $z$ are pairwise adjacent: $xy\in E$ because $y\in B_{\T_G}(x)$; $yz\in E$ because  $z\in B_{\T_G}(y)$; and finally, $xz\in E$ because $x\in A_{\T_G}(y)$ and $N(z)\cap A_{\T_G}(y)=N(y)\cap A_{\T_G}(y)$. As $A_{\T_G}(x)\subset A_{\T_G}(y)$, we also have that $N(x)\cap A_{\T_G}(x)=N(y)\cap A_{\T_G}(x)=N(z)\cap A_{\T_G}(x)$. So the fact that $z\notin B_{\T_G}(x)$, implies that there exists $v\in N(A_{\T_G}(x))\cap (D_{\T_G}(x)\cup\{x\})$ such that $vz\notin E$. As $xz\in E$, we know that $x\neq v$. As $y\in B_{\T_G}(x)$, we also have $yv\in E$. It follows that $v\in N(A_{\T_G}(y))\cap  (D_{\T_G}(y)\cup\{y\})$. Finally, since $z\in B_{\T_G}(y)$, we also have $vz\in E$: contradiction.
\end{proof}

\begin{lemma} \label{lem_connected_subtree}
Let $\T_G=(T,r,\rho_G)$ be an indifference tree-layout of a connected proper chordal graph $G=(V,E)$. For every vertex $x$, $G[D_{\T_G}(x)]$ is connected.
\end{lemma}
\begin{proof}
Let $x$ and $y$ be two vertices such that $x$ is the parent of $y$ in $\T_G$. We prove that $xy\in E$. For the sake of contradiction, suppose that $xy\notin E$. Observe that $N(x)\cap D_{\T_G}(x)=\emptyset$, as otherwise $\T$ would not be an indifference tree-layout. For the same reason, we have that for every $z\in D_{\T}(y)$, $N(z)\cap A_{\T_G}(y)=\emptyset$. This implies that $D_{\T_G}(x)$ and $A_{\T_G}(y)$ are not in the same connected component of $G$: contradiction.
\end{proof}

\subsection{The block tree of an indifference tree-layout}

Given an indifference tree-layout $\T_G=(T,r,\rho)$ of a connected proper chordal graph $G=(V,E)$, we define the set $\mathcal{B}(\T_G)$ as the set containing the inclusion-maximal blocks of $\T_G$. Observe that $\{\rho^{-1}(r)\}\in \mathcal{B}(\T_G)$. \autoref{cor_partition} is a direct consequence of \autoref{lem_maximal_block}.

\begin{corollary} \label{cor_partition}
Let $\T_G=(T,r,\rho_G)$ be an indifference tree-layout of a connected proper chordal graph $G=(V,E)$. Then, $\mathcal{B}(\T_G)$ is a partition of $V$.
\end{corollary}

From \autoref{lem_block_consecutive} and \autoref{cor_partition}, we can define the \emph{block tree} of $\T_G$, which will be denoted $\T_{G}\!\!\mid_{\mathcal{B}_{\T}(G)}$, obtained by contracting every block of $\mathcal{B}(\T_G)$ into a single node.

\begin{lemma} \label{lem_block_tree}
Let $\T_G=(T,r,\rho_G)$ and $\T'_G=(T',r',\rho'_G)$ be two indifference tree-layouts of  a connected proper chordal graph $G=(V,E)$ both having the same root $x=\rho_G^{-1}(r)=\rho_G'^{-1}(r')$. Then $\mathcal{B}(\T_G)=\mathcal{B}(\T'_G)$ and $T_G\!\!\mid_{\mathcal{B}(\T_G)}$ and $T'_G\!\!\mid_{\mathcal{B}(\T'_G)}$ are isomorphic trees.
\end{lemma}
\begin{proof}
We prove the property by induction on the inclusion-maximal blocks of $\T_G$ and $\T_G'$. We have that $\{x\}$ belongs to both $\mathcal{B}(\T_G)$ and $\mathcal{B}(\T'_G)$. Let us assume that the property holds for a subset $\mathcal{B}\subset \mathcal{B}(\T_G)\cap \mathcal{B}(\T'_G)$ of maximal blocks inducing a connected subtree of  $\T_G$ and $\T_G'$ containing the root block $\{x\}$. Let $S$ be the set of vertices belonging to blocks in $\mathcal{B}$. Let $C$ be a connected component of $G-S$. By \autoref{lem_connected_subtree}, $C$ induces a subtree of $\T_G$ (respectively of $\T_G'$) such that the root $y$ (respectively $y'$) of that subtree is the child of some vertex $z$ of some block in $\mathcal{B}$ (respectively $z'$). Now observe that
as $\T_G$ and $\T_G'$ are two indifference tree-layouts of $G$, both $B_{\T_G}(y)$ and  $B_{\T_G'}(y')$ contain the vertices of $C$ that are $S$-maximal vertices and that are $N(S)\cap C$-universal. It follows that $B_{\T_G}(y)=B_{\T_G'}(y')$. Set $B=B_{\T_G}(y)$. Moreover by the indifference property, observe that $N(B)\cap S$ is consecutive in $[r,y]_{\T_G}$ and contains $z$. Similarly,  $N(B)\cap S$ is consecutive in $[r,y']_{\T_G'}$ and contains $z'$. As the induction hypothesis holds on $\mathcal{B}$, $y'$ and $z'$ belong to the same block of $\mathcal{B}$. It follows that the property is satisfied on $\mathcal{B}\cup\{B\}$ as well, implying that $\mathcal{B}(\T_G)=\mathcal{B}(\T'_G)$ and $T_G\!\!\mid_{\mathcal{B}(\T_G)}$ and $T'_G\!\!\mid_{\mathcal{B}(\T'_G)}$  are isomorphic trees.
\end{proof}

\newcommand{\Btree}{\mathsf{B}^{\sf tree}}

It follows that the block tree only depends on the root vertex and not on a given indifference tree-layout rooted at that vertex. Thereby, from now on, we let $\mathcal{B}_G(x)$ and $\Btree_G(x)$ respectively denote the set of maximal blocks and the block tree of any indifference tree-layout of $G$ rooted at vertex $x$. 
However, we observe that the block tree does not fully describe the possible set of indifference tree-layouts. The block tree does not reflect how each block is precisely attached to its parent block.

Let $\T_G=(T,r,\rho_G)$ be an indifference tree layout of a graph $G=(V,E)$. Let $B_{\T_G}(x)\in \mathcal{B}(\T_G)$ with $x$ distinct from the root of $\T_G$. We denote by $C_{\T_G}(x)$ the connected component of $G-A_{\T_G}(x)$ containing $x$. Let $C_1,\dots C_k$ be the connected components of $G[C_{\T_G}(x)\setminus B_{\T_G}(x)]$. We define $\mathcal{C}_{\T_G}(x)=\langle \mathcal{N}_1,\dots , \mathcal{N}_k\rangle$ a collection of sets of $2^{B_{\T_G}(x)}$: for every $1\leq i\leq k$ and every vertex $y\in C_i$ such that $N(y)\cap B_{\T_G}(x)\neq\emptyset$, then we add $N(y)\cap B_{\T_G}(x)$ to $\mathcal{N}_i$. We define $\mathcal{S}_x=\bigcup_{1\leq i\leq k} \mathcal{N}_i$.

\begin{lemma} \label{lem_nested_convex}
Let $\T_G=(T,r,\rho_G)$ be an indifference tree-layout of a connected proper chordal graph $G$. Let $x$ be a vertex of $G$ distinct from the root $\rho_G(r)$ such that $B_{\T_G}(x)\in \mathcal{B}_{\T_G}(G)$. Then $\mathcal{C}_{\T_G}(x)=\langle \mathcal{N}_1,\dots , \mathcal{N}_k\rangle$ is a collection of nested sets such that $\textsf{Nested-Convex}(\mathcal{C}_{\T_G}(x),\mathcal{S}_x)\neq\emptyset$.
\end{lemma}
\begin{proof}
We first argue that $\mathcal{C}_{\T_G}(x)$ is a collection of nested sets. Consider two vertices $y$ and $z$ of some connected component $C_i$ of $G[C_{\T_G}(x)\setminus B_{\T_G}(x)]$ that both have a neighbor in $B_{\T_G}(x)$. Observe that $\lca_{\T_G}(y,z)$ is a node mapped to a vertex of $C_i$. By the indifference property of $\T_G$, it follows that $N(y)\cap B_{\T_G}(x)$ and $N(z)\cap B_{\T_G}(x)$ form in $\T_G$ two subpaths sharing the same least boundary, implying that $\mathcal{N}_i$ is a nested set. Finally, observe  the permutation $\sigma$ of $B_{\T_G}(x)$ such that for every two vertices $y,z\in B_{\T_G}(x)$, $y\prec_{\sigma} z$ if and only if $y\in A_{\T_G}(z)$ is a permutation of $\textsf{Nested-Convex}(\mathcal{C}_{\T_G}(x),\mathcal{S}_x)$.
\end{proof}

\subsection{A canonical \FPQ{} hierarchy}

We can now prove, using \autoref{lem_partitive}, \autoref{lem_OPQ_nested_convex}, \autoref{cor_partition} and \autoref{lem_nested_convex}, the existence of a canonical $\FPQ$-hierarchy $\mathsf{H}_G(x)$ that encodes the set of indifference tree-layouts of $G$ rooted at $x$, if there exists one (see~\autoref{fig_indif_OPQ}).

\begin{theorem} \label{th_canonical}
Let $G=(V,E)$ be a proper chordal graph. If $G$ has an indifference tree-layout rooted at some vertex $x$, then there exists an \FPQ-hierarchy $\mathsf{H}_G(x)$ such that a tree-layout $\T_G=(T,r,\rho_G)$ of $G$ is an indifference tree-layout such that $\rho_G^{-1}(r)=x$ if and only if $\T_G\in\mathfrak{T}_{\FPQ}(\mathsf{H}_G(x))$.

Moreover $\mathsf{H}_G(x)$ is unique and, given an indifference tree-layout rooted at $x$, it can be computed in polynomial time.
\end{theorem}
\begin{proof}
Let $\T_G=(T,r,\rho_G)$ be an indifference tree-layout of $G$ such that $\rho_G^{-1}(r)=x$. We can construct an $\FPQ$-hierarchy $\mathsf{H}_G(x)$ as follows (see \autoref{fig_indif_OPQ}):
\begin{enumerate}
\item The skeleton tree of $\mathsf{H}_G(x)$ is the block tree $\Btree_{G}(x)$.
\item For each block $B_{\T_G}(y)\in \mathcal{B}_{\T_G}(G)$, with $y\neq x$, compute the \FPQ-tree $\mathsf{T}_y$ of the collection $\mathcal{C}_{\T_{G}}(y)$ of nested sets such that $\sigma_{\mathsf{T}_y}=\textsf{Nested-Convex}(\mathcal{C}_{\T_{G}}(y),\mathcal{S}_y)$.
\item Suppose that $B_{\T_G}(z)$ is a child of $B_{\T_G}(y)$ in $B^{\sf tree}_{G}(x)$. Then the parent in $\mathsf{H}_G(x)$ of the root $r_z$ of $\mathsf{T}_z$ is the node $u_{z}=\lca_{\mathsf{T}_y}(I_z)$ where $I_z=N(z)\cap B_{\T_G}(y)\in\mathcal{S}_y$. Suppose that 
$u$ has $k$ children $v_1,\dots, v_k$. We set the label of the skeleton-edge $r_zu$ to the interval $[h,j]$ such that $I_z=\bigcup_{h\leq i\leq j} \mathcal{L}_{\mathcal{T}_y}(v_i)$.
\end{enumerate}

\begin{figure}[h]
\begin{center}
\begin{tikzpicture}[thick,scale=1]
\tikzstyle{sommet}=[circle, draw, fill=black, inner sep=0pt, minimum width=4pt]          

\tikzstyle{Onode}=[circle, draw, fill=white, inner sep=0pt, minimum width=15pt]          
\tikzstyle{Pnode}=[circle, draw, fill=white, inner sep=0pt, minimum width=15pt]          
\tikzstyle{Qnode}=[circle, draw, fill=white, inner sep=0pt, minimum width=15pt]

\begin{scope}[]
\draw[fill=gray!20] (-0.5,5.6) rectangle (0.3,6.4) ;
\draw[fill=gray!20] (-0.5,4.6) rectangle (0.3,5.4) ;
\draw[fill=gray!20] (-0.5,-0.35) rectangle (0.3,4.4) ;
\draw[shift=(-120:1),fill=gray!20] (-0.5,-0.4) rectangle (0.3,0.35) ;
\draw[shift=(-120:1),fill=gray!20] (-0.5,-1.4) rectangle (0.3,-0.6) ;

\draw[shift=(-40:1),fill=gray!20] (-0.3,-0.4) rectangle (0.5,0.4) ;
\begin{scope}[shift=(90:1)]
\draw[shift=(-40:1),fill=gray!20] (-0.3,-0.4) rectangle (0.5,0.4) ;
\end{scope}
\begin{scope}[shift=(90:2)]
\draw[shift=(-40:1),fill=gray!20] (-0.3,-0.4) rectangle (0.5,0.4) ;
\end{scope}
\begin{scope}[shift=(90:5)]
\draw[shift=(-40:1),fill=gray!20] (-0.3,-0.4) rectangle (0.5,0.4) ;
\end{scope}

\node (T) at (90:-2.7) {$\T$};

\node[left] (x) at (90:6) {$x$};
\node[left] (y) at (90:5) {$y$};
\node[right,shift=(90:5)] (z) at (-40:1) {$z$};
\node[left] (5) at (90:4) {$a$};
\node[left] (4) at (90:3) {$b$};
\node[left] (3) at (90:2) {$c$};
\node[left] (2) at (90:1) {$d$};
\node[left] (1) at (90:0) {$e$};
\node[left] (u) at (-120:1) {$s$};
\node[left,shift=(-120:1)] (v) at (-90:1) {$t$};

\node[right] (a) at (-40:1) {$w$};
\node[right,shift=(90:1)] (b) at (-40:1) {$v$};
\node[right,shift=(90:2)] (c) at (-40:1) {$u$};

\draw (90:6) -- (90:5) -- (90:4) -- (90:3) -- (90:2) -- (90:1) -- (0:0) -- (-120:1) ;
\draw (0:0) -- (-40:1) ;
\draw[shift=(90:1)] (0:0) -- (-40:1) ;
\draw[shift=(90:2)] (0:0) -- (-40:1) ;
\draw[shift=(-120:1)] (90:0) -- (90:-1) ;
\draw[shift=(90:5)] (0:0) -- (-40:1) ;

\draw[red,thick] (90:0) .. controls (120:4.5) .. (90:6) ;
\draw[red,thick] (90:1) .. controls (115:4.7) .. (90:6) ;
\draw[red,thick] (90:2) .. controls (110:4.9) .. (90:6) ;
\draw[red,thick] (90:3) .. controls (105:5.1) .. (90:6) ;
\draw[red,thick] (90:4) .. controls (100:5.3) .. (90:6) ;
\draw[red,thick] (90:5) .. controls (95:5.5) .. (90:6) ;

\draw[blue,thick] (-120:1) .. controls (115:3.5) .. (90:5) ;
\draw[blue,thick] (v.east) .. controls (150:2) .. (90:3) ;

\draw[blue,thick] (a.west) .. controls (0:0.6) .. (2.east) ;
\draw[blue,thick,shift=(90:1)] (-40:1) .. controls (0:0.6) .. (90:1) ;
\draw[blue,thick,shift=(90:2)] (-40:1) .. controls (0:0.6) .. (90:1) ;
\draw[red,thick,shift=(90:5)] (-40:1) .. controls (0:0.6) .. (90:1) ;

\draw[shift=(90:5)] (-40:1) node[sommet]{};
\draw (90:6) node[sommet]{};
\draw (90:5) node[sommet]{};
\draw (90:4) node[sommet]{};
\draw (90:3) node[sommet]{};
\draw (90:2) node[sommet]{};
\draw (90:1) node[sommet]{};
\draw (0:0) node[sommet]{};
\draw (-120:1) node[sommet]{};
\draw (-40:1) node[sommet]{};
\draw[shift=(-120:1)] (90:-1) node[sommet]{};
\draw[shift=(90:1)]  (-40:1) node[sommet]{};
\draw[shift=(90:2)] (-40:1) node[sommet]{};
\end{scope}

\begin{scope}[xshift=5cm,yshift=1cm]

\begin{scope}[xshift=-0.5cm,yshift=-2.5cm]
\draw[red,thick,dashed] (0,1) .. controls (0,3) .. (0.5,4.5) ;
\node[red] (Iu) at (-0.4,2.5) {\small{$[1,2]$}};
\end{scope}

\begin{scope}[xshift=-1.5cm,yshift=-3.5cm]
\draw[red,thick,dashed] (0,1) .. controls (0.2,1.5) .. (1,2) ;
\node[red] (Iv) at (-0.2,1.6) {\small $[1,1]$};
\end{scope}

\begin{scope}[xshift=1cm,yshift=-2.5cm]
\draw[red,thick,dashed] (-0.5,1) .. controls (-0.7,1.9) ..  (0.5,3.5) ;
\node[red] (Ia) at (-1,1.7) {\small $[1,2]$};
\end{scope}

\begin{scope}[xshift=2cm,yshift=-2.5cm]
\draw[red,thick,dashed] (-0.5,1) -- (-0.5,3.5) ;
\node[red] (Ia) at (-0.9,1.7) {\small $[2,3]$};
\end{scope}

\begin{scope}[xshift=3cm,yshift=-2.5cm]
\draw[red,thick,dashed] (-0.5,1) .. controls (-0.1,1.8) .. (-1.5,3.5) ;
\node[red] (Ia) at (-0.7,1.7) {\small $[3,4]$};
\end{scope}

\begin{scope}[xshift=0cm,yshift=5cm]
\draw (-0.5,0) -- (-0.5,-1);

\draw[red,thick,dashed] (-0.5,0) -- (1,-1.5) ;
\node[red] (Iz) at (2.2,-1.8) {\small $[1,1]$};

\draw[red,thick,dashed] (1,-1.5) -- (2.5,-2.5) ;
\node[red] (Iy) at (0.1,-1.2) {\small $[1,1]$};

\draw (-0.5,0) node[Qnode]{};
\node[] (p) at (-0.5,0) {\textsf{Q}};
\node[below] (a) at (-0.5,-1) {$x$};
\end{scope}

\begin{scope}[xshift=1cm,yshift=3.5cm]
\draw (0,0) -- (0,-1);
\draw[red,thick,dashed] (0,0) -- (-1,-1.5) ;
\node[red] (I5) at (-1,-0.7) {\small $[1,1]$};
\draw (0,0) node[Qnode]{};
\node[] (p) at (0,0) {\textsf{Q}};
\node[below] (a) at (0,-1) {$y$};
\end{scope}

\begin{scope}[xshift=2.5cm,yshift=2.5cm]
\draw (0,0) -- (0,-1);
\draw (0,0) node[Qnode]{};
\node[] (pa) at (0,0) {\textsf{Q}};
\node[below] (a) at (0,-1) {$z$};
\end{scope}

\draw (1.5,1) -- (0,2);
\draw (1.5,1) -- (0,0);
\draw (1.5,1) -- (1,0);
\draw (1.5,1) -- (2,0);
\draw (1.5,1) -- (3,0);
\draw (0,2) -- (-1,1);

\draw (0,0) node[sommet]{};
\draw (1,0) node[sommet]{};
\draw (2,0) node[sommet]{};
\draw (3,0) node[sommet]{};
\draw (-1,1) node[sommet]{};

\draw (1.5,1) node[Qnode]{};
\node[] (p) at (1.5,1) {\textsf{Q}};
\draw (0,2) node[Onode]{};
\node[] (o) at (0,2) {\textsf{F}};

\node[below] (a) at (0,0) {$b$};
\node[below] (b) at (1,0) {$c$};
\node[below] (c) at (2,0) {$d$};
\node[below] (d) at (3,0) {$e$};
\node[below] (e) at (-1,1) {$a$};
\node[] (o) at (0,2) {\textsf{F}};

\begin{scope}[xshift=-0.5cm,yshift=-1.5cm]

\draw (0,0) -- (0,-1);
\draw (0,0) node[Qnode]{};
\node[] (pu) at (0,0) {\textsf{Q}};
\node[below] (a) at (0,-1) {$s$};

\end{scope}

\begin{scope}[xshift=-1.5cm,yshift=-2.5cm]
\draw (0,0) -- (0,-1);

\draw (0,0) node[Qnode]{};
\node[] (pv) at (0,0) {\textsf{Q}};
\node[below] (a) at (0,-1) {$t$};
\end{scope}

\begin{scope}[xshift=0.5cm,yshift=-1.5cm]

\draw (0,0) -- (0,-1);
\draw (0,0) node[Qnode]{};
\node[] (p) at (0,0) {\textsf{Q}};
\node[below] (a) at (0,-1) {$u$};
\end{scope}

\begin{scope}[xshift=1.5cm,yshift=-1.5cm]

\draw (0,0) -- (0,-1);
\draw (0,0) node[Qnode]{};
\node[] (p) at (0,0) {\textsf{Q}};
\node[below] (a) at (0,-1) {$v$};
\end{scope}

\begin{scope}[xshift=2.5cm,yshift=-1.5cm]

\draw (0,0) -- (0,-1);
\draw (0,0) node[Qnode]{};
\node[] (p) at (0,0) {\textsf{Q}};
\node[below] (a) at (0,-1) {$w$};
\end{scope}

\end{scope}

\begin{scope}[shift=(0:12)]
\draw[fill=gray!20] (-0.5,-0.35) rectangle (0.3,4.4) ;

\node[left] (x) at (90:6) {$x$};
\node[left] (y) at (90:5) {$y$};
\node[right,shift=(90:5)] (z) at (-40:1) {$z$};
\node[left] (5) at (90:4) {$a$};
\node[left] (4) at (90:3) {$e$};
\node[left] (3) at (90:2) {$d$};
\node[left] (2) at (90:1) {$c$};
\node[left] (1) at (90:0) {$b$};
\node[left] (u) at (-120:1) {$s$};
\node[left,shift=(-120:1)] (v) at (-90:1) {$t$};

\node[right] (a) at (-40:1) {$u$};
\node[right,shift=(90:1)] (b) at (-40:1) {$v$};
\node[right,shift=(90:2)] (c) at (-40:1) {$w$};

\draw (90:6) -- (90:5) -- (90:4) -- (90:3) -- (90:2) -- (90:1) -- (0:0) -- (-120:1) ;
\draw (0:0) -- (-40:1) ;
\draw[shift=(90:1)] (0:0) -- (-40:1) ;
\draw[shift=(90:2)] (0:0) -- (-40:1) ;
\draw[shift=(-120:1)] (90:0) -- (90:-1) ;
\draw[shift=(90:5)] (0:0) -- (-40:1) ;

%
%


\draw[shift=(90:5)] (-40:1) node[sommet]{};
\draw (90:6) node[sommet]{};
\draw (90:5) node[sommet]{};
\draw (90:4) node[sommet]{};
\draw (90:3) node[sommet]{};
\draw (90:2) node[sommet]{};
\draw (90:1) node[sommet]{};
\draw (0:0) node[sommet]{};
\draw (-120:1) node[sommet]{};
\draw (-40:1) node[sommet]{};
\draw[shift=(-120:1)] (90:-1) node[sommet]{};
\draw[shift=(90:1)]  (-40:1) node[sommet]{};
\draw[shift=(90:2)] (-40:1) node[sommet]{};

\node (TT) at (90:-2.7) {$\T'$};
\end{scope}

\end{tikzpicture}
\end{center}
\caption{\label{fig_indif_OPQ} On the left hand side, an indifference tree-layout $\T$ rooted at vertex $x$ of a proper chordal graph $G$. For every vertex, only the edge (blue or red) to its highest neighbor in $\T$ is depicted. The boxes represent the partition into blocks. The \FPQ-hierarchy $\mathsf{H}_{G}(x)$ is depicted in the middle. Observe that for the \FPQ-tree $\mathsf{T}$ of the block $B_{\T}(a)=\{a,b,c,d,e\}$, $\mathfrak{S}(\mathsf{T})=\{abcde,aedcb\}$. It follows that $\mathfrak{T}_{\FPQ}(\mathsf{H}_{G}(x))$ contains two indifference tree-layouts $\T$ and $\T'$. The latter one is shown on  the right hand side.}
\end{figure}
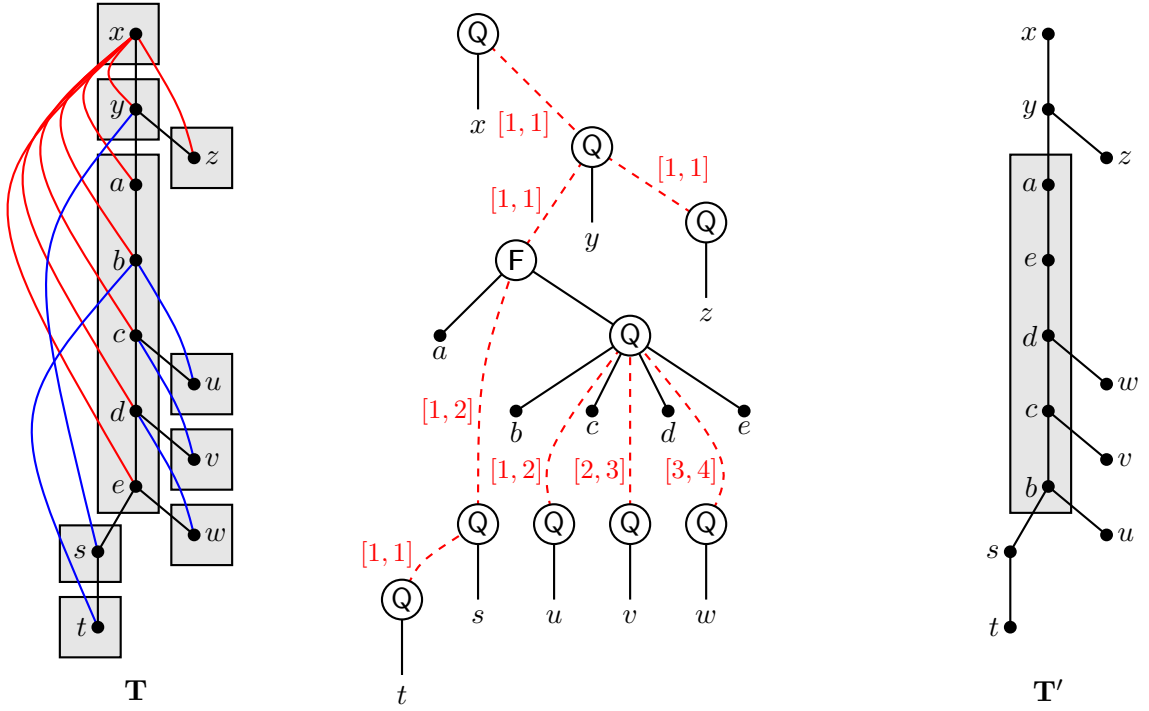

\begin{claim}
The \FPQ-hierarchy $\mathsf{H}_G(x)$ is uniquely defined and can be computed in polynomial time. 
\end{claim}

By \autoref{lem_block_tree}, the block tree $\Btree_G(x)$ is unique. Moreover, given $\T_G$, it can be computed in polynomial time. Let $y$ be a vertex distinct from $x$ such that $B_{\T_G}(y)\in \mathcal{B}_{\T_G}(G)$. By \autoref{lem_nested_convex}, the collection $\mathcal{C}_{\T_G}(y)$ verifies that $\textsf{Nested-Convex}(\mathcal{C}_{\T_G}(y),\mathcal{S}_y)\neq\emptyset$. So by \autoref{lem_OPQ_nested_convex}, the \FPQ-tree $\mathsf{T}_y$ exists and can be computed in polynomial time. Finally, suppose that $B_{\T_G}(z)$ is a child of $B_{\T_G}(y)$ in the block tree $B^{\sf tree}_G(x)$. By \autoref{lem_partitive}, the node $u_z$ is well defined and the interval $I_z$ is unique. 

It directly follows from the construction above and the definition of  $\mathfrak{T}_{\FPQ}(\mathsf{H}_{G}(x))$ that $\T_G\in\mathfrak{T}_{\FPQ}(\mathsf{H}_{G}(x))$. 

\begin{claim} \label{cl_indifference_treelayout}
If $\T_G\in\mathfrak{T}_{\FPQ}(\mathsf{H}_G(x))$ then $\T_G$ is an indifference tree-layout of $G$ rooted at $x$.
\end{claim}

Let us argue that $\T_G$ is a tree-layout of $G$. Let $y$ and $z$ be two vertices of $G$ such that neither $y\prec_{\T_G} z$ nor $z\prec_{\T_G} y$. We denote $B_y$ and $B_z$ the blocks of $\mathcal{B}_{\T_G}(G)$ respectively containing $y$ and $z$. There are two cases to consider. Suppose first that $B_y$ is an ancestor of $B_z$ in the block tree $\Btree_G(x)$. Observe that by the choice of the permutation $\sigma_{\mathsf{T}_y}$ of $B_y$ in the construction of $\T_G$, for every vertex $v\in N(z)\cap B_y$, $v \prec_{\T_G} \lca_{\T_G}(y,z)$. It follows that $yz\notin E$. So let us assume that $B_y$ and $B_z$ are not an ancestor of one another in the block tree $\Btree_G(x)$. Let $B_{\T_G}(u)\in \mathcal{B}_{\T_G}(G)$ be their least common ancestor. Observe that $B_{\T_G}(u)\cup A_{\T_G}(u)$ is a separator in $G$ for $y$ and $z$.
It follows that $yz\notin E$. Thereby, $\T_G$ is indeed a tree-layout of $G$.

To prove that $\T_G$ is an indifference tree-layout, let us consider three distinct vertices $v$, $y$ and $z$ of $V$ such that $y\prec_{\T_G} v\prec_{\T_G} z$ and $yz\in E$. 
Let $B_z\in \mathcal{B}_{\T_G}(G)$ be the block containing $z$.  Suppose that $v\in B_z$. As, by \autoref{lem_tree_layout}, every block  of $\mathcal{B}_{\T_G}(G)$ induces a clique, $vz\in E$. If $y\in B_z$, then, for the same reason, $yv\in E$. Otherwise, by definition of a block, $N(v)\cap B=N(z)\cap B$ for every block $B$ that is an ancestor of $B_y$ in $\Btree_G(x)$, implying that $yv\in E$. 
Suppose that $v\notin B_z$. Let $B_y\in \mathcal{B}_{\T_G}(G)$ be the block containing $y$. Suppose that $v\in B_y$. As $B_y$ induces a clique, $vy\in E$. Moreover, by construction of $\T_G$ and the choice of $\sigma_{\mathsf{T}_y}\in\mathfrak{S}(\mathsf{T}_y)$, where $\mathsf{T}_y$ is the \FPQ-tree associated to the block $B_y$, $N(z)\cap B_y$ is an interval that contains $D_{\T_G}(y)\cap B_y$. It follows that $zv\in E$. So assume that $v \notin B_{y}$. Observe then that by the definition of blocks, we have that $N(z)\cap B_y\subseteq N(v)\cap B_y$ and that $y$ is adjacent to every vertex of $B_z$ that has a neighbor in $B_y$, implying that $vy\in E$ and $vz\in E$. \end{proof}
%

\section{Algorithmic aspects}

This section is dedicated to the design of polynomial time algorithms for proper chordal graphs. We first tackle the recognition problem of proper chordal graphs. Then we show that using the \FPQ{}-hierarchies, we can resolve the graph isomorphism problem between two proper chordal graphs in polynomial time. In a way, this generalizes the \PQ{}-tree based graph isomorphism algorithm for interval graphs~\cite{LuekerB79}.  This latter result is of interest as it shows that proper chordal graphs extends the known limit of tractability for the graph isomorphism problem. Indeed, it is known that graph isomorphism is \textsf{GI}-complete on  strongly chordal graphs~\cite{UeharaTN05}.

\subsection{Recognition}

Given a graph $G=(V,E)$, the recognition algorithm test for every vertex $x\in V$, if $G$ has an indifference tree-layout rooted at $x$. So from now on, we assume that we are given as input the graph $G$ and a fixed vertex $x$. The algorithm is then two-step. First, the first step of the algorithm aims at computing the block tree $\Btree_G(x)$ of $G$ rooted at $x$ that would correspond to the skeleton tree of the \FPQ{}-hierarchy $\H_G(x)$ if $G$ has an indifference tree-layout rooted at $x$. Then, in the second step, instead of computing $\H_G(x)$, we verify that $\Btree_G(x)$ can indeed be turned into an indifference tree-layout of $G$. If eventually we can construct an indifference tree-layout, then $G$ is proper chordal. If $G$ is proper chordal and has an indifference tree-layout rooted at $x$, then the algorithm will succeed.


\paragraph{Computing the blocks and the block tree.} 
Let us assume that the input graph $G=(V,E)$ is proper chordal and let us consider a vertex $x\in V$ such that $x$ is the root of some indifference tree-layout of $G$. As discussed above, the first step aims at computing the skeleton tree of $\H_G(x)$ (see \autoref{th_canonical}). That skeleton tree is the block tree $\Btree_G(x)$ that can be obtained from any indifference tree-layout rooted at $x$ by contracting the blocks of $\mathcal{B}_G(x)$ into a single node each (see \autoref{lem_block_tree}). To compute $\Btree_G(x)$, we perform a search on $G$ starting at $x$ (see  Algorithm~\ref{alg_block_tree} below). At every step of the search, if the set of searched vertices is $S$, then the algorithm either identifies, in some connected component $C$ of $G-S$, a new block $S$-block of $\mathcal{B}_G(x)$ and connects it to the current block tree, or (if $C$ does not contain an $S$-block) stops and declares that there is no block tree rooted at $x$.

\begin{algorithm}[h]
{\small 
\KwIn{A graph $G=(V,E)$ and a vertex $x\in V$.}
\KwOut{The block tree $\Btree_G(x)$, if $G$ has an indifference tree-layout rooted at $x$.}
\BlankLine
set $S\leftarrow \{x\}$, $\mathcal{B}\leftarrow \big\{\{x\}\big\}$ and $\Btree\leftarrow(\mathcal{B},\emptyset)$\;
\While{$S\neq V$}{
	let $C$ be a connected component of $G-S$\;
	\eIf{$C$ contains a $S$-block $X$}{
		$S\leftarrow S\cup X$ and $\mathcal{B}\leftarrow \mathcal{B}\cup \big\{X\big\}$\;
		let $B\in \mathcal{B}$ such that $N(X)\cap B\neq\emptyset$ and that is the deepest\;
		add an edge in $\Btree$ between $B$ and $X$\;
		}{\textbf{stop} and \Return{$G$ has no indifference tree-layout rooted at $x$}\;
		}
	}
\Return{$\Btree$}\;

\caption{Block tree computation} \label{alg_block_tree}
}
\end{algorithm}

\begin{lemma} \label{lem_alg_block_tree}
Let $x$ be a vertex of a graph $G=(V,E)$. If $G$ is proper chordal and has an indifference tree-layout rooted at $x$, then Algorithm~\ref{alg_block_tree} returns the block tree $\Btree_G(x)$ that is the skeleton tree of the \FPQ{}-hierarchy $\H_G(x)$.
\end{lemma}
\begin{proof}
Suppose that $G$ is proper chordal and has an indifference tree-layout $\T_G=(T,r,\rho_G)$ rooted at $x$. From \autoref{lem_block_tree}, $\Btree_G(x)$ is well-defined and unique. To prove the correctness of   Algorithm~\ref{alg_block_tree}, we establish the following invariant. At every step, if $S$ is the set of searched vertices, then
\begin{enumerate}
\item[(1).] $\mathcal{B}\subseteq \mathcal{B}_{G}(x)$ and $\Btree=\Btree_{G[S]}(x)$ is the unique block-tree rooted at $x$ of $G[S]$;
\item[(2).] if $C$ is a connected component of $G-S$, then for every pair of block, $B$, $B'\in \mathcal{B}$ such that $N(C)\cap B\neq\emptyset$ and $N(C)\cap B'\neq\emptyset$, either $B$ is an ancestor of $B'$ in $\Btree$ or vice-versa.
\end{enumerate}

The invariant clearly holds when $S=\{x\}$. Let $C$ be a connected component of $G-S$. 
Let $y\in C$ that is the closest to the root of $\T_G$. Observe that as (1) holds on $G[S]$, $X=B_{\T_G}(y)$ is an $S$-block and it belongs to $\mathcal{B}_{G}(x)$.
By definition of $\Btree_G(x)$, observe that the parent of every block $B'$ of $\mathcal{B}_G(x)$ is the deepest block $B$ in $\Btree_G(x)$ such that $N(B')\cap B\neq\emptyset$. By (2),  $\Btree$ contains a unique deepest block $B\in \mathcal{B}$ such that $N(C)\cap B\neq\emptyset$. It follows that if we attach $X$ to $B$ in $\Btree$, we extend $\Btree$ to $\Btree_{G[S\cup X]}(x)$, the block tree of $G[S\cup X]$. So (1) is still valid after adding $X$ to the set of visited vertices. Moreover, observe that when removing $X$ from $C$ every connected component of $G[C]-X$ has neighbors in $X$, implying that (2) is maintained as well. 

Finally, we observe that every step of the algorithm can be performed in polynomial time, including testing the existence of an $S$-block in $C$.
\end{proof}

Before describing the second step of the algorithm, let us discuss some properties of $\Btree$ returned by Algorithm~\ref{alg_block_tree} when $\mathcal{B}$ is a partition of $V$. An \emph{extension} of $\Btree$ is any tree $T_{\Btree}$ obtained by substituting every node $B\in\mathcal{B}$ by an arbitrary permutation $\sigma_B$ of the vertices of $B$. In this construction, if $B$ is the parent of $B'$ in $\Btree$, $x$ is the last vertex of $B$ in $\sigma_B$ that has a neighbor in $B'$ and $x'$ is the first vertex of $B'$ in $\sigma_{B'}$, then the parent of $x'$ in $T_{\Btree}$ is a vertex of $B$ that appears after $x$ in $\sigma$.

\begin{observation} \label{obs_alg_blocktree}
Let $x$ be a vertex of a graph $G=(V,E)$. If $\Btree$ is returned by Algorithm~\ref{alg_block_tree}, then every extension $T_{\Btree}$ of $\Btree$ is a tree-layout of $G$. And moreover, if $B$, $B'$, and $B''$ are three blocks of $\mathcal{B}$ such that $B\prec_{\Btree} B'\prec_{\Btree} B''$, then for every $y\in B$, $y'\in B'$, $y''\in B''$, $yy''\in E$  implies that $y'y\in E$ and $y'y''\in E$.
\end{observation}
\begin{proof}
The fact that $T_{\Btree}$ is a tree-layout of $G$ directly follows from property (2) of the invariant of Algorithm~\ref{alg_block_tree}. Concerning triples of vertices such as $y$, $y'$ and $y''$, observe first that $y\prec_{T_\Btree} y'\prec_{T_\Btree} y''$. The fact that $yy''\in E$  implies that $y'y\in E$ and $y'y''\in E$ follows from \autoref{def_block} of a $S$-block and the fact that at the step $B'$ is selected, then $B''$ is contained in the same connected component as $B'$.
\end{proof}

%
%

\paragraph{Nested sets.} 
From \autoref{obs_alg_blocktree}, an extension of $\Btree$ is not yet an indifference tree-layout of $G$. However, if $G$ is a proper chordal graph that has an indifference tree-layout $\T_G=(T,r,\rho_G)$ rooted at $x$, then, by \autoref{lem_alg_block_tree}, $\Btree=\Btree_G(x)$. It then follows from the proof of  \autoref{th_canonical} that $\T_G$ is an extension of $\Btree$. The second step of the algorithm consists in testing if $\Btree$ has an extension that is an indifference tree-layout. 

By \autoref{lem_alg_block_tree}, we can assume that Algorithm~\ref{alg_block_tree} has returned $\mathcal{B}_G(x)$ and $\Btree_G(x)$. To every block $B$ of $\mathcal{B}_G(x)$, we assign  a collection of nested subsets of $2^B$ which we denote $\mathcal{C}_B=\langle\mathcal{N}_1,\dots,\mathcal{N}_k\rangle$ (see Algorithm~\ref{alg_recognition} for a definition of $\mathcal{C}_B$). The task of the second step of the recognition algorithm is to verify that for every block $B$, $\textsf{Nested-Convex}(\mathcal{C}_B,\mathcal{S}_B)\neq\emptyset$ where $\mathcal{S}_B=\cup_{1\leq i\leq k}\mathcal{N}_i$. This amounts to testing whether every block of $\Btree_G(x)$ satisfies \autoref{lem_nested_convex}.

To define the collection $\mathcal{C}_B$, we need some notations. Let  $\mathcal{A}_B$ denote the subset of $\mathcal{B}_G(x)$ such that if $B'\in \mathcal{A}_B$, then $B'$ is an ancestor of $B$ in $\Btree_G(x)$. Then we set $A_B=\cup_{B'\in\mathcal{A}_B} B'$ and denote $C_B$ the connected component of $G-A_B$ containing $B$.

\begin{algorithm}[h]
{\small 
\KwIn{A graph $G=(V,E)$;}
\KwOut{Decide if $G$ is a proper chordal graph.}
\BlankLine
\ForEach{$x\in V$}{
	\If{Algorithm~\ref{alg_block_tree} applied on $G$ and $x$ returns $\Btree_G(x)$}{
		\ForEach{block $B\in\mathcal{B}_G(x)$}{
			let $C_{1},\dots,C_{k}$ be the connected components of $G[C_B]-B$\;
			\ForEach{$C_i$ in $C_1,\dots , C_k$}{
				$\mathcal{N}_i\leftarrow\{X\subseteq B\mid \exists y\in C_i, X=N(y)\cap B\}$\;
				}
			set $\mathcal{C}_B=\langle \mathcal{N}_1,\dots, \mathcal{N}_k\rangle$ and ${\cal S}_{B} = \bigcup_{i \in [k]} {\cal N}_{i}$\;
			}
		\If{$\forall B\in\mathcal{B}(x)$, $\mathcal{C}_B$ is a collection of nested sets st. $\textsf{Nested-Convex}(\mathcal{C}_B,\mathcal{S}_B)\neq\emptyset$	}{
			\textbf{stop} and \Return{$G$ is a proper chordal graph}\;
			}
		}
	}
\Return{$G$ is not a proper chordal graph}\;
\caption{Proper chordal graph recognition} \label{alg_recognition}
}
\end{algorithm}

\begin{theorem} \label{th_recognition}
We can decide in polynomial time whether a graph $G=(V,E)$ is proper chordal. Moreover, if $G$ is proper chordal, an indifference tree-layout of $G$ can be constructed in polynomial time.
\end{theorem}
\begin{proof}
We prove that Algorithm~\ref{alg_recognition} recognizes proper chordal graphs. Suppose first that $G$ is proper chordal. Then for some vertex $x \in V$ there is an indifference tree-layout $\T_G= (T, r, \rho_G)$ rooted at $x = \rho^{-1}(r)$. By \autoref{lem_alg_block_tree}, Algorithm~\ref{alg_block_tree} computes $\mathcal{B}_G(x)$ and $\Btree_G(x)$.
Let $B \in {\cal B}_{G}(x)$. Then, we can additionally observe that for some vertex $y \in B$, $B=B_{\T_G}(y)\in {\cal B}_{\T_G}(G)$ and thereby 
${\cal C}_{B} = {\cal C}_{\T_G}(y)$. Thus Lemma \ref{lem_nested_convex} implies that $\textsf{Nested-Convex}(\mathcal{C}_{B},\mathcal{S}_{B})\neq \emptyset$. It follows that Algorithm~\ref{alg_recognition} correctly concludes that $G$ is a proper chordal graph.

Now suppose that Algorithm~\ref{alg_recognition} concludes that $G$ is a proper chordal graph. This implies that for some vertex $x\in V$: Algorithm~\ref{alg_block_tree} returns $\Btree_G(x)$; and for every block $B\in\mathcal{B}(x)$, $\mathcal{C}_B$ is a collection of nested sets such that $\textsf{Nested-Convex}(\mathcal{C}_B,\mathcal{S}_B)\neq\emptyset$. Now, using the proof of \autoref{th_canonical} and the construction therein, it is possible to compute (in polynomial time) from $\Btree_G(x)$ and $G$, an indifference tree-layout of $G$ rooted at $x$, certifying that $G$ is proper chordal.

Finally observe that every step of Algorithm~\ref{alg_recognition} (including Algorithm~\ref{alg_block_tree}, see \autoref{lem_alg_block_tree}) requires polynomial time.
\end{proof}

\subsection{Isomorphism}

We observe that, because an \FPQ-hierarchy does not carry enough information to reconstruct the original graph, two non-isomorphic proper chordal graphs $G$ and $G'$ may share an \FPQ-hierarchy (\autoref{fig_non_isomorphic}) satisfying the conditions of \autoref{th_canonical}. 
 More precisely, given an \FPQ-hierarchy of a proper chordal graph $G$ that satisfies the conditions of \autoref{th_canonical}, one can reconstruct an indifference tree-layout $\T$ of $G$. But $\T$ is not sufficient to test the adjacency between a pair of vertices. Indeed, for a given vertex $y$, we cannot retrieve $N(y)\cap A_{\T}(y)$ from $\H_{G}(x)$ since only the intersection of $N(y)$ with the parent block is present in $\H_{G}(x)$. In the example of \autoref{fig_non_isomorphic}, the vertices $w$ and $w'$ are also adjacent to $c$ and $b$, which do not belong to their parent block.

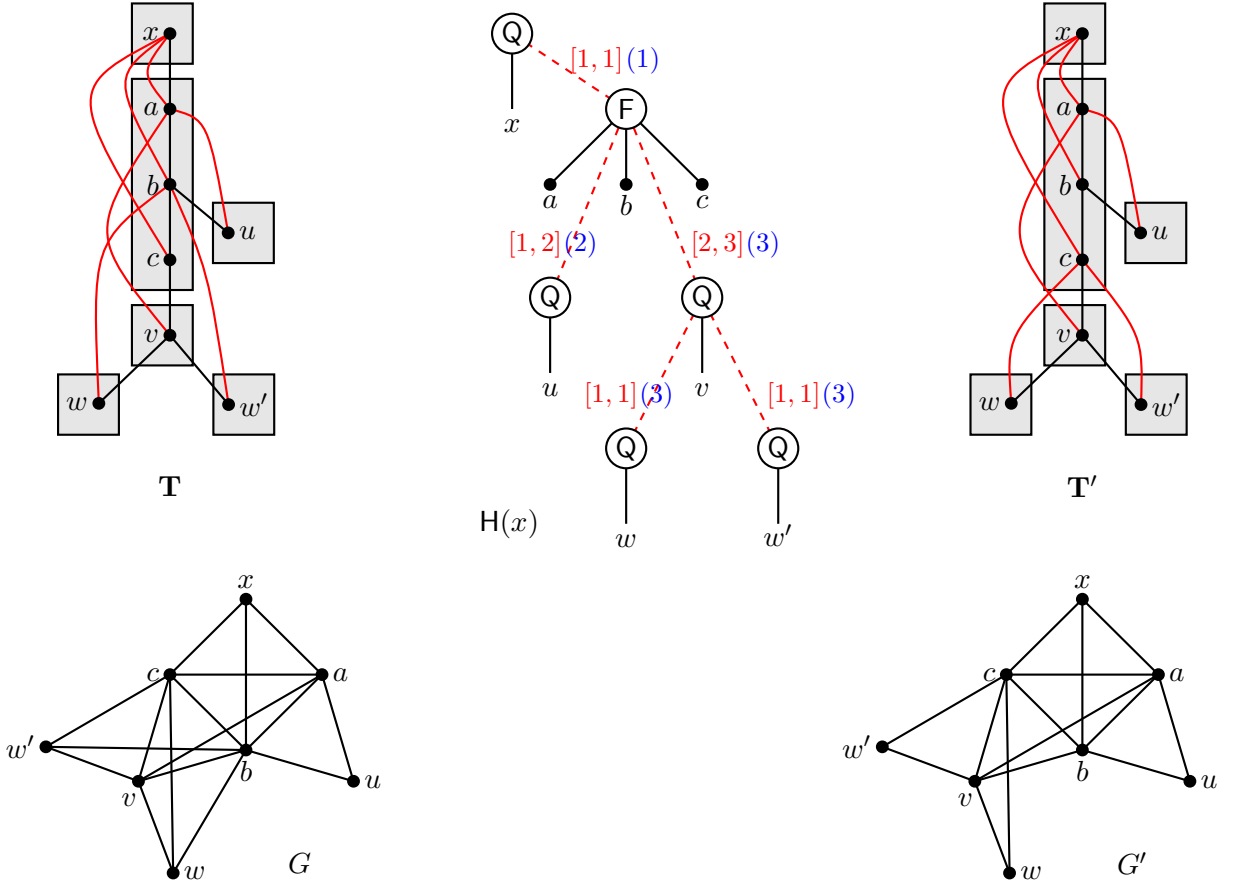
\begin{figure}[h]
\begin{center}
\begin{tikzpicture}[thick,scale=1]
\tikzstyle{sommet}=[circle, draw, fill=black, inner sep=0pt, minimum width=4pt]          

\tikzstyle{Onode}=[circle, draw, fill=white, inner sep=0pt, minimum width=15pt]          
\tikzstyle{Pnode}=[circle, draw, fill=white, inner sep=0pt, minimum width=15pt]          
\tikzstyle{Qnode}=[circle, draw, fill=white, inner sep=0pt, minimum width=15pt]

\begin{scope}[]
\draw[fill=gray!20] (-0.5,4.6) rectangle (0.3,5.4) ;
\draw[fill=gray!20] (-0.5,1.6) rectangle (0.3,4.4) ;
\draw[fill=gray!20] (-0.5,0.6) rectangle (0.3,1.4) ;

\begin{scope}[xshift=-0.4cm,yshift=1cm]
\draw[shift=(-130:1.2),fill=gray!20] (-0.3,-0.4) rectangle (0.5,0.4) ;
\end{scope}

\begin{scope}[shift=(90:1)]
\draw[shift=(-50:1.2),fill=gray!20] (-0.2,-0.4) rectangle (0.6,0.4) ;
\end{scope}

\begin{scope}[shift=(90:3)]
\draw[shift=(-40:1),fill=gray!20] (-0.2,-0.4) rectangle (0.6,0.4) ;
\end{scope}

\draw  (90:5) -- (90:4) -- (90:3) -- (90:2) -- (90:1)  ;
\draw[shift=(90:3)] (0:0) -- (-40:1) ;
\draw[shift=(90:1)] (0:0) -- (-50:1.2) ;
\draw[shift=(90:1)] (0:0) -- (-136:1.3) ;

\draw[red,thick] (90:2) .. controls (108:4.5) .. (90:5) ;
\draw[red,thick] (90:3) .. controls (100:4.5) .. (90:5) ;
\draw[red,thick] (90:4) .. controls (95:4.5) .. (90:5) ;
\draw[red,thick] (90:1) .. controls (115:2.6) .. (90:4) ;
\draw[red,thick,shift=(90:3)] (-40:1) .. controls (60:1) .. (90:1) ;
\draw[red,thick,shift=(90:1)] (-50:1.2) .. controls (60:1) .. (90:2) ;
\draw[red,thick,shift=(90:1)] (-136:1.3) .. controls (130:1.6) .. (90:2) ;

\node[left] (y) at (90:5) {$x$};
\node[left] (5) at (90:4) {$a$};
\node[left] (4) at (90:3) {$b$};
\node[left] (3) at (90:2) {$c$};
\node[left] (2) at (90:1) {$v$};
\node[right,shift=(90:3)] (u) at (-40:1) {$u$};
\node[left,shift=(90:1)] (w) at (-136:1.3) {$w$};
\node[right,shift=(90:1)] (w') at (-50:1.2) {$w'$};

\draw[shift=(90:3)] (-40:1) node[sommet]{};
\draw[shift=(90:1)] (-50:1.2) node[sommet]{};
\draw[shift=(90:1)] (-136:1.3) node[sommet]{};
\draw (90:5) node[sommet]{};
\draw (90:4) node[sommet]{};
\draw (90:3) node[sommet]{};
\draw (90:2) node[sommet]{};
\draw (90:1) node[sommet]{};

\node (T) at (90:-1) {$\T$};

\end{scope}

\begin{scope}[xshift=1cm,yshift=-3.5cm,rotate=-90]
\draw  (90:1) -- (90:-1)  ;
\draw  (0:1) -- (180:1)  ;
\draw  (0:1) -- (90:1)  ;
\draw  (0:1) -- (90:-1)  ;
\draw  (180:1) -- (90:1)  ;
\draw  (180:1) -- (90:-1)  ;
\draw  (45:2) -- (90:1)  ;
\draw  (45:2) -- (0:1)  ;
\draw  (-45:2) -- (-90:1)  ;
\draw  (-45:2) -- (0:1)  ;
\draw (-20:2.8) -- (-45:2) ;
\draw (-20:2.8) -- (0:1) ;
\draw (-20:2.8) -- (-90:1) ;

\draw (-70:2.8) -- (-45:2) ;
\draw (-70:2.8) -- (0:1) ;
\draw (-70:2.8) -- (-90:1) ;
\draw (-45:2) -- (90:1) ;

\draw (0:1) node[sommet]{};
\draw (180:1) node[sommet]{};
\draw (90:1) node[sommet]{};
\draw (-90:1) node[sommet]{};
\draw (45:2) node[sommet]{};
\draw (-45:2) node[sommet]{};
\draw (-20:2.8) node[sommet]{};
\draw (-70:2.8) node[sommet]{};

\node[above] (xx) at (180:1) {$x$};
\node[right] (aa) at (90:1) {$a$};
\node[below] (bb) at (0:1) {$b$};
\node[left] (cc) at (-90:1) {$c$};
\node[right] (uu) at (45:2) {$u$};
\node[below] (vv) at (-47:2.1) {$v$};
\node[right] (ww) at (-20:2.8) {$w$};
\node[left] (ww') at (-70:2.8) {$w'$};

\node[left] (G) at (2.5,1) {$G$};

\end{scope}

\begin{scope}[shift=(0:12)]
\draw[fill=gray!20] (-0.5,4.6) rectangle (0.3,5.4) ;
\draw[fill=gray!20] (-0.5,1.6) rectangle (0.3,4.4) ;
\draw[fill=gray!20] (-0.5,0.6) rectangle (0.3,1.4) ;

\begin{scope}[xshift=-0.4cm,yshift=1cm]
\draw[shift=(-130:1.2),fill=gray!20] (-0.3,-0.4) rectangle (0.5,0.4) ;
\end{scope}

\begin{scope}[shift=(90:1)]
\draw[shift=(-50:1.2),fill=gray!20] (-0.2,-0.4) rectangle (0.6,0.4) ;
\end{scope}

\begin{scope}[shift=(90:3)]
\draw[shift=(-40:1),fill=gray!20] (-0.2,-0.4) rectangle (0.6,0.4) ;
\end{scope}

\draw  (90:5) -- (90:4) -- (90:3) -- (90:2) -- (90:1)  ;
\draw[shift=(90:3)] (0:0) -- (-40:1) ;
\draw[shift=(90:1)] (0:0) -- (-50:1.2) ;
\draw[shift=(90:1)] (0:0) -- (-136:1.3) ;

\draw[red,thick] (90:2) .. controls (108:4.5) .. (90:5) ;
\draw[red,thick] (90:3) .. controls (100:4.5) .. (90:5) ;
\draw[red,thick] (90:4) .. controls (95:4.5) .. (90:5) ;
\draw[red,thick] (90:1) .. controls (115:2.6) .. (90:4) ;
\draw[red,thick,shift=(90:3)] (-40:1) .. controls (60:1) .. (90:1) ;
\draw[red,thick,shift=(90:1)] (-50:1.2) .. controls (0:0.8) .. (90:1) ;
\draw[red,thick,shift=(90:1)] (-136:1.3) .. controls (-180:1) .. (90:1) ;

\node[left] (y) at (90:5) {$x$};
\node[left] (5) at (90:4) {$a$};
\node[left] (4) at (90:3) {$b$};
\node[left] (3) at (90:2) {$c$};
\node[left] (2) at (90:1) {$v$};
\node[right,shift=(90:3)] (u) at (-40:1) {$u$};
\node[left,shift=(90:1)] (w) at (-136:1.3) {$w$};
\node[right,shift=(90:1)] (w') at (-50:1.2) {$w'$};

\draw[shift=(90:3)] (-40:1) node[sommet]{};
\draw[shift=(90:1)] (-50:1.2) node[sommet]{};
\draw[shift=(90:1)] (-136:1.3) node[sommet]{};
\draw (90:5) node[sommet]{};
\draw (90:4) node[sommet]{};
\draw (90:3) node[sommet]{};
\draw (90:2) node[sommet]{};
\draw (90:1) node[sommet]{};

\node (T') at (90:-1) {$\T'$};

\end{scope}

\begin{scope}[xshift=12cm,yshift=-3.5cm,rotate=-90]
\draw  (90:1) -- (90:-1)  ;
\draw  (0:1) -- (180:1)  ;
\draw  (0:1) -- (90:1)  ;
\draw  (0:1) -- (90:-1)  ;
\draw  (180:1) -- (90:1)  ;
\draw  (180:1) -- (90:-1)  ;
\draw  (45:2) -- (90:1)  ;
\draw  (45:2) -- (0:1)  ;
\draw  (-45:2) -- (-90:1)  ;
\draw  (-45:2) -- (0:1)  ;
\draw (-20:2.8) -- (-45:2) ;
\draw (-45:2) -- (90:1) ;
\draw (-70:2.8) -- (-90:1) ;
\draw (-20:2.8) -- (-90:1) ;

\draw (-70:2.8) -- (-45:2) ;

\draw (0:1) node[sommet]{};
\draw (180:1) node[sommet]{};
\draw (90:1) node[sommet]{};
\draw (-90:1) node[sommet]{};
\draw (45:2) node[sommet]{};
\draw (-45:2) node[sommet]{};
\draw (-20:2.8) node[sommet]{};
\draw (-70:2.8) node[sommet]{};

\node[above] (xx) at (180:1) {$x$};
\node[right] (aa) at (90:1) {$a$};
\node[below] (bb) at (0:1) {$b$};
\node[left] (cc) at (-90:1) {$c$};
\node[right] (uu) at (45:2) {$u$};
\node[below] (vv) at (-47:2.1) {$v$};
\node[right] (ww) at (-20:2.8) {$w$};
\node[left] (ww') at (-70:2.8) {$w'$};

\node[left] (G) at (2.5,1) {$G'$};

\end{scope}

\begin{scope}[xshift=5cm,yshift=1cm]

\begin{scope}[xshift=0cm,yshift=4cm]
\draw (-0.5,0) -- (-0.5,-1);

\draw[red,thick,dashed] (-0.5,0) -- (1,-1) ;
\node[red] (Iy) at (-30:0.7) {\small $[1,1]$};
\node[blue] (Iu) at (-16:1.27) {\small{$(1)$}};

\draw (-0.5,0) node[Qnode]{};
\node[] (p) at (-0.5,0) {\textsf{Q}};
\node[below] (a) at (-0.5,-1) {$x$};
\end{scope}


\begin{scope}[xshift=0cm,yshift=0.5cm]
\draw[red,thick,dashed] (0,0) -- (1,2.5) ;
\node[red] (Iu) at (-0.2,0.7) {\small{$[1,2]$}};
\node[blue] (Iu) at (0.4,0.7) {\small{$(2)$}};
\end{scope}


\begin{scope}[xshift=0cm,yshift=0.5cm]
\draw[red,thick,dashed] (2,0) -- (1,2.5) ;
\node[red] (Iu) at (2.2,0.7) {\small{$[2,3]$}};
\node[blue] (Iu) at (2.8,0.7) {\small{$(3)$}};
\end{scope}


\begin{scope}[xshift=1cm,yshift=-1.5cm]
\draw[red,thick,dashed] (0,0) -- (1,2) ;
\node[red] (Iu) at (-0.2,0.7) {\small{$[1,1]$}};
\node[blue] (Iu) at (0.4,0.7) {\small{$(3)$}};
\end{scope}


\begin{scope}[xshift=3cm,yshift=-1.5cm]
\draw[red,thick,dashed] (0,0) -- (-1,2) ;
\node[red] (Iu) at (0.2,0.7) {\small{$[1,1]$}};
\node[blue] (Iu) at (0.8,0.7) {\small{$(3)$}};
\end{scope}

\begin{scope}[shift=(90:2)]
\draw (1,1) -- (0,0);
\draw (1,1) -- (1,0);
\draw (1,1) -- (2,0);

\draw (0,0) node[sommet]{};
\draw (1,0) node[sommet]{};
\draw (2,0) node[sommet]{};

\draw (1,1) node[Onode]{};
\node[] (p) at (1,1) {\textsf{F}};

\node[below] (a) at (0,0) {$a$};
\node[below] (b) at (1,0) {$b$};
\node[below] (c) at (2,0) {$c$};
\end{scope}

\begin{scope}[xshift=0cm,yshift=0.5cm]

\draw (0,0) -- (0,-1);
\draw (0,0) node[Qnode]{};
\node[] (pu) at (0,0) {\textsf{Q}};
\node[below] (a) at (0,-1) {$u$};

\end{scope}

\begin{scope}[xshift=2cm,yshift=0.5cm]
\draw (0,0) -- (0,-1);

\draw (0,0) node[Qnode]{};
\node[] (pv) at (0,0) {\textsf{Q}};
\node[below] (a) at (0,-1) {$v$};
\end{scope}

\begin{scope}[xshift=1cm,yshift=-1.5cm]

\draw (0,0) -- (0,-1);
\draw (0,0) node[Qnode]{};
\node[] (p) at (0,0) {\textsf{Q}};
\node[below] (a) at (0,-1) {$w$};
\end{scope}

\begin{scope}[xshift=3cm,yshift=-1.5cm]

\draw (0,0) -- (0,-1);
\draw (0,0) node[Qnode]{};
\node[] (p) at (0,0) {\textsf{Q}};
\node[below] (a) at (0,-0.86) {$w'$};
\end{scope}

\node[left] (H) at (0,-2.5) {$\H(x)$};

\end{scope}

\end{tikzpicture}
\end{center}
\caption{\label{fig_non_isomorphic} Two proper chordal graphs $G$ and $G'$ with their respective indifference tree-layouts $\T$ and $\T'$. We observe that $G$ and $G'$ are not isomorphic, but their respective skeleton trees $T_G(x)$ and $T_{G'}(x)$ are. Moreover, $\H(x)$ is an \FPQ-hierarchy such that $\mathfrak{T}_{\FPQ}(\H(x))$ contains all the indifference tree-layouts rooted at $x$ of $G$ and of $G'$. We obtain $\H^*(x)$ for $G$ by adding to $\H(x)$ the blue labels on the skeleton edges.}
\end{figure}

\newcommand{\A}{\hat{A}}

Let $\H_{G}(x)$ be the \FPQ-hierarchy satisfying the conditions of \autoref{th_canonical}, we define the \emph{indifference \FPQ-hierarchy}, denoted $\H^*_{G}(x)$, obtained from $\H_{G}(x)$ by adding to every skeleton edge $e$, a label $\A(e)$. Suppose that $e$ is incident to the root of the \FPQ-tree of the block $B$, then 
we set $\A(e)=|N(B)\cap A_{\T}(B)|$. We say that two indifference \FPQ-hierarchies $\H^*_1$ and $\H^*_2$ are equivalent, denoted $\H^*_1\approx^*_{\FPQ} \H^*_2$, if $\H_1\approx_{\FPQ} \H_2$ and for every pair of mapped skeleton edges $e_1$ and $e_2$ we have $\A(e_1)=\A(e_2)$.

Let $\mathcal{S}_1\in 2^{X_1}$ be a set of subsets of $X_1$ and $\mathcal{S}_2 \in 2^{X_2}$ be a set of subsets of $X_2$. We say that $\mathcal{S}_1$ and $\mathcal{S}_2$ are \emph{isomorphic} if there exists a bijection $f:X_1\rightarrow X_2$ such that $S_1\in\mathcal{S}_1$ if and only if $S_2=\{f(x)\mid x\in S_1\}\in \mathcal{S}_2$. For $S_1\subseteq X_1$, we denote by $f(S_1)=\{f(y)\mid y\in S_1\}$.

\begin{lemma} \label{lem_isomorphism_equivalence}
Let $G_1=(V_1,E_1)$ and $G_2=(V_2,E_2)$  be two (connected) proper chordal graphs. Let $\H^*_1(x_1)$ be an indifference \FPQ-hierarchy of $G_1$ and  $\H^*_2(x_2)$ be an indifference \FPQ-hierarchy of $G_2$. Then $\H^*_1(x_1)\approx^*_{\FPQ} \H^*_2(x_2)$ if and only if $G_1$ and $G_2$ are isomorphic with $x_1$ mapped to $x_2$.
\end{lemma}
\begin{proof}
Suppose that $f$ is a graph isomorphism between $G_1$ and $G_2$ such that $f(x_1)=x_2$. Let $\T_1=(T,r,\rho)$ be an indifference tree-layout of $G_1$ such that $\rho^{-1}(r)=x_1$. Then as for every pair of vertices $y$ and $z$, $yz\in E_1$ if and only if $f(y)f(z)\in E_2$, $\T_2=(T,r,\rho\circ f)$ is an indifference tree-layout of $G_2$ such that $(\rho\circ f)^{-1}(r)=x_2$.
We let  $\H_1(x_1)$ and  $\H_2(x_2)$ respectively denote the  \FPQ-hierarchy of $G_1$ and $G_2$ such that $\mathfrak{T}_{\FPQ}( \H_1(x_1))$ contains $\T_1$ and  $\mathfrak{T}_{\FPQ}( \H_2(x_2))$ contains $\T_2$.

We also observe that $\mathcal{B}_{\T_1}(G_1)$ and $\mathcal{B}_{\T_2}(G_2)$ are isomorphic partitions of $V_1$ and $V_2$ respectively.
It follows that the skeleton trees $\Btree_{G_1}(x_1)$ and $\Btree_{G_2}(x_2)$ are isomorphic trees. 
So if we consider $y_1\in V_1$ such that $B_{\T_1}(y_1)\in \mathcal{B}_{\T_1}(G_1)$, then
$B_{\T_2}(y_2)\in \mathcal{B}_{\T_2}(G_2)$ where $y_2=f(y_1)$. Let $C_{\T_1}(x_1)$ be the connected component of $G_1-A_{\T_1}(x_1)$ containing $x_1$. Then the connected component of $G_2-A_{\T_2}(x_2)$ is the set $C_{\T_2}(x_2)=f(C_{\T_1}(x_1))$ and the collections $\mathcal{C}_{\T_1}(x_1)$ and $\mathcal{C}_{\T_2}(x_2)$ of nested sets are isomorphic. We set $\mathcal{S}_1=\bigcup_{\mathcal{N}_1\in\mathcal{C}_{\T_1}(x_1)} \mathcal{N}_1$ and $\mathcal{S}_2=\bigcup_{\mathcal{N}_2\in\mathcal{C}_{\T_2}(x_2)} \mathcal{N}$ and let denote  $\mathsf{T}_1$ an \FPQ-tree such that $\mathfrak{S}_{\T_1}=\textsf{Nested-Convex}(\mathcal{C}_{\T_1}(x_1),\mathcal{S}_1)$ and $\mathsf{T}_2$ an \FPQ-tree such that $\mathfrak{S}_{\T_2}=\textsf{Nested-Convex}(\mathcal{C}_{\T_2}(x_2),\mathcal{S}_2)$. Then, as $\mathcal{S}_1$ and $\mathcal{S}_2$ are isomorphic, we have that $\mathsf{T}_1\equiv_{\FPQ{}}\mathsf{T}_2$. In turn, this implies that $\mathsf{H}_1(x_1)\approx_{\FPQ}\mathsf{H}_2(x_2)$. 

To conclude, we observe that as $\mathcal{B}_{\T_1}(G_1)$ and $\mathcal{B}_{\T_2}(G_2)$ are isomorphic partitions and   $\Btree_{G_1}(x_1)$ and  $\Btree_{G_2}(x_2)$  are isomorphic trees, for every skeleton edge $e_1$ of  $\Btree_{G_1}(x_1)$, the label $A(e_1)$ is equal to the label $A(e_2)$ where $e_2$ is the corresponding skeleton edge of  $\Btree_{G_2}(x_2)$. Thereby we can conclude that if $G_1$ and $G_2$ are isomorphic graphs, then $\mathsf{H}^*_1(x_1)\approx^*_{\FPQ}\mathsf{H}^*_2(x_2)$. 

\medskip
Suppose now that $\mathsf{H}^*_1(x_1)\approx^*_{\FPQ}\mathsf{H}^*_2(x_2)$. It follows that $\mathsf{H}_2(x_2)$ can be turned into an \FPQ-hierarchy $\widetilde{\mathsf{H}}_2(x_2)$ isomorphic to $\mathsf{H}_1(x_1)$. Let $\T_1$ and $\T_2$ be respectively the trees that can be computed from $\widetilde{\mathsf{H}}_2(x_2)$ and $\mathsf{H}_1(x_1)$. By Claim~\ref{cl_indifference_treelayout} in the proof of \autoref{th_canonical}, $\T_1$ and $\T_2$ are respectively indifference tree-layouts of $G_1$ and $G_2$. We observe that $\T_1$ and $\T_2$ are isomorphic rooted trees, respectively rooted at $x_1$ and $x_2$.
Let $f$ be the isomorphism from  $\T_1$ to $\T_2$ 
We claim that $f$ is an isomorphism from $G_1$ to $G_2$. Let $y_1$ and $z_1$ be two vertices of $G_1$. We denote $y_2=f(y_1)$ and $z_2=f(z_1)$ the corresponding vertices of $G_2$. First as $\T_1$ and $\T_2$ are tree-layouts, if $y_1\not\prec_{\T_1} z_1$  and $z_1\not\prec_{\T_1} y_1$, then $y_1z_1\notin E_1$. This implies that $y_2z_2\notin E_2$. Suppose without loss of generality that $y_1\prec_{\T_1} z_1$, and hence $y_2\prec_{\T_2} z_2$.
As $\mathsf{H}^*_1(x_1)\approx^*_{\FPQ}\mathsf{H}^*_2(x_2)$, we have that $\mathcal{B}_{\T_1}(G_1)$ and $\mathcal{B}_{\T_2}(G_2)$ are isomorphic. It follows that $y_1$ and $z_1$ belong to the same block of $\mathcal{B}_{\T_1}(G_1)$ if and only if $y_2$ and $z_2$ belong to the same block of $\mathcal{B}_{\T_2}(G_2)$.
As blocks induce cliques, if $y_1$ and $z_1$ belong to the same block of $B_1\in\mathcal{B}_{\T_1}(G_1)$, then $y_1z_1\in E_2$ and $y_2z_2\in E_2$.
So suppose that the block of $B_1\in\mathcal{B}_{\T_1}(G_1)$ containing $z_1$ does not contain $y_1$. Let $B_2\in\mathcal{B}_{\T_2}(G_2)$ be the corresponding block that contains $z_2$ but not $y_2$. Let $e_1$, respectively $e_2$, be the skeleton edge of $\H_1$, respectively $\H_2$ incident to the root of the \FPQ-tree of $B_1$, respectively of $B_2$. As $\mathsf{H}^*_1(x_1)\approx^*_{\FPQ}\mathsf{H}^*_2(x_2)$, we have $\A(e_1)=\A(e_2)$. We observe that $y_1z_1\in E_1$ if and only if the distance along the path from $x_1$ to $z_1$ in $\T_1$, between $y_1$ and $z_1$ is at most $\A(e_1)+|A_{\T_1}(z_1)\cap B_1|$. Likewise, $y_2z_2\in E_2$ if and only if the distance along the path from $x_2$ to $z_2$ in $\T_2$, between $y_2$ and $z_2$ is at most $\A(e_2)+|A_{\T_2}(z_2)\cap B_2|$. Since $\T_1$ and $\T_2$ are isomorphic, and since $\mathcal{B}_{\T_1}(G_1)$ and $\mathcal{B}_{\T_2}(G_2)$ are isomorphic, this implies that $|A_{\T_1}(z_1)\cap B_1|=|A_{\T_2}(z_2)\cap B_2|$. It follows that $y_1z_1\in E_1$ if and only if $y_2z_2\in E_2$. We conclude that $G_1$ and $G_2$ are isomorphic graphs.
\end{proof}

From \autoref{lem_isomorphism_equivalence},  testing graph isomorphism on proper chordal graphs reduces to testing the equivalence between two indifference \FPQ-hierarchies. To that aim, we use a similar approach to the one developed for testing interval graph isomorphism~\cite{LuekerB79}. That is, we adapt the standard unordered tree isomorphism algorithm that assigns to every unordered tree a canonical \emph{ isomorphism code}~\cite{Valiente02,AhoHU74}. Testing isomorphism then amounts to testing equality between two isomorphism codes.

\newcommand{\code}{\mathsf{code}}
\newcommand{\size}{\mathsf{size}}
\newcommand{\info}{\mathsf{label}}
\newcommand{\type}{\mathsf{type}}
\newcommand{\inflex}{<_{\mathsf{lex}}}

Let $\H^*$ be an indifference \FPQ-hierarchy of a proper chordal graph $G=(V,E)$. Intuitively, the isomorphism code of $\H^*$ is a string obtained by concatenating information about the root node of $\H^*$ and the isomorphism codes of the sub-hierarchies rooted at its children. To guarantee the canonicity of the isomorphism code of  $\H^*$, some of the codes of these sub-hierarchies need to be sorted lexicographically. To that aim, we use the following convention:
$$\mathsf{L}\inflex \mathsf{F}\inflex\mathsf{P}\inflex\mathsf{Q}\inflex 0 \inflex 1 \dots \inflex n \inflex \dots,$$
Moreover the separating symbols (such as brackets, commas\dots) used in the isomorphism code for the sake of readability are irrelevant for the sort.

Before formally describing the isomorphism code of $\H^*$, let us remind that, in an indifference \FPQ-hierarchy, we can classify the children of any node $t$ in two categories: we call a node $t'$ a \emph{skeleton child} of $t$ if the tree edge $e=tt'$ is a skeleton edge of $\H^*$, otherwise we call it a \emph{block child} of $t$.  We observe that the block children  of a node $t$ belong with $t$ to the \FPQ-tree of some block of $\mathcal{B}_G(x)$. It follows from the definition of an \FPQ-tree, that the block children of a given node $t$ are ordered and depending on the type of $t$, these nodes can be reordered. On the contrary, the skeleton children of a node $t$ are not ordered. 

For every node $t$ of $\H^*$, we define a code, denoted $\code(t)$. We will define the isomorphism code of $\H^*$ as $\code(\H^*)=\code(r)$, where $r$ is the root node of $\H^*$.  We let $b_1, \dots, b_k$ denote the block children of node $t$ (if any, and ordered from $1$ to $k$) and $s_1, \dots, s_{\ell}$ denote the skeleton children of $t$ (if any). For a node $t$, the set of \emph{eligible permutations} of the indices $[1,k]$ of its chidlren depends on $\type(t)$:
\begin{itemize}
\item if $\type(t)=\mathsf{F}$, then the identity permutation is the unique eligible permutation;
\item if $\type(t)=\mathsf{P}$, then every permutation is eligible;
\item if $\type(t)=\mathsf{Q}$, then the identity or its reverse permutation are the two eligible permutations.
\end{itemize}

The code of $t$, denoted $\code(t)$, is obtained by minimizing with respect to $\inflex$ over all eligible permutations $\beta$ of $t$:

$$\mathsf{code}(t,\beta)=\left\{
\begin{array}{l}
\size(t)\circ\type(t)\circ~\\
~~~~~~~~\code(b_{\beta(1)})\circ \dots \circ \code(b_{\beta(k)}) \circ~\\
~~~~~~~~~~~~~~~~~~~~~~\info(s_{\pi_{\beta}(1)},\beta)\circ\code(s_{\pi_{\beta}(1)})\circ\dots\circ\info(s_{\pi_{\beta}(\ell)},\beta)\circ\code(s_{\pi_{\beta}(\ell)})
\end{array}
\right.
$$
where:
\begin{itemize}
\item $\type(t)\in\{\mathsf{L},\mathsf{F},\mathsf{P},\mathsf{Q}\}$, indicates whether $t$ a leaf ($\mathsf{L}$), a $\mathsf{F}$-node, a $\mathsf{P}$-node, or a $\mathsf{Q}$-node.
\item $\size(t)\in\mathbb{N}$, stores the number of nodes in the sub-hierarchy rooted at $t$ (including $t$).
\item $\info(s,\beta)$, with $s$ being a skeleton child $s$ of $t$ and $\beta$ being a permutation of $[1,k]$.  Let $e$ be the skeleton edge of $\H^*$ between $s$ and $t$. If $I(e)=[a,b]$, we set $I^c(e)=[k+1-b,k+1-a]$. Then, we set $\info(s)=\langle I^{\beta}(e),A(e)\rangle$, where 
$$I^{\beta}(e)=\left\{
\begin{array}{cl}
I(e), & \mbox{$\beta$ is the identity permutation} \\
I^c(e), & \mbox {otherwise.}
\end{array}
\right.
$$

\item $\pi_{\beta}$ is, for some permutation $\beta$ of $[1,k]$, a permutation of $[1,\ell]$ that minimizes, with respect to $\inflex$: 
$$\info(s_{\pi_{\beta}(1)},\beta)\circ\code(s_{\pi_{\beta}(1)})\circ\dots\circ\info(s_{\pi_{\beta}(\ell)},\beta)\circ\code(s_{\pi_{\beta}(\ell)}).$$
\end{itemize}

\medskip
\noindent
\textbf{Remark.} Observe that for a $\mathsf{P}$-node $t$, the permutation $\pi_{\beta}$ is independant of the choice of $\beta$. Indeed, thanks to \autoref{lem_partitive}, $I(e)=[1,k]$ and thereby $I^{\beta}(e)=I(e)$. So, if $t$ is an $\mathsf{F}$-node or a $\mathsf{P}$-node, then $\code(t)$ can be greedily computed in polynomial time. If $t$ is a $\mathsf{Q}$-node, $\code(t)$ can still be computed in polynomial time by comparing $\code(t,\beta)$ and $\code(t,\overline{\beta})$ where $\beta$ is the identity permutation and $\overline{\beta}$ its reverse.

\begin{lemma} \label{lem_isomorphism_code}
Let $\H_1^*$ and $\H_2^*$ be indifference \FPQ-hierarchies of the graphs $G_1$ and $G_2$ respectively.
Then $\H_1^* \approx^*_{\FPQ} \H_2^*$ if and only if $\code(\H_1^*) = \code(\H_2^*)$.
\end{lemma}
\begin{proof}
Let $\H^*$ be an indifference \FPQ-hierarchy. We claim that given $\code(\H^*)$, one can deterministically construct an \FPQ-hierarchy $\H^{*(rep)}$ such that $\H^*\approx_{\FPQ}^*\H^{*(rep)}$. To that aim, we need to parse $\code(\H^*)$ in order to identify the children of each node. This can be done using the $\size(\cdot)$ information. The smallest possible indifference \FPQ-hierarchy $\H^*$ contains a  $\mathsf{Q}$-node as root and one leaf. It can clearly be reconstructed from $\code(\H^*)=2 \mathsf{Q} 1 \mathsf{L}$. So suppose that the property is true for every indifference \FPQ-hierarchy containing $i$ nodes with $2\leq i\leq n$ and assume that $\H^*$ contains $n+1$ nodes. Let $r$ be the root of $\H^*$. Then $\size(r)\circ\type(r)\circ\size(c_1)$, where $c$ is the first child of $r$ is a prefix of $\code(\H^*)$. Then the second child $c_2$ of $r$ can be identified after reading $\size(c_1)$ times a $\type(\cdot)$ tag in $\code(\H^*)$. So every child of $r$ can be identified. Moreover, it is possible to distinguish whether a child $c_j$ of $r$ is a block child or a skeleton child by the presence or not of the string $\info(c_j)$ delimited with the special characters $\langle\cdot\rangle$ and to assign each of them the correct $\info(\cdot)$ information. Finally, observe that by construction of $\code(\H^*)$ the block children of $r$ are permuted in a way compatible with its type ($\mathsf{F}$, $\mathsf{P}$ or $\mathsf{Q}$). As every child of $r$ is the root of a sub-hierarchy containing at most $n$ nodes, by the recursive hypothesis, we can conclude that the returned \FPQ-hierarchy $\H^{*(rep)}$ verifies $\H^*\approx_{\FPQ}^*\H^{*(rep)}$. 

From the above discussion, if $\code({\H}_{1}^*) = \code({\H}_2^*)$, then ${\H}_{1}^{*(rep)}$ and $\H_2^{*(rep)}$ are isomorphic \FPQ-hierarchies (in terms of labeled ordered trees). As $\H_1^* \approx^*_{\FPQ} \H_1^{*(rep)}$ and  $\H_2^* \approx^*_{\FPQ} \H_2^{*(rep)}$, we obtain that ${\H}_{1}^{*} \approx^*_{\FPQ} {\H}_{2}^{*}$.

Suppose now that $\H_1^* \approx^*_{\FPQ} \H_2^*$. Consider $\H_1^{*(rep)}$ and $\H_2^{*(rep)}$ respectively computed from $\code({\H}_{1}^*)$ and $\code({\H}_2^*)$. It follows that $\H_1^{*(rep)}\approx^*_{\FPQ} \H_2^{*(rep)}$. Now observe that  by construction $\code(\H_1^*)=\code(\H_1^{*(rep)})$ and henceforth among the set of \FPQ-hierarchies equivalent to $\H_1^*$, $\code(\H_1^{*(rep)})$ is the smallest lexicographic one. The same holds for  $\H_1^*$ and $\H_1^{*(rep)}$. Consequently, if $\code(\H_1^*)\neq \code(\H_2^*)$, then one is smaller lexicographically than the other: a contradiction.
\end{proof}

\begin{theorem} \label{th_isomorphism}
Let $G_1$ and $G_2$ be two proper chordal graphs. One can test in polynomial time if $G_1=(V_1,E_1)$ and $G_2=(V_2,E_2)$ are isomorphic graphs.
\end{theorem}
\begin{proof}
The algorithm is working as follows. First, compute a tree-layout $\T_1$ of $G_1$ and the indifference \FPQ-hierarchy $\H_1^*$  such that $\T_1\in\mathfrak{T}_{\FPQ}(H_1)$. This can be done in polynomial time by~\autoref{th_recognition}. Then for every vertex $x_2\in V_2$, we test if there exists an indifference tree-layout $\T_2$ rooted at $x_2$; compute the corresponding indifference \FPQ-hierarchy $\H_2^*$ and test whether $\H_1^*\approx_{\FPQ}^* \H_2^*$. By \autoref{lem_isomorphism_code}, testing equivalence between \FPQ-hierarchies can be done by computing and comparing the isomorphism codes of $\H_1^*$ and $\H_2^*$. Moreover, this latter task can be achieved in polynomial time. By \autoref{lem_isomorphism_equivalence}, if one of these tests is positive, then we can conclude that $G_1$ and $G_2$ are isomorphic graphs.
\end{proof}

\section{Conclusion}

Our results demonstrate that proper chordal graphs form a rich class of graphs. First, its relative position with respect to important graph subclasses of chordal graphs and the fact that the isomorphism problem belongs to \textsf{P} for proper chordal graphs shows that they form a non-trivial potential island of tractability for many other algorithmic problems. In this line, we let open the status of Hamiltonian cycle, which is polynomial time solvable in proper interval graphs \cite{Bertossi83,Ibarra09} and interval graphs \cite{Keil85,BertoissiB86}, but \textsc{NP}-complete on strongly chordal graphs~\cite{Muller96}. We were only able to resolve the special case of split proper chordal graphs. An intriguing algorithmic question is whether proper chordal graphs can be recognized in linear time. Second, the canonical representation we obtained of the set of indifference tree-layouts rooted at some vertex witnesses the rich combinatorial structure of proper chordal graphs. We believe that this structure has to be further explored and could be important for the efficient resolution of more computational problems. For example, as proper chordal graphs form a hereditary class of graphs, one could wonder if the standard graph modification problems (vertex deletion, edge completion or deletion, etc.), which are \textsf{NP}-complete by \cite{LewisY80}, can be resolved in \textsf{FPT} time. The structure of proper chordal graphs is not yet fully understood. The first natural question on this aspect is to provide a forbidden induced subgraph characterization. This will involve infinite families of forbidden subgraphs. Furthermore, understanding what makes a vertex the root of an indifference tree-layout is certainly a key ingredient for a fast recognition algorithm. We would like to stress that a promising line of research is to consider further tree-layout based graph classes. For this, following the work of Damaschke~\cite{Damaschke90}, Hell et al.~\cite{HellMR14} and Feuilloley and Habib~\cite{FeuilloleyH21} on layouts, we need to investigate in a more systematic way various patterns to exclude, including rooted tree patterns. 


\bibliographystyle{plainurl}

\end{document}